\documentclass[12pt,a4paper]{article}
\usepackage{amssymb}
\usepackage{graphicx}
\usepackage{amsmath}
\usepackage{bbm}
\usepackage[sort&compress, square, numbers]{natbib}
\newtheorem{theorem}{Theorem}

\newtheorem{lemma}[theorem]{Lemma}

\newtheorem{remark}[theorem]{Remark}

\newenvironment{proof}[1][Proof]{\noindent\textbf{#1.} }{\ \rule{0.5em}{0.5em}}

\begin{document}

\title{ Removing Gaussian Noise by Optimization of
Weights in Non-Local Means}
\author{Qiyu Jin$^{a,b}$ \and Ion Grama$^{a,b}$ \and Quansheng Liu$^{a,b}$ \\
\\
{\small {{qiyu.jin@univ-ubs.fr} \ \ \ \ \ \ \ {ion.grama@univ-ubs.fr}}}\\
{\small {{quansheng.liu@univ-ubs.fr} } }\\
{\small {$^{a}$Universit\'{e} de Bretagne-Sud, Campus de Tohaninic, BP 573,
56017 Vannes, France }}\\
{\small {\ $^{b}$Universit\'{e} Europ\'{e}enne de Bretagne, France } }}
\date{}
\maketitle

\begin{abstract}
A new image denoising algorithm to deal with the additive Gaussian white
noise model is given. Like the non-local means method, the filter is based
on the weighted average of the observations in a neighborhood, with weights
depending on the similarity of local patches. But in contrast to the
non-local means filter, instead of using a fixed Gaussian kernel, we propose
to choose the weights by minimizing a tight upper bound of mean square error. This approach makes it possible to define the weights adapted to the function at
hand, mimicking the weights of the oracle filter. Under some regularity
conditions on the target image, we show that the obtained estimator
converges at the usual optimal rate. The proposed algorithm is parameter free in the sense that it automatically calculates the bandwidth of the smoothing kernel; it is fast and its
implementation is straightforward. The performance of the new filter is
illustrated by numerical simulations.
\end{abstract}

\medskip \vspace{.5em} {\textit{Keywords}:\,\relax} {\ Non-local means,
image denoising, optimization weights, oracle, statistical estimation.}

\section{Introduction}
 We deal with the additive
Gaussian noise model%
\begin{equation}
Y(x)=f(x)+\varepsilon (x),\;x\in \mathbf{I,}  \label{s1Y1}
\end{equation}%
where $\mathbf{I}$ is a uniform $N\times N$ grid of pixels on the unit square, $Y=\left( Y\left( x\right) \right) _{x\in \mathbf{I}}$ is the observed
image brightness, $f:[0,1]^{2}\rightarrow \mathbf{R}_{+}$ is an unknown target
regression function and $\varepsilon =\left( \varepsilon \left( x\right)
\right) _{x\in \mathbf{I}}$ are independent and identically
distributed (i.i.d.) Gaussian random variables with mean $0$ and standard
deviation $\sigma >0.$ Important denoising techniques for the model (\ref%
{s1Y1}) have been developed in recent years, see for example Buades,
Coll and Morel (2005 \citep{Bu}), Kervrann (2006 \citep{kervrann2006optimal}), Lou,
Zhang, Osher and Bertozzi (2010 \citep{lou2010image}), Polzehl and Spokoiny
(2006 \citep{polzehl2006propagation}), Garnett, Huegerich and Chui (2005 \citep{Garnett2005universal}%
), Cai, Chan, Nikolova (2008 \citep{cai2008two}),  Katkovnik,  Foi,  Egiazarian, and  Astola (
2010 \citep{Katkovnik2010local}), Dabov,  Foi,  Katkovnik and
Egiazarian (2006 \citep{buades2009note}). A
significant step in these developments was the introduction of the Non-Local
Means filter by Buades, Coll and Morel \citep{Bu} and its variants (see e.g.
\citep{kervrann2006optimal}, \citep{kervrann2008local}, \citep{lou2010image}).\ In these filters,
the basic idea is to estimate the unknown image $f(x_{0})$ by a weighted
average of the form
\begin{equation}
\widetilde{f}_{w}(x_{0})=\sum_{x\in \mathbf{I}}w(x)Y(x),  \label{s1fx}
\end{equation}%
where $w=\left( w\left( x\right) \right) _{x\in \mathbf{I}}$ are some
non-negative weights satisfying $\sum_{x\in \mathbf{I}}w(x)=1.$ The choice of
the weights $w$ are based essentially on two criteria: a local criterion so
that the weights are as a decreasing function of the distance to the estimated
pixel, and a non-local criterion which gives more important weights to the
pixels whose brightness is close to the brightness of the estimated pixel (see e.g. Yaroslavsky (1985 \citep{yaroslavsky1985digital}) and Tomasi and Manduchi (1998 \citep{tomasi1998bilateral})). The
non-local approach has been further completed by a fruitful idea which
consists in attaching small regions, called data patches, to each pixel and
comparing these data patches instead of the pixels themselves.

The methods based on the non-local criterion consist of a comparatively novel
direction which is less studied in the literature. In this paper we shall
address two  problems related to this criterion.

The first problem is
how to choose data depending on weights $w$ in (\ref{s1fx}) in some optimal
way. Generally, the weights $w$ are defined through some priory fixed kernels,
often the Gaussian one, and the important problem of the choice of the
kernel has not been addressed so far for the non-local approach. Although
the choice of the Gaussian kernel seems to show reasonable numerical
performance, there is no particular reason to restrict ourselves only to
this type of kernel. Our theoretical results and the accompanying
simulations show that another kernel should be preferred. In addition to
this, for the obtained optimal kernel we shall also be interested in deriving a locally adaptive rule for the bandwidth choice.
The second  problem that we shall address is the convergence of the
obtained filter to the true image. Insights can be found in \citep{Bu}, \citep%
{kervrann2006optimal}, \citep{kervrann2008local} and  \citep{li2010new}, however the problem of convergence of
the Non-Local Means Filter has not been completely settled so far. In this
paper, we shall give some new elements of the proof of the convergence of the
constructed filter, thereby giving a theoretical justification of the
proposed approach from the asymptotic point of view.

Our main idea is to produce a very tight upper bound of the mean square error
\begin{equation*}
R\left( \widetilde{f}_{w}(x_{0})\right) =\mathbb{E}\left( \widetilde{f}%
_{w}(x_{0})-f(x_{0})\right) ^{2}
\end{equation*}%
in terms of the bias and variance and to minimize this upper bound in $w$
under the constraints $w\geq 0$ and $\sum_{x\in \mathbf{I}}w(x)=1.$
In contrast to the usual approach where a specific class of target functions is considered,  here we give a bound of the bias depending only on the target function $f$ at hand, instead of using just a bound expressed in terms of the parameters of the class.
We first
obtain an explicit formula for the optimal weights $w^{\ast }$ in terms of
the unknown function $f.$ In order to get a computable filter, we estimate $%
w^{\ast }$ by some adaptive weights $\widehat{w}$ based on data patches from
the observed image $Y.$ We thus obtain a new filter, which we call \textit{%
Optimal Weights} Filter. To justify theoretically our filter,\ we prove that
it achieves the optimal rate of convergence under some regularity conditions
on $f.$ Numerical results show that Optimal Weights Filter outperforms the
typical Non-Local Means Filter, thus giving a practical justification that
the optimal choice of the kernel improves the quality of the denoising,
while all other conditions are the same.

We would like to   point out that related optimization problems for non
parametric signal and density recovering have been proposed earlier in Sacks
and Ylvysaker (1978 \citep{Sacks1978linear}), Roll (2003 \citep{Roll2003local}), Roll and Ljung
(2004 \citep{Roll2004}), Roll, Nazin and Ljung (2005 \citep{roll2005nonlinear}), Nazin,
Roll, Ljung and Grama (2008 \citep{Nazin2008direct}). In these papers the weights
are optimized over a given class of regular functions and thus depend only
on some parameters of the class. This approach corresponds to the minimax setting, where the resulting minimax estimator has the best rate
of convergence corresponding to the worst image in the given class of
images. If the image happens to have better regularity than the
worst one, the minimax estimator will exhibit a slower rate of convergence
than expected. The novelty of our work is to find the optimal weights
depending on the image $f$ at hand, which implicates that our
Optimal Weights Filter automatically attains the optimal rate of convergence
for each particular image $f.$ Results of this type are related to the "oracle"  concept developed in Donoho and Johnstone (1994 \citep%
{donoho1994ideal}).

Filters with data-dependent weights have been previously studied in many
papers, among which we mention Polzehl and Spokoiny (2000 \citep{polzehl2000adaptive}%
, 2003 \citep{polzehl2003image}, 2006 \citep{polzehl2006propagation}), Kervrann (2006 \citep%
{kervrann2006optimal} and 2007 \citep{kervrann2010bayesian}). Compared with these filters our
algorithm is straightforward to implement and gives a quality of denoising
which is close to that of the best recent methods (see Table \ref{Table
compar}). The weight optimization approach can also be applied with these
algorithms to improve them. In particular, we can use it with recent versions
of the Non-Local Means Filter, like the BM3D (see 2006 \citep{buades2009note}, 2007 \citep{Dabov2007color,Dabov2007image});
however this is beyond the scope of the present paper and will be done
elsewhere.

The paper is organized as follows. Our new filter based on  the optimization of weights in the introduction in Section \ref{Sec:constr} where we present the main idea and the algorithm.  Our main theoretical results are
presented in Section \ref{Sec:main} where we give the rate of convergence of
the constructed estimators. In Section \ref{Sec:simulations}, we present our
simulation results with a brief analysis. Proofs of the main results are
deferred to Section \ref{Sec:Appendix Proofs}.

To conclude this section, let us set some important notations to be used
throughout the paper. The Euclidean norm of a vector $x=\left(
x_{1},...,x_{d}\right) \in \mathbf{R}^{d}$ is denoted by $%
\left\Vert x\right\Vert _{2}=\left( \sum_{i=1}^{d}x_{i}^{2}\right) .$ The
supremum norm of $x$ is denoted by $\Vert x\Vert _{\infty }=\sup_{1\leq
i\leq d}\left\vert x_{i}\right\vert .$ The cardinality of a set $\mathbf{A}$ is denoted $\text{card}\, \mathbf{A}$. For a positive integer $N$ the uniform $N\times N$-grid of pixels on the unit square is defined
by
\begin{equation}
\mathbf{I}=\left\{ \frac{1}{N},\frac{2}{N},\cdots ,\frac{N-1}{N},1\right\}
^{2}.  \label{def I}
\end{equation}%
Each element $x$ of the grid $\mathbf{I}$ will be called pixel. The number of pixels
is $n=N^{2}.$ For any pixel $x_{0}\in \mathbf{I}$ and a given $h>0,$ the
square window of pixels
\begin{equation}
\mathbf{U}_{x_{0},h}=\left\{ x\in \mathbf{I:\;}\Vert x-x_{0}\Vert _{\infty
}\leq h\right\}  \label{def search window}
\end{equation}%
will be called \emph{search window} at $x_{0}.$ We naturally take $h$ as a multiple of $\frac{1}{N}$ ($ h=\frac{k}{N}$ for some $k\in \{ 1, 2,\cdots,N\}$). The size of the square
search window $\mathbf{U}_{x_{0},h}$ is the positive integer number
$$%
M=nh^{2}=\mathrm{card\ }\mathbf{U}_{x_{0},h}.
$$
 For any pixel $x\in \mathbf{U}%
_{x_{0},h}$ and a given $\eta >0$ a second square window of pixels
\begin{equation}
\mathbf{V}_{x,\eta }=\left\{ y\in \mathbf{I:\;}\Vert y-x\Vert _{\infty }\leq
\eta \right\}  \label{def patch}
\end{equation}%
will be called for short a \emph{patch window} at $x$ in order to be
distinguished from the search window $\mathbf{U}_{x_{0},h}.$ Like $h$, the parameter $\eta$ is also taken as a multiple of $\frac{1}{N}$. The size of the
patch window $\mathbf{V}_{x,\eta }$ is the positive integer
$$m=n\eta
^{2}=\mathrm{card\ }\mathbf{V}_{x_{0},\eta }.
$$
 The vector $\mathbf{Y}%
_{x,\eta }=\left( Y\left( y\right) \right) _{y\in \mathbf{V}_{x,\eta }}$
formed by the the values of the observed noisy image $Y$ at pixels in the
patch $\mathbf{V}_{x,\eta }$ will be called simply \emph{data patch} at $%
x\in \mathbf{U}_{x_{0},h}.$ Finally, the positive part of  a real number $%
a$ is denoted by $a^{+},$ that is%
\begin{equation*}
a^{+}=\left\{
\begin{array}{cc}
a & \text{if }a\geq 0, \\
0 & \text{if }a<0.%
\end{array}%
\right.
\end{equation*}

\section{\label{Sec:constr}Construction of the estimator}

Let $h>0$ be fixed. For any pixel $x_{0}\in \mathbf{I}$ consider a family of
weighted estimates $\widetilde{f}_{h,w}(x_{0})$ of the form
\begin{equation}
\widetilde{f}_{h,w}(x_{0})=\sum_{x\in \mathbf{U}_{x_{0},h}}w(x)Y(x),
\label{s2fx}
\end{equation}%
where the unknown weights satisfy
\begin{equation}
w(x)\geq 0\;\;\;\text{and\ \ \ }\sum_{x\in \mathbf{U}_{x_{0},h}}w(x)=1.
\label{s2wx}
\end{equation}%
The usual bias plus variance decomposition of the mean square error
gives%
\begin{equation}
\mathbb{E}\left( \widetilde{f}_{h,w}(x_{0})-f(x_{0})\right)
^{2}=Bias^{2}+Var,  \label{s2ef}
\end{equation}%
with%
\begin{equation*}
Bias^{2}=\left( \sum_{x\in \mathbf{U}_{x_{0},h}}w(x)\left(
f(x)-f(x_{0})\right) \right) ^{2}\;\;\text{and\ \ }Var=\sigma ^{2}\sum_{x\in
\mathbf{U}_{x_{0},h}}w(x)^{2}.
\end{equation*}%
The decomposition (\ref{s2ef}) is commonly used to construct asymptotically
minimax estimators over some given classes of functions in the nonparametric
function estimation. In order to highlight the difference between the
approach proposed in the present paper and the previous work, suppose that $f$
belongs to the class of functions satisfying the H\"{o}lder condition $%
|f(x)-f(y)|\leq L\Vert x-y\Vert _{\infty }^{\beta },\,\,\,\forall x,\,y\in
\mathbf{I}.$ In this case, it is easy to see that%
\begin{equation}
\mathbb{E}\left( \widetilde{f}_{h,w}(x_{0})-f(x_{0})\right) ^{2}\leq \left(
\sum_{x\in \mathbf{U}_{x_{0},h}}w(x)L\left\vert x-x_{0}\right\vert ^{\beta
}\right) ^{2}+\sigma ^{2}\sum_{x\in \mathbf{U}_{x_{0},h}}w(x)^{2}.
\label{expect E}
\end{equation}%
Optimizing further the weights $w$ in the obtained upper bound gives an
asymptotically minimax estimate with weights depending on the unknown
parameters $L$ and $\beta $ (for details see  \citep{Sacks1978linear}).
With our approach the bias term $Bias^{2}$ will be bounded in terms of the
unknown function $f$ itself. As a result we obtain some "oracle" weights $w$
adapted to the unknown function $f$ at hand, which will be estimated
further using data patches from the image $Y.$

First, we shall address the problem of determining the "oracle" weights. With
this aim denote%
\begin{equation}
\rho _{f,x_{0}}\left( x\right) \equiv \left\vert f(x)-f(x_{0})\right\vert .
\label{s2rx}
\end{equation}%
Note that the value $\rho _{f,x_{0}}\left( x\right) $ characterizes the
variation of the image brightness of the pixel $x$ with respect to the pixel
$x_{0}.$ From the decomposition (\ref{s2ef}), we easily obtain a tight upper
bound in terms of the vector $\rho _{f,x_{0}}:$%
\begin{equation}
\mathbb{E}\left( \widetilde{f}_{h}(x_{0})-f(x_{0})\right) ^{2}\leq g_{\rho
_{f,x_{0}}}(w),  \label{s2ef-vers2}
\end{equation}%
where
\begin{equation}
g_{\rho _{f,x_{0}}}(w)=\left( \sum_{x\in \mathbf{U}_{x_{0},h}}w(x)\rho
_{f,x_{0}}(x)\right) ^{2}+\sigma ^{2}\sum_{x\in \mathbf{U}%
_{x_{0},h}}w(x)^{2}.  \label{s2gw}
\end{equation}

From the following theorem we can obtain the form of the weights $w$ which
minimize the function $g_{\rho _{f,x_{0}}}(w)$ under the constraints (\ref%
{s2wx}) in terms of the values $\rho _{f,x_{0}}\left( x\right) .$ For the
sake of generality, we shall formulate the result for an arbitrary
non-negative function $\rho \left( x\right) ,$ $x\in \mathbf{U}_{x_{0},h}.$
Define the objective function
\begin{equation}
g_{\rho }(w)=\left( \sum_{x\in \mathbf{U}_{x_{0},h}}w(x)\rho (x)\right)
^{2}+\sigma ^{2}\sum_{x\in \mathbf{U}_{x_{0},h}}w(x)^{2}.
\label{def gw}
\end{equation}%
Introduce into consideration the strictly increasing function%
\begin{equation}
M_{\rho }\left( t\right) =\sum_{x\in \mathbf{U}_{x_{0},h}}\rho (x)(t-\rho
(x))^{+},\ \ \ t\geq 0.
\label{def mt}
\end{equation}%
Let $K_{\text{tr}}$ be the usual triangular kernel:
\begin{equation}
K_{\text{tr}}\left( t\right) =\left( 1-\left\vert t\right\vert \right)
^{+},\quad t\in \mathbb{R}^{1}.
\label{def kernel}
\end{equation}

\begin{theorem}
\label{Th weights 001}Assume that $\rho \left( x\right) ,$ $x\in \mathbf{U}%
_{x_{0},h}$, is a  non-negative function. Then the
unique weights which minimize $g_{\rho }(w)$ subject to (\ref{s2wx}) are
given by
\begin{equation}
w_{\rho }(x)=\frac{K_{\text{tr}}(\frac{\rho (x)}{a})}{\sum\limits_{y\in
\mathbf{U}_{x_{0},h}}K_{\text{tr}}(\frac{\rho (x)}{a})},\quad x\in \mathbf{U}%
_{x_{0},h},  \label{eq th weights 001}
\end{equation}%
where the bandwidth $a>0$ is the unique solution on $(0,\infty)$ of the equation
\begin{equation}
M_{\rho }\left( a \right) =\sigma ^{2}.  \label{eq th weights 002}
\end{equation}
\end{theorem}\par
Theorem \ref{Th weights 001} can be obtained from a result of Sacks and Ylvysaker \citep{Sacks1978linear}. The proof is deferred to Section  \ref{Sec: proof of Th weights 001}.
\begin{remark}\label{calculate a}
The value of $a>0$ can be calculated as follows. We sort the set $\{\rho(x)\, | \, x \in \mathbf{U}_{x_0,h}\}$  in the ascending order $0=\rho_{1}\leq\rho_{2}\leq\cdots\leq\rho_{M}<\rho_{M+1}=+\infty$, where $M=Card \, \mathbf{U}_{x_0,h}$.
Let
\begin{equation}
a_k=\frac{\sigma^2+\sum\limits_{i=1}^{k}\rho_i^2}
{\sum\limits_{i=1}^{k}\rho_i}, \quad 1\leq k\leq M,
\label{a k}
\end{equation}
and
\begin{equation}
k^*=\max\{1\leq k \leq M\, |\, a_k\geq \rho_{k} \}=\min\{1\leq k \leq M \, | \, a_{k}< \rho_{k} \}-1,
\label{k star}
\end{equation}
with the convention that $a_k=\infty$ if  $\rho_k=0$ and that $\min \varnothing =M+1$.
Then  the solution $a>0$ of (\ref{eq th weights 002}) can be expressed as $a=a_{k^*}$; moreover, $k^*$ is the unique integer $k\in \{1,\cdots, M\}$ such that $a_k\geq \rho_k$ and $a_{k+1}<\rho_{k+1}$ if $k<M$.
\end{remark}
\par
The proof of the remark is deferred to Section  \ref{Sec: proof of remark}.
\par
Let $x_{0}\in \mathbf{I.}$ Using the optimal weights given by Theorem \ref{Th
weights 001}, we first introduce the following non computable approximation
of the true image, called "oracle":
\begin{equation}
f_{h}^{\ast }(x_{0})=\frac{\sum_{x\in \mathbf{U}_{x_{0},h}}K_{\text{tr}}(%
\frac{\rho _{f,x_{0}}\left( x\right) }{a})Y(x)}{\sum\limits_{y\in \mathbf{U}%
_{x_{0},h}}K_{\text{tr}}(\frac{\rho _{f,x_{0}}\left( x\right) }{a})},
\label{oracle}
\end{equation}%
where the bandwidth $a$ is the solution of the equation $M_{\rho
_{f,x_{0}}}\left( a\right) =\sigma ^{2}.$ A computable filter can be
obtained by estimating the unknown function $\rho _{f,x_{0}}\left( x\right) $
and the bandwidth $a$ from the data as follows.

Let $h>0$ and $\eta >0$ be fixed numbers. For any $x_{0}\in \mathbf{I}$ and
any $x\in \mathbf{U}_{x_{0},h}$ consider a distance between the data patches
$\mathbf{Y}_{x,\eta }=\left( Y\left( y\right) \right) _{y\in \mathbf{V}%
_{x,\eta }}$ and $\mathbf{Y}_{x_{0},\eta }=\left( Y\left( y\right) \right)
_{y\in \mathbf{V}_{x_{0},\eta }}$ defined by
\begin{equation*}
d^{2}\left( \mathbf{Y}_{x,\eta },\mathbf{Y}_{x_{0},\eta }\right) =\frac{1}{m}%
\left\Vert \mathbf{Y}_{x,\eta }-\mathbf{Y}_{x_{0},\eta }\right\Vert _{2}^{2},
\end{equation*}%
where $m=\mathrm{card\ }\mathbf{V}_{x,\eta }$, and $\left\Vert \mathbf{Y}_{x,\eta }-\mathbf{Y}_{x_{0},\eta }\right\Vert _{2}^{2}=\sum\limits_{x_0+z\in \mathbf{V}_{x_0,\eta}} \left( Y(x+z)-Y(x_0+z)\right)^2$. Since Buades, Coll and Morel
\citep{Bu} the distance $d^{2}\left( \mathbf{Y}_{x,\eta },\mathbf{Y}%
_{x_{0},\eta }\right) $ is known to be a flexible tool to measure the
variations of the brightness of the image $Y.$
As
\begin{equation*}
Y(x+z)-Y(x_0+z)=f(x+z)-f(x_0+z)+\epsilon(x+z)-\epsilon(x_0+z)
\end{equation*}
we have
\begin{equation*}
\mathbb{E}(Y(x+z)-Y(x_0+z))^2=(f(x+z)-f(x_0+z))^2+2\sigma^2.
\end{equation*}
If we use the approximation
\begin{equation*}
(f(x+z)-f(x_0+z))^2\approx (f(x)-f(x_0))^2=\rho^2_{f,x_0}(x)
\end{equation*}
and  the law of large numbers, it seems reasonable that
\begin{equation*}
\rho^2_{f,x_0}(x)\approx d^2(\mathbf{Y}_{x,\eta},\mathbf{Y}_{x_0,\eta})-2\sigma^2.
\end{equation*}
But our simulations show that a much better approximation is
\begin{equation}
\rho_{f,x_0}(x) \approx \widehat{\rho}_{x_0}(x)=\left(d(\mathbf{Y}_{x,\eta},\mathbf{Y}_{x_0,\eta})
-\sqrt{2}\sigma\right)^+.
\label{rho-estim}
\end{equation}
The fact that $\widehat{\rho}_{x_0}(x)$ is a good estimator of $\rho_{f,x_0}$ will be justified by  convergence theorems: cf. Theorems \ref{Th adapt 001} and \ref{Th adapt 002} of Section \ref{Sec:main}.
Thus our Optimal Weights Filter
is defined by%
\begin{equation}
\widehat{f}(x_{0})=\widehat{f}_{h,\eta }(x_{0})=\frac{\sum_{x\in \mathbf{U}%
_{x_{0},h}}K_{\text{tr}}(\frac{\widehat{\rho }_{x_{0}}\left( x\right) }{%
\widehat{a}})Y(x)}{\sum\limits_{y\in \mathbf{U}_{x_{0},h}}K_{\text{tr}}(%
\frac{\widehat{\rho }_{x_{0}}\left( x\right) }{\widehat{a}})},
\label{OWFilter}
\end{equation}%
where the bandwidth $\widehat{a}>0$ is the solution of the equation $M_{%
\widehat{\rho }_{x_{0}}}\left( \widehat{a}\right) =\sigma ^{2}$, which can be calculated as in Remark \ref{calculate a} (with $\rho(x)$ and $a$ replaced by $\widehat{\rho}_{x_0}(x)$ and $\widehat{a}$ respectively). We end this section
by giving an algorithm for computing the filter (\ref{OWFilter}). The input
values of the algorithm are the image $Y\left( x\right) ,$ $x\in \mathbf{I}$
, the variance of the noise $\sigma $ and two numbers $m$ and $M$
representing  the sizes of the patch window and  the search window
respectively.

\bigskip

\textbf{Algorithm :}\quad Optimal Weights Filter

Repeat for each $x_0\in \mathbf{I}$

\quad \quad give an initial value of $\widehat {a}$: $\widehat{a}=1$ (it can be an arbitrary positive number).

\quad \quad compute $\{\widehat{\rho} _{x_0}(x)\,|\,  x\in \mathbf{U}_{x_0,h}\}$ by (\ref{rho-estim})

/\emph{compute the bandwidth }$\widehat{a}$\emph{\ at }$x_{0}$

\quad \quad reorder $\{\widehat{\rho }_{x_0}(x)\,|\,  x\in \mathbf{U}_{x_0,h}\}$ as increasing sequence, say

\quad \quad\quad \quad  $\widehat{\rho}_{x_0}(x_1)\leq\widehat{\rho}_{x_0}(x_2)\leq \cdots \leq \widehat{\rho}_{x_0}(x_M)$ \ \ \ \

\quad \quad loop from $k=1$ to $M$

\quad \quad \quad \quad if $\sum_{i=1}^{k}\widehat{\rho }_{x_0}(x_{i})>0$

\quad \quad \quad \quad \quad \quad if $\frac{\sigma ^{2}+\sum_{i=1}^{k}%
\widehat{\rho }_{x_0}^{2}(x_{i})}{\sum_{i=1}^{k}\widehat{\rho }_{x_0}(x_{i})}\geq
\widehat{\rho }(x_{k})$ then $\widehat{a}=\frac{\sigma ^{2}+\sum_{i=1}^{k}%
\widehat{\rho }_{x_0}^{2}(x_{i})}{\sum_{i=1}^{k}\widehat{\rho}_{x_0}(x_i)}$

\quad \quad \quad \quad \quad \quad else quit loop

\quad \quad \quad \quad else continue loop

\quad \quad end loop

/\emph{compute the estimated weights }$\widehat{w}$\emph{\ at }$x_{0}$

\quad \quad compute $\widehat{w}(x_{i})=\frac{K_{\text{tr}}(1-\widehat{\rho }_{x_0}%
(x_{i})/\widehat{a})^{+}}{\sum_{x_{i}\in \mathbf{U}_{x_{0},h}}K_{\text{tr}%
}(1-\widehat{\rho }_{x_0}(x_{i})/\widehat{a})^{+}}$

/\emph{compute the filter }$\widehat{f}$\emph{\ at }$x_{0}$

\quad \quad compute $\widehat{f}(x_{0})=\sum_{x_{i}\in \mathbf{U}_{x_{0},h}}%
\widehat{w}(x_{i})Y(x_{i})$.

\bigskip

\noindent \bigskip

\par
The proposed algorithm is computationally fast and its implementation is
straightforward compared to more sophisticated algorithms developed in recent years. Notice that an important issue in the non-local means filter is the choice of the bandwidth parameter in the Gaussian kernel; our algorithm is parameter free in the sense that it automatically chooses the bandwidth.
\par The numerical simulations show that our filter outperforms the
classical non-local means filter under the same conditions. The overall
performance of the proposed filter compared to its simplicity is very good
which can be a big advantage in some practical applications. We hope that
optimal weights that we deduced can be useful with more complicated
algorithms and can give similar improvements of the denoising quality.
However, these investigations are beyond the scope of the present paper. A
detailed analysis of the performance of our filter is given in Section \ref%
{Sec:simulations}.

\section{\label{Sec:main}Main results}

In this section, we  present two theoretical results.
\par
The first result is a
mathematical justification of the "oracle" filter introduced in the previous
section. It shows that despite the fact that we minimized an upper bound of
the mean square error instead of the mean square error itself, the
obtained "oracle" still has the optimal rate of convergence. Moreover, we
show that the weights optimization approach possesses the following important
adaptivity property: our procedure automatically chooses the correct
bandwidth $a>0$ even if the radius $h>0$ of the search window $U_{x_{0},h}$ is
larger than necessary.

The second result shows the convergence of the Optimal Weights Filter $%
\widehat{f}_{h,\eta}$ under some more restricted conditions than those formulated
in Section \ref{Sec:constr}. To prove the convergence, we split the image
into two independent parts. From the first one, we construct  the
"oracle" filter; from the second one, we estimate the weights. Under
some regularity assumptions on the target image we are able to show that the
resulting filter has nearly the optimal rate of convergence.

Let $\rho \left( x\right) ,$ $x\in \mathbf{U}_{x_{0},h},$ be an arbitrary
non-negative function and let $w_{\rho }$ be the optimal weights given by (\ref{eq th weights
001}). Using these weights $w_{\rho }$
we define the family of estimates%
\begin{equation}
f_{h}^{\ast }(x_{0})=\sum_{x\in \mathbf{U}_{x_{0},h}}w_{\rho }(x)Y(x)
\label{s2fx3}
\end{equation}%
depending on the unknown function $\rho .$ The next theorem shows that one
can pick up a useful estimate from the family $f_{h}^{\ast }$ if the the
function $\rho $ is close to the "true" function $\rho _{f,x_{0}}(x)=\left\vert
f\left( x\right) -f\left( x_{0}\right) \right\vert ,$ i.e. if
\begin{equation}
\rho \left( x\right) =\left\vert f\left( x\right) -f\left( x_{0}\right)
\right\vert +\delta _{n},  \label{s2fx3a}
\end{equation}%
where $\delta _{n}\geq 0$ is a small deterministic error. We shall prove the
convergence of the estimate $f_{h}^{\ast }$ under the local H\"{o}lder
condition
\begin{equation}
|f(x)-f(y)|\leq L\Vert x-y\Vert _{\infty }^{\beta },\,\,\,\forall x,\,y\in
\mathbf{U}_{x_{0},h},  \label{Local Holder cond}
\end{equation}%
where $\beta >0$ is a constant, $h>0,$ and $x_{0}\in \mathbf{I}.$
\par
In the
following, $c_{i}>0$ $(i\geq 1)$ denotes a positive constant, and $O(a_n)$  $(n\geq 1)$ denotes a number bounded by $c\cdot a_n$ for some constant $c>0$. All the constants $c_i>0$ and $c>0$ depend only on $L$, $\beta$ and $\sigma$; their values can be different from line to line.
\begin{theorem}
\label{Th oracle 001} Assume that $h=c_{1}n^{-\frac{1}{2\beta +2}}$ with $c_{1}>c_{0}=\left( %
\frac{\sigma ^{2}\left( \beta +2\right) \left( 2\beta +2\right) }{8L^{2}\beta }\right)
^{\frac{1}{2\beta +2}}$, or  $h\geq
c_{1}n^{-\alpha }$ with $0\leq \alpha <\frac{1}{2\beta +2}$ and $c_{1}>0$.  Suppose that
 $f$ satisfies the local H\"{o}%
lder's condition (\ref{Local Holder cond}) and that $\delta _{n}=O\left( n^{-%
\frac{\beta }{2+2\beta }}\right) .$ Then
\begin{equation}
\mathbb{E}\left( f_{h}^{\ast }(x_{0})-f(x_{0})\right) ^{2}=O\left( n^{-\frac{%
2\beta }{2+2\beta }}\right) .  \label{s2ef2}
\end{equation}%
\end{theorem}
\par
 The proof will be given in Section  \ref{Sec: proof of Th oracle 001}.
 \par
Recall that the bandwidth $h$ of order $n^{-\frac{1}{2+2\beta }}$ is
required to have the optimal minimax rate of convergence $O\left( n^{-\frac{2\beta }{%
2+2\beta }}\right) $ of the mean squared error for estimating the function $f
$ of global H\"{o}lder smoothness $\beta $  (cf. e.g.
 \cite{FanGijbels1996}). To better understand the adaptivity property of
the oracle $f_{h}^{\ast }(x_{0}),$ assume that the image $f$ at $x_{0}$ has
H\"{o}lder smoothness $\beta $ (see \citep{Wh}) and that $h\geq c_{0}n^{-\alpha }$ with $0\leq
\alpha <\frac{1}{2\beta +2},$ which means that the radius $h>0$ of the search
window $U_{x_{0},h}$ has been chosen larger than the "standard" $n^{-\frac{1%
}{2\beta +2}}.$ Then, by Theorem \ref{Th oracle 001}, the rate of convergence
of the oracle is still of order $n^{-\frac{\beta }{2+2\beta }},$ contrary to
the global case mentioned above. If we choose a sufficiently large
search window $U_{x_{0},h},$ then the oracle $f_{h}^{\ast }(x_{0})$ will
have a rate of convergence which depends only on the unknown maximal local
smoothness $\beta $ of the image $f.$ In particular, if $\beta $ is very
large, then the rate will be close to $n^{-1/2},$ which ensures good
estimation of the flat regions in cases where the regions are indeed flat.
More generally, since Theorem \ref{Th oracle 001} is valid for arbitrary $%
\beta ,$ it applies for the maximal local H\"{o}lder smoothness $\beta _{x_{0}}$
at $x_{0},$ therefore the oracle $f_{h}^{\ast }(x_{0})$ will exhibit the
best rate of convergence of order $n^{-\frac{2\beta _{x_{0}}}{2+2\beta
_{x_{0}}}}$ at $x_{0}.$ In other words, the procedure adapts to the best
rate of convergence at each point $x_{0}$ of the image.

We justify by simulation results that the difference
between the oracle $f_{h}^{\ast }$ computed with $\rho =\rho
_{f,x_{0}}=\left\vert f\left( x\right) -f\left( x_{0}\right) \right\vert ,$
and the true image $f$, is extremely small (see Table\ \ref{Table oracle}).
This shows that, at least from the practical point of view, it is justified
to optimize the upper bound $g_{\rho _{f,x_{0}}}(w)$ instead of optimizing
the mean square error $\mathbb{E}\left( f^{\ast }_h(x_{0})-f(x_{0})\right) ^{2}$ itself.

The estimate $f_{h}^*$ with the choice $\rho \left(
x\right) =\rho_{f,x_0}\left( x\right) $ will be called oracle filter. In
particular for the oracle filter $f_{h}^{\ast },$ under the conditions of
Theorem \ref{Th oracle 001}, we have%
\begin{equation*}
\mathbb{E}\left( f_{h}^{\ast }(x_{0})-f(x_{0})\right) ^{2}\leq g_{\rho}\left(
w_{\rho}\right) \leq cn^{-\frac{2\beta }{2+2\beta }}.
\end{equation*}

Now, we turn to the study of the convergence of the Optimal Weights Filter. Due to the difficulty in dealing with the dependence of the weights we shall consider a slightly modified version of the proposed algorithm:  we divide  the set of pixels  into two independent parts,  so that the weights are constructed from the one part,  and  the estimation of the target function is a weighted mean  along  the other part. More precisely,  assume  that $x_{0}\in \mathbf{I},$ $h>0$ and $\eta >0.$ To prove the
convergence we split the set of pixels into two parts $\mathbf{I}=\mathbf{I}%
_{x_{0}}^{\prime }\cup \mathbf{I}_{x_{0}}^{\prime \prime },$ where
\begin{equation}
\mathbf{I}_{x_{0}}^{\prime }=\left\{ x_{0}+\left( \frac{i}{N},\frac{j}{N}%
\right) \in \mathbf{I}:i+j\text{ is even }\right\}
\label{set I1}
\end{equation}%
is the set of pixels with an even sum of coordinates $i+j$ and $\mathbf{I}%
_{x_{0}}^{\prime \prime }=\mathbf{I}\diagdown \mathbf{I}_{x_{0}}^{\prime }.$
Denote $\mathbf{U}_{x_{0},h}^{\prime }=\mathbf{U}_{x_{0},h}\cap \mathbf{I}%
_{x_{0}}^{\prime }$ and $\mathbf{V}_{x,\eta }^{\prime \prime }=\mathbf{V}%
_{x,\eta }\cap \mathbf{I}_{x_{0}}^{\prime \prime }.$ Consider the distance
between the data patches $\mathbf{Y}_{x,\eta }^{\prime \prime }=\left(
Y\left( y\right) \right) _{y\in \mathbf{V}_{x,\eta }^{\prime \prime }}$ and $%
\mathbf{Y}_{x_{0},\eta }^{\prime \prime }=\left( Y\left( y\right) \right)
_{y\in \mathbf{V}_{x_{0},\eta }^{\prime \prime }}$ defined by
\begin{equation*}
d\left( \mathbf{Y}_{x,\eta }^{\prime \prime },\mathbf{Y}_{x_{0},\eta
}^{\prime \prime }\right) =\frac{1}{\sqrt{m^{\prime \prime }}}\left\Vert
\mathbf{Y}_{x,\eta }^{\prime \prime }-\mathbf{Y}_{x_{0},\eta }^{\prime
\prime }\right\Vert _{2},
\end{equation*}%
where $m^{\prime \prime }=\mathrm{card\ }\mathbf{V}_{x,\eta }^{\prime \prime
}.$ An estimate of the function $\rho _{f,x_{0}}$ is given by
\begin{equation}
\rho _{f,x_{0}}\left( x\right) \approx \widehat{\rho }_{x_{0}}^{\prime
\prime }\left( x\right) =\left( d\left( \mathbf{Y}_{x,\eta }^{\prime \prime
},\mathbf{Y}_{x_{0},\eta }^{\prime \prime }\right) -\sqrt{2}\sigma \right)
^{+},  \label{rho-estim prim}
\end{equation}%
see (\ref{rho-estim}).
Define the filter $\widehat{f}_{h,\eta }^{\prime }$ by
\begin{equation}
\widehat{f}_{h,\eta }^{\prime }(x_{0})=\sum_{x\in \mathbf{U}%
_{x_{0},h}^{\prime }}\widehat{w}^{\prime \prime }(x)Y(x),
\label{OWFilter prim}
\end{equation}%
where
\begin{equation}
\widehat{w}^{\prime \prime }=\arg \min_{w}\left( \sum_{x\in \mathbf{U}%
_{x_{0},h}^{\prime }}w(x)\widehat{\rho }^{\prime \prime }_{x_0}(x)\right)
^{2}+\sigma ^{2}\sum_{x\in \mathbf{U}_{x_{0},h}^{\prime }}w^{2}(x).
\label{s3ww}
\end{equation}

The next theorem gives a rate of convergence of the Optimal Weights Filter
if the parameters $h>0$ and $\eta >0$ are chosen properly
according to the local smoothness $\beta .$

\begin{theorem}
\label{Th adapt 001} Assume that $h=c_{1}n^{-\frac{1}{%
2\beta +2}}$ with $c_{1}>c_{0}=\left(\frac{\sigma ^{2}\left( \beta
+2\right) \left( 2\beta +2\right) }{8L^{2}\beta }\right) ^{\frac{1}{2\beta +2%
}}$, and that $\eta =c_{2}n^{-\frac{1}{2\beta +2}}.$ Suppose that  function $f$ satisfies the local H\"{o}%
lder condition (\ref{Local Holder cond}). Then
\begin{equation}
\mathbb{E}\left( \widehat{f}_{h,\eta }^{\prime }(x_{0})-f(x_{0})\right)
^{2}=O\left( n^{-\frac{2\beta }{2\beta +2}}\ln n\right) .  \label{s3ef}
\end{equation}
\end{theorem}

For the proof of this theorem see Section \ref{Sec: proof of Th adapt 001}.

 Theorem \ref{Th adapt 001} states that with  the proper
choices of the parameters $h$ and $\eta $, the mean square error of the
estimator $\widehat{f}^{\prime }_{h,\eta}(x_{0})$ converges nearly at the rate $O(n^{-%
\frac{2\beta }{2\beta +2}})$ which is the usual optimal rate of convergence
for a given H\"{o}lder smoothness $\beta >0$ (cf. e.g.
\citep{FanGijbels1996}).

Simulation results show that the adaptive bandwidth  $\widehat{a}$
provided by our algorithm depends essentially on the local
properties of the image and does not depend much on the radius $h$ of the
search window. These simulations, together with Theorem \ref{Th oracle
001}, suggest that the Optimal Weights Filter (\ref{OWFilter}) can also be applied with larger $h,$ as is the case of the "oracle" filter $%
f_{h}^{\ast }.$ The following theorem deals with the case where $h$ is large.

\begin{theorem}
\label{Th adapt 002} Assume that $h=c_{1}n^{-\alpha }$
with $c_{1}>0,$ and $0<\alpha \leq \frac{1}{2\beta +2}$ and that $\eta
=c_{2}n^{-\frac{1}{2\beta +2}}.$ Suppose that the function $f$ satisfies the local H\"{o}%
lder condition (\ref{Local Holder cond}). Then
\begin{equation*}
\mathbb{E}\left( \widehat{f}_{h,\eta }^{\prime }(x_{0})-f(x_{0})\right)
^{2}=O\left( n^{-\frac{\beta }{2\beta +2}}\ln n\right) .
\end{equation*}
\end{theorem}

For the proof of this theorem see Section \ref{Sec: proof of Th adapt 002}.
Note that in this case the obtained rate of convergence is not the usual
optimal one, in contrast to Theorems \ref{Th oracle 001} and \ref{Th adapt
001}, but we believe that this is the best rate that can be obtained for the
proposed filter.

\section{\label{Sec:simulations}Numerical performance of the Optimal Weights
Filter}

The performance of the Optimal Weights Filter $\widehat{f}_{h,\eta }(x_{0})$
is measured by the usual Peak Signal-to-Noise Ratio (PSNR) in decibels (db)
defined as%
\begin{equation*}
PSNR=10\log _{10}\frac{255^{2}}{MSE},\,\,\,MSE=\frac{1}{card\, \mathbf{I}}%
\sum\limits_{x\in \mathbf{I}}(f(x)-\widehat{f}_{h,\eta}(x))^{2},
\end{equation*}%
where $f$ is the original image, and $\widehat{f}$ the estimated one.

In the simulations, we sometimes shall use the smoothed version of the
estimate of brightness variation $d_{K}\left( \mathbf{Y}_{x,\eta },\mathbf{Y}%
_{x_{0},\eta }\right) $ instead of the non smoothed one $d\left( \mathbf{Y}%
_{x,\eta },\mathbf{Y}_{x_{0},\eta }\right) .$ It should be noted that for
the smoothed versions of the estimated brightness variation we can establish
similar convergence results. The smoothed estimate $d_{K}\left( \mathbf{Y}%
_{x,\eta },\mathbf{Y}_{x_{0},\eta }\right) $ is defined by%
\begin{equation*}
d_{K}\left( \mathbf{Y}_{x,\eta },\mathbf{Y}_{x_{0},\eta }\right) =\frac{%
\left\Vert K\left( y\right) \cdot \left( \mathbf{Y}_{x,\eta }-\mathbf{Y}%
_{x_{0},\eta }\right) \right\Vert _{2}}{\sqrt{\sum_{y^{\prime }\in \mathbf{V}%
_{x_{0},\eta }}K(y^{\prime })}},
\end{equation*}%
where $K$ are some weights defined on $\mathbf{V}_{x_{0},\eta }.$ The
corresponding estimate of brightness variation $\rho _{f,x_{0}}\left(
x\right) $ is given by%
\begin{equation}
\widehat{\rho }_{K,x_{0}}(x)= \left( {d_{K}\left( \mathbf{Y}%
_{x,\eta },\mathbf{Y}_{x_{0},\eta }\right) }-\sqrt{2}\sigma \right)^+ .
\label{empir simil func}
\end{equation}%
With the rectangular kernel%
\begin{equation}
K_{r}\left( y\right) =\left\{
\begin{array}{ll}
1, & y\in \mathbf{V}_{x_{0},\eta }, \\
0, & \text{otherwise,}%
\end{array}%
\right.   \label{rect kernel}
\end{equation}%
we obtain exactly the distance $d\left( \mathbf{Y}_{x,\eta },\mathbf{Y}%
_{x_{0},\eta }\right) $ and the filter described in Section \ref{Sec:constr}%
. Other smoothing kernels $K$ used in the simulations are the Gaussian
kernel
\begin{equation}
K_{g}(y)=\exp \left( -\frac{N^{2}\Vert y-x_{0}\Vert _{2}^{2}}{2h_g}\right) ,
\label{s4kg}
\end{equation}%
where $h_g$ is the bandwidth parameter and the following kernel%
\begin{equation}
K_{0}\left( y\right) =\left\{
\begin{array}{ll}
\sum_{k=N\left\Vert y-x_{0}\right\Vert _{\infty }}^{p}\frac{1}{(2k+1)^{2}}
&  \text{if } y\neq x_{0}, \\
\sum_{k=1}^{p}\frac{1}{(2k+1)^{2}} &  \text{if } y=x_{0},%
\end{array}%
\right.   \label{s4ky}
\end{equation}%
with the width of the similarity window $m=(2p+1)^{2}.$ The shape of these
two kernels are displayed in Figure\ \ref{Fig Kernels 1}.
\par
To avoid the undesirable border effects in our simulations, we mirror the
image outside the image limits, that is we extend the image outside the
image limits symmetrically with respect to the border. At the corners, the
image is extended symmetrically with respect to the corner pixels.

\begin{table}[tbp]
\begin{center}
\renewcommand{\arraystretch}{0.6} \vskip3mm {\fontsize{8pt}{\baselineskip}%
\selectfont
\begin{tabular}{c|ccccc}
\hline
Images & Lena & Barbara & Boat & House & Peppers \\
Sizes & $512 \times 512$ & $512 \times 512$ & $512 \times 512$ & $256 \times
256$ & $256 \times 256$ \\ \hline\hline
$\sigma /PSNR$ & 10/28.12db & 10/28.12db & 10/28.12db & 10/28.11db &
10/28.11db \\ \hline
$11 \times 11$ & 41.20db & 40.06db & 40.23db & 41.50db & 40.36db \\
$13 \times 13$ & 41.92db & 40.82db & 40.99db & 42.24db & 41.01db \\
$15 \times 15$ & 42.54db & 41.48db & 41.62db & 42.85db & 41.53db \\
$17 \times 17$ & 43.07db & 42.05db & 42.79db & 43.38db & 41.99db \\
\hline\hline
$\sigma /PSNR$ & 20/22.11db & 20/22.11db & 20/22.11db & 20/28.12db &
20/28.12db \\ \hline
$11 \times 11$ & 37.17db & 35.92db & 36.23db & 37.18db & 36.25db \\
$13 \times 13$ & 37.91db & 36.70db & 37.01db & 37.97db & 36.85db \\
$15 \times 15$ & 38.57db & 37.37db & 37.65db & 38.59db & 37.38db \\
$17 \times 17$ & 39.15db & 37.95db & 38.22db & 39.11db & 37.80db \\
\hline\hline
$\sigma /PSNR$ & 30/18.60db & 30/18.60db & 30/18.60db & 30/18.61db &
30/18.61db \\ \hline
$11 \times 11$ & 34.81db & 33.65db & 33.79db & 34.93db & 33.57db \\
$13 \times 13$ & 35.57db & 34.47db & 34.58db & 35.78db & 34.23db \\
$15 \times 15$ & 36.24db & 35.15db & 35.25db & 36.48db & 34.78db \\
$17 \times 17$ & 36.79db & 35.75db & 35.84db & 37.07db & 35.26db \\ \hline
\end{tabular}
} \vskip1mm
\end{center}
\caption{\small PSNR values when oracle estimator $f_{h}^{\ast }$ is applied with
different values of $M$.}
\label{Table oracle}
\end{table}

\begin{figure}[tbp]
\begin{center}
\includegraphics[width=0.49\linewidth]{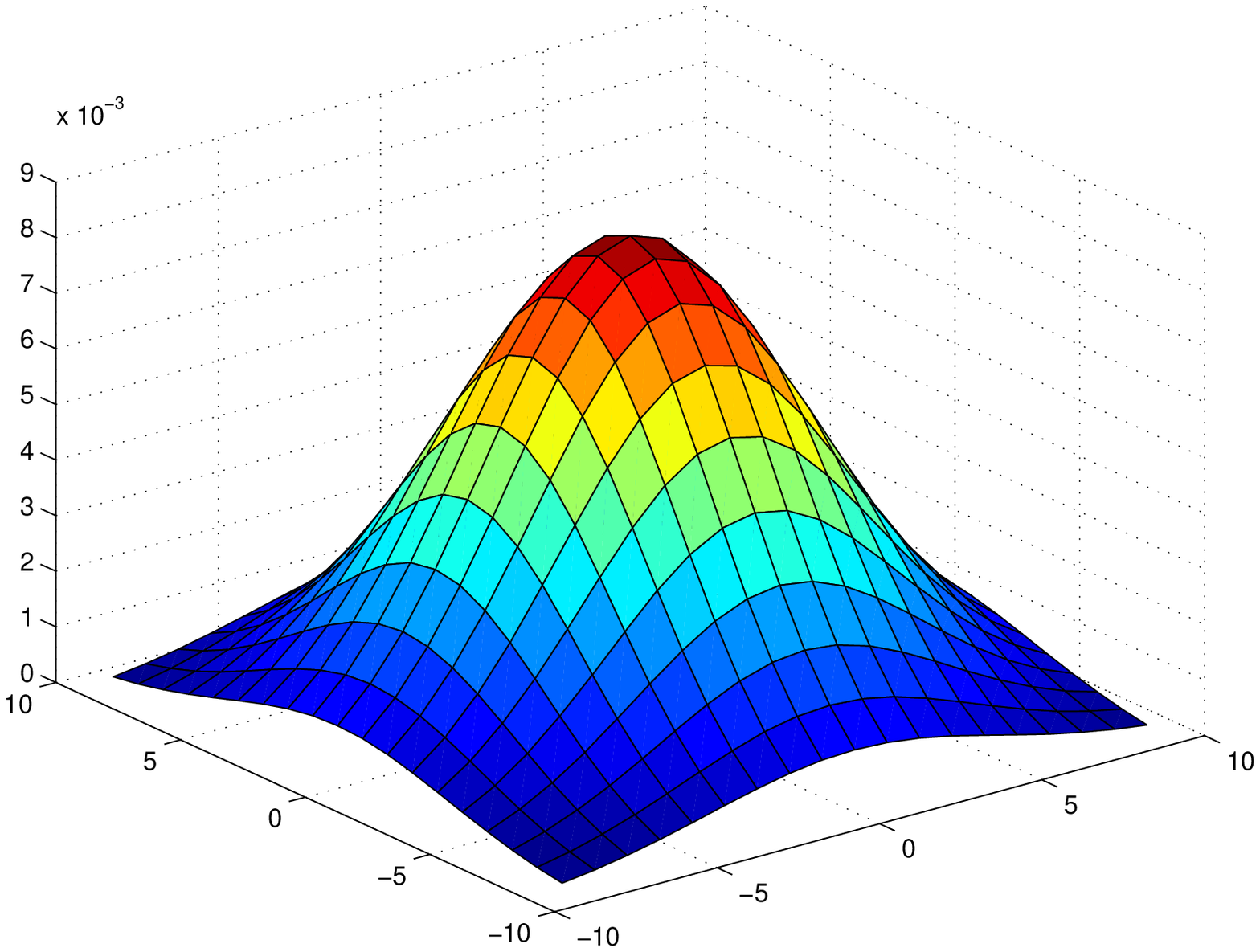}\ %
\includegraphics[width=0.49\linewidth]{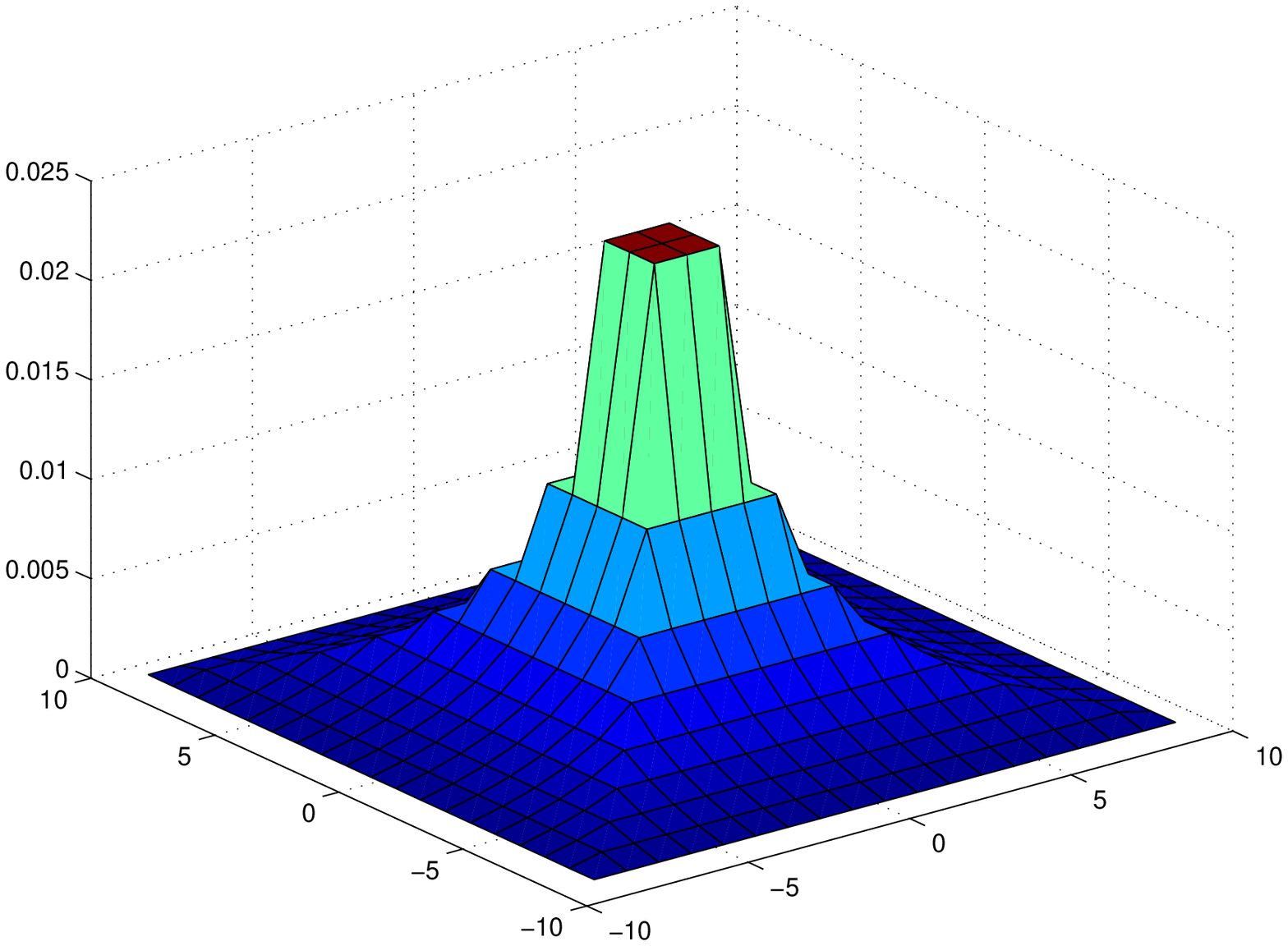}
\end{center}
\caption{\small The shape of the kernels $K_{g}$ (left) and $K_{0}$ (right) with $%
M= $ $21\times 21$.}
\label{Fig Kernels 1}
\end{figure}
{\small \ }
\par
We have done simulations on a commonly-used set of images available at
http://decsai.\newline
ugr.es/javier/denoise/test images/ which includes Lena, Barbara, Boat,
House, Peppers. The potential of the estimation method is illustrated with
the $512\times 512$ image "Lena" (Figure\ \ref{fig4}(a)) and "Barbara" (Figure \ \ref{fig barbara}(a) )  corrupted by an
additive white Gaussian noise (Figures\ \ref{fig4}(b), PSNR$=22.10db$, $%
\sigma =20$ and \ \ref{fig barbara} (b), PSNR$=18.60$, $\sigma=30$ ). We first used the rectangular kernel $K_{0}$ for computing the
estimated brightness variation function $\widehat{\rho }_{K,x_0},$ which
corresponds to the Optimal Weights Filter as defined in Section \ref%
{Sec:constr}. Empirically we found that the parameters $m$ and $M$ can be
fixed to $m=21\times 21$ and $M=13\times 13.$ In Figures\ \ref{fig4}(c) and \ \ref{fig barbara}(c), we
can see that the noise is reduced in a natural manner and significant
geometric features, fine textures, and original contrasts are visually well
recovered with no undesirable artifacts (PSNR$=32.52db$ for "Lena" and PSNR $=28.89$ for "Barbara"). To better
appreciate the accuracy of the restoration process, the square of the difference
between the original image and the recovered image is shown in Figures\ \ref%
{fig4}(d) and \ \ref{fig barbara}(d), where the dark values correspond to a high-confidence estimate.
As expected, pixels with a low level of confidence are located in the
neighborhood of image discontinuities. For comparison, we show the image
denoised by Non-Local Means Filter in Figures \ \ref{fig4}(e),(f) and \ \ref{fig barbara} (e), (f). The overall visual
impression and the numerical results are improved using our algorithm.

\begin{figure}[tbp]
\begin{center}
\renewcommand{\arraystretch}{0.5} \addtolength{\tabcolsep}{-6pt} \vskip3mm {%
\fontsize{8pt}{\baselineskip}\selectfont
\begin{tabular}{cc}
\includegraphics[width=0.42\linewidth]{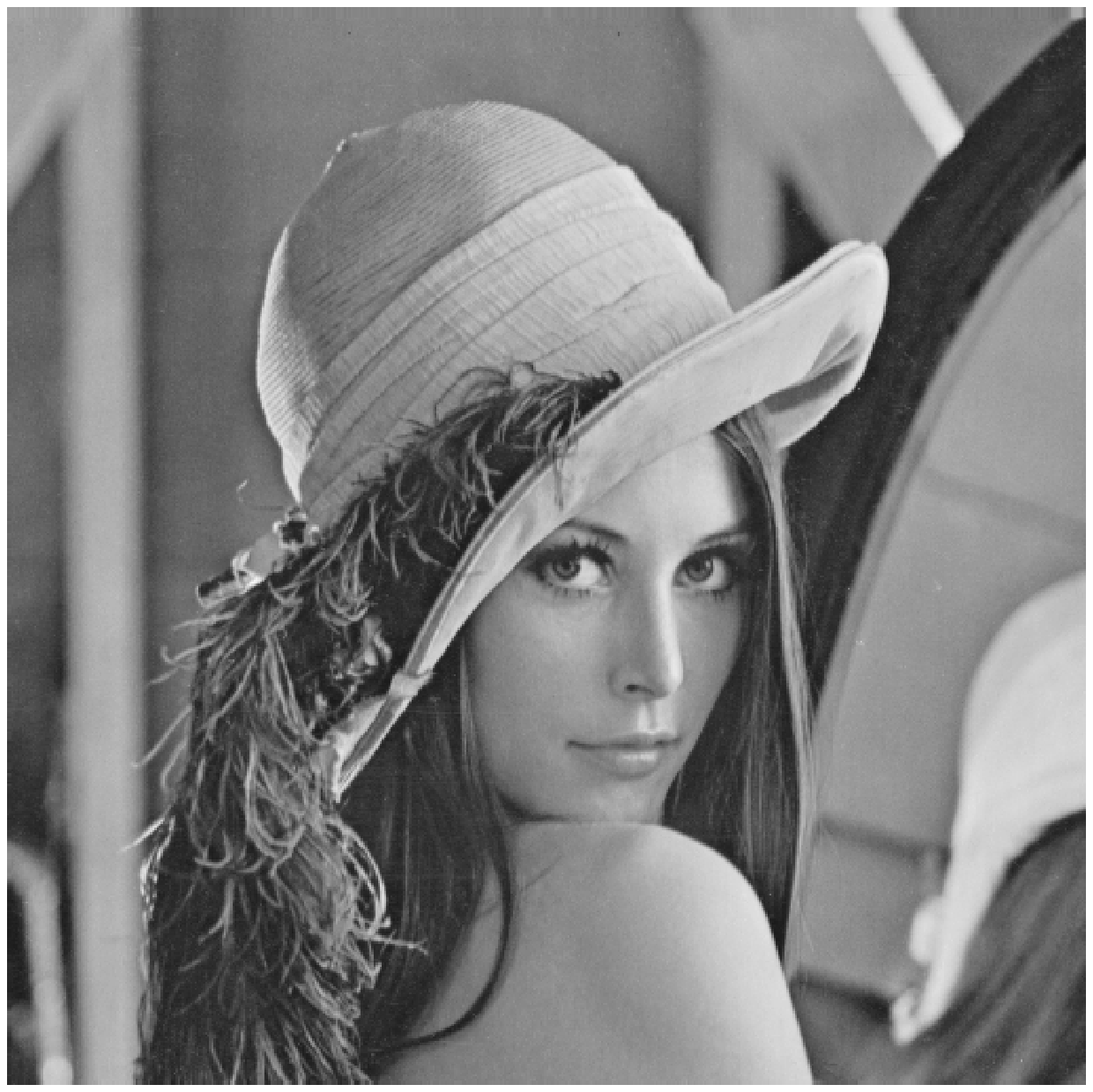}
& \includegraphics[width=0.42\linewidth]{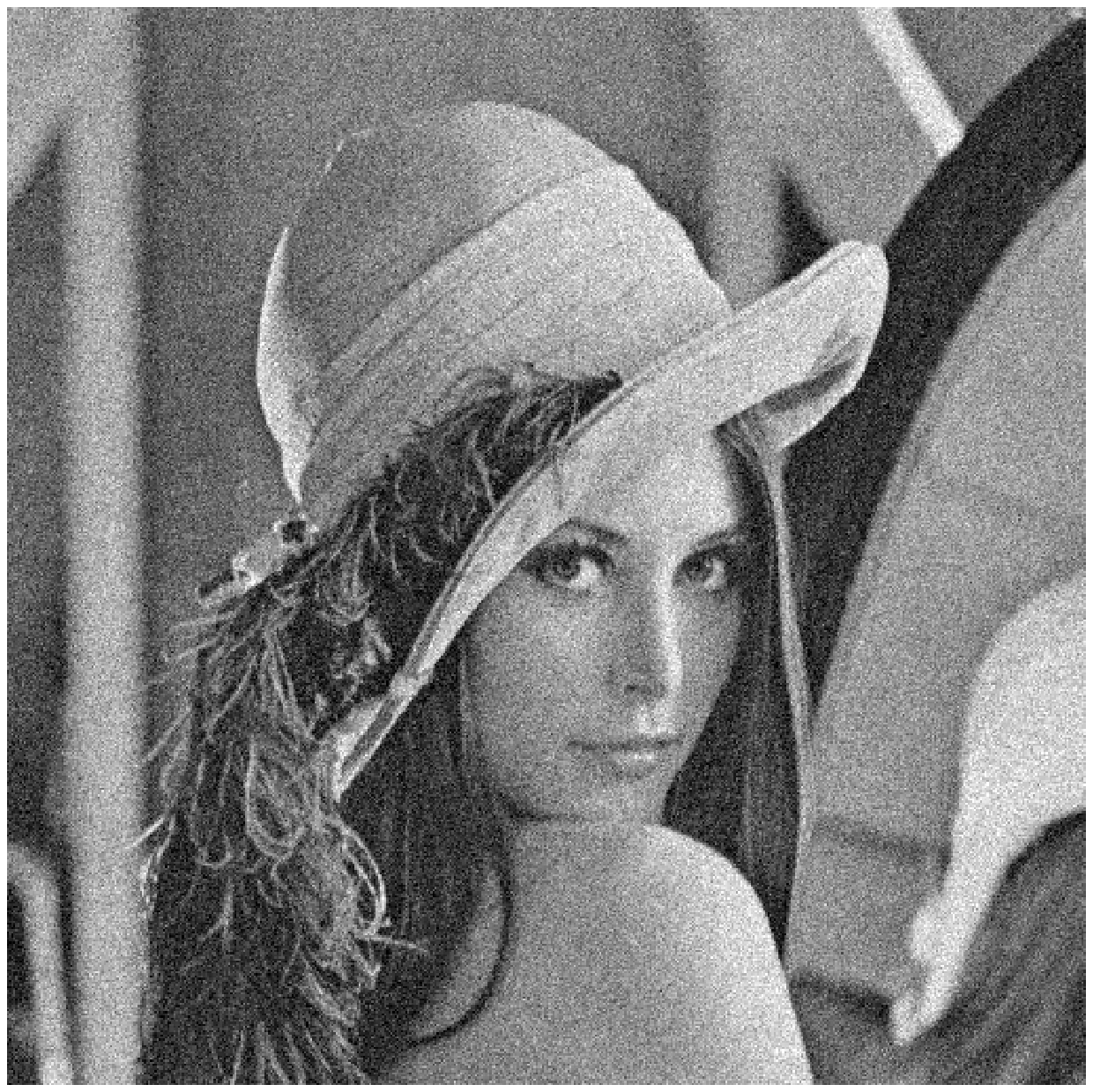} \\
(a) Original image "Lena"&
(b) Noisy image with $\sigma =20$, $PSNR=22.11db$\\
\includegraphics[width=0.42\linewidth]{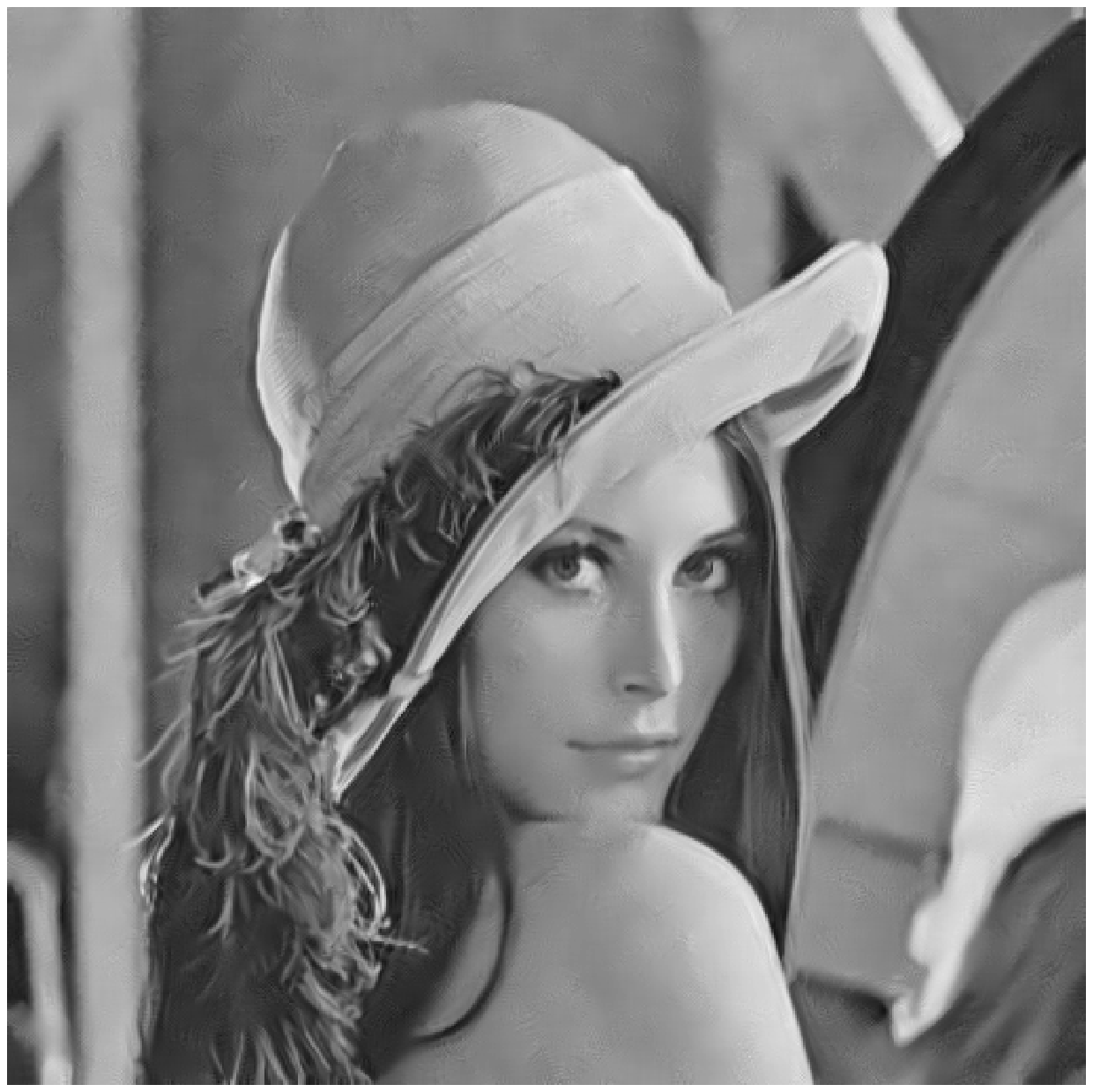} &
\includegraphics[width=0.42\linewidth]{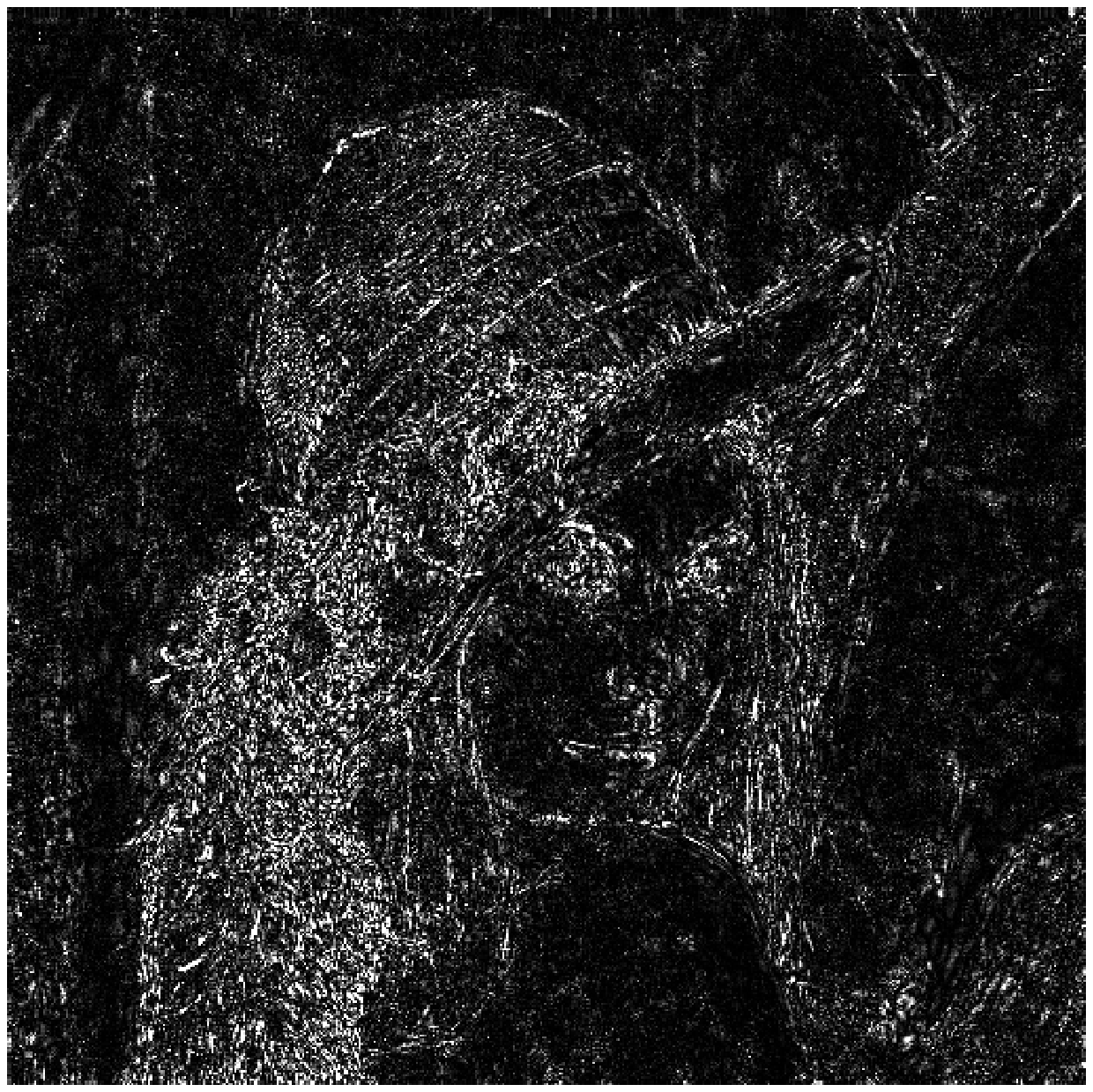} \\
(c) Restored   with OWF,  PSNR$=32.52db$  &
(d) Square error with OWF  \\
\includegraphics[width=0.42\linewidth]{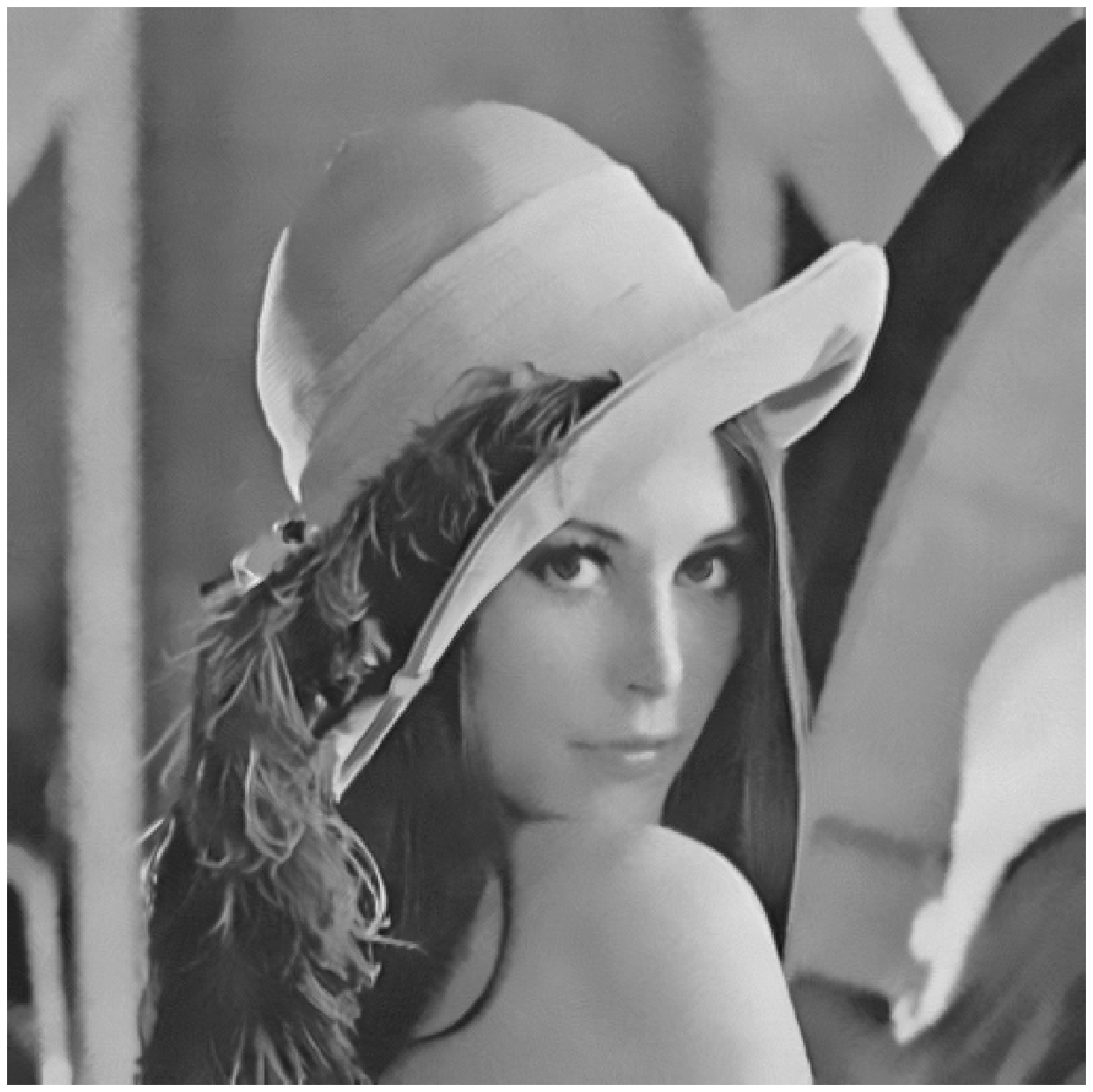} &
\includegraphics[width=0.42\linewidth]{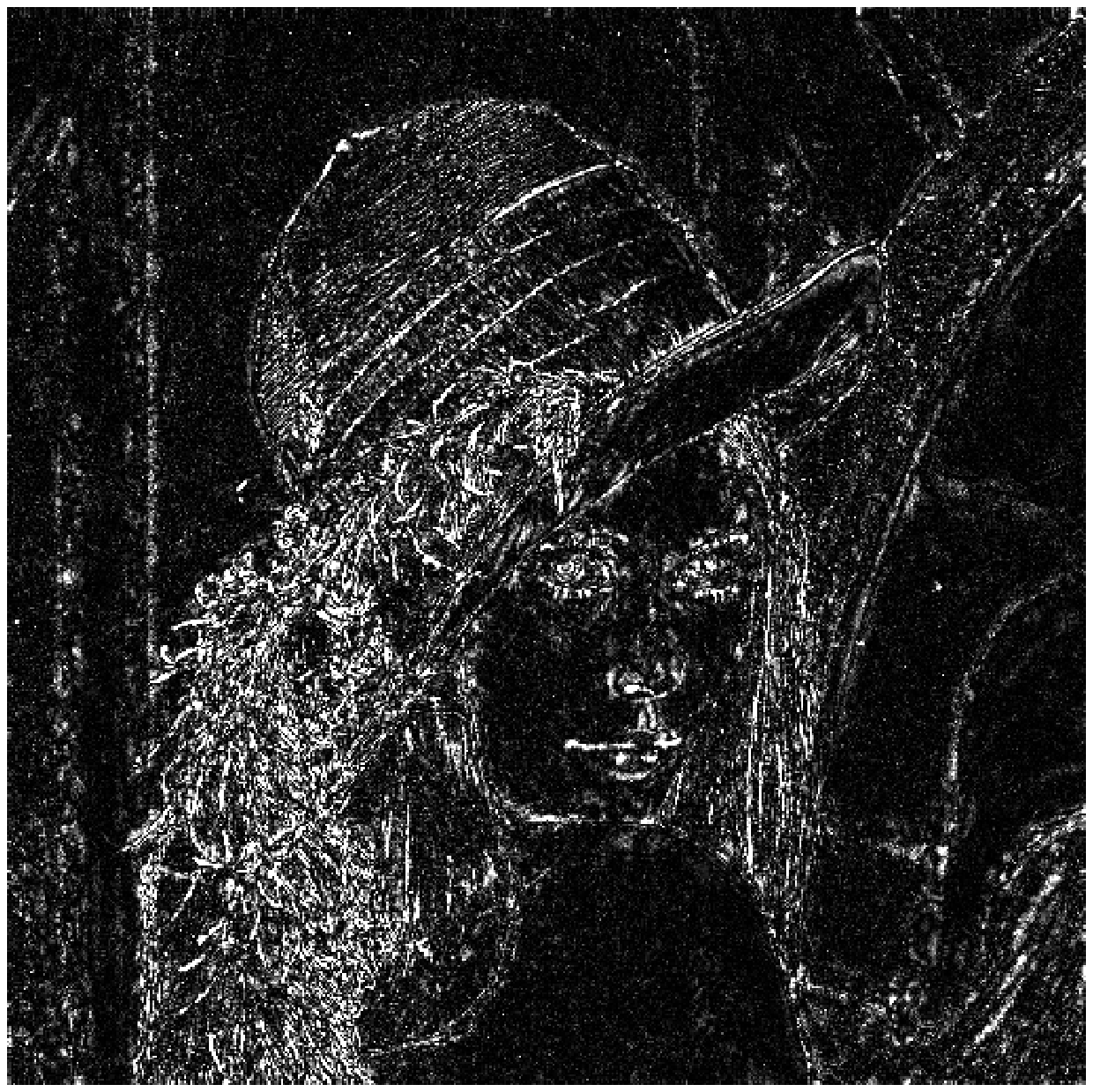}\\
(e) Restored  with NLMF,   PSNR$=31.73db$ &
(f) Square error with NLMF
\end{tabular}
} \vskip1mm
\par
\rule{0pt}{-0.2pt}%
\par
\vskip1mm 
\end{center}
\caption{\small Results of denoising "Lena" $512\times 512$ image. Comparing (d) and (f) we see that the Optimal Weights Filter (OWF) captures more details than the Non-Local Means Filter (NLMF).}
\label{fig4}
\end{figure}

\begin{figure}[tbp]
\begin{center}
\renewcommand{\arraystretch}{0.5} \addtolength{\tabcolsep}{-6pt} \vskip3mm {%
\fontsize{8pt}{\baselineskip}\selectfont
\begin{tabular}{cc}
\includegraphics[width=0.42\linewidth]{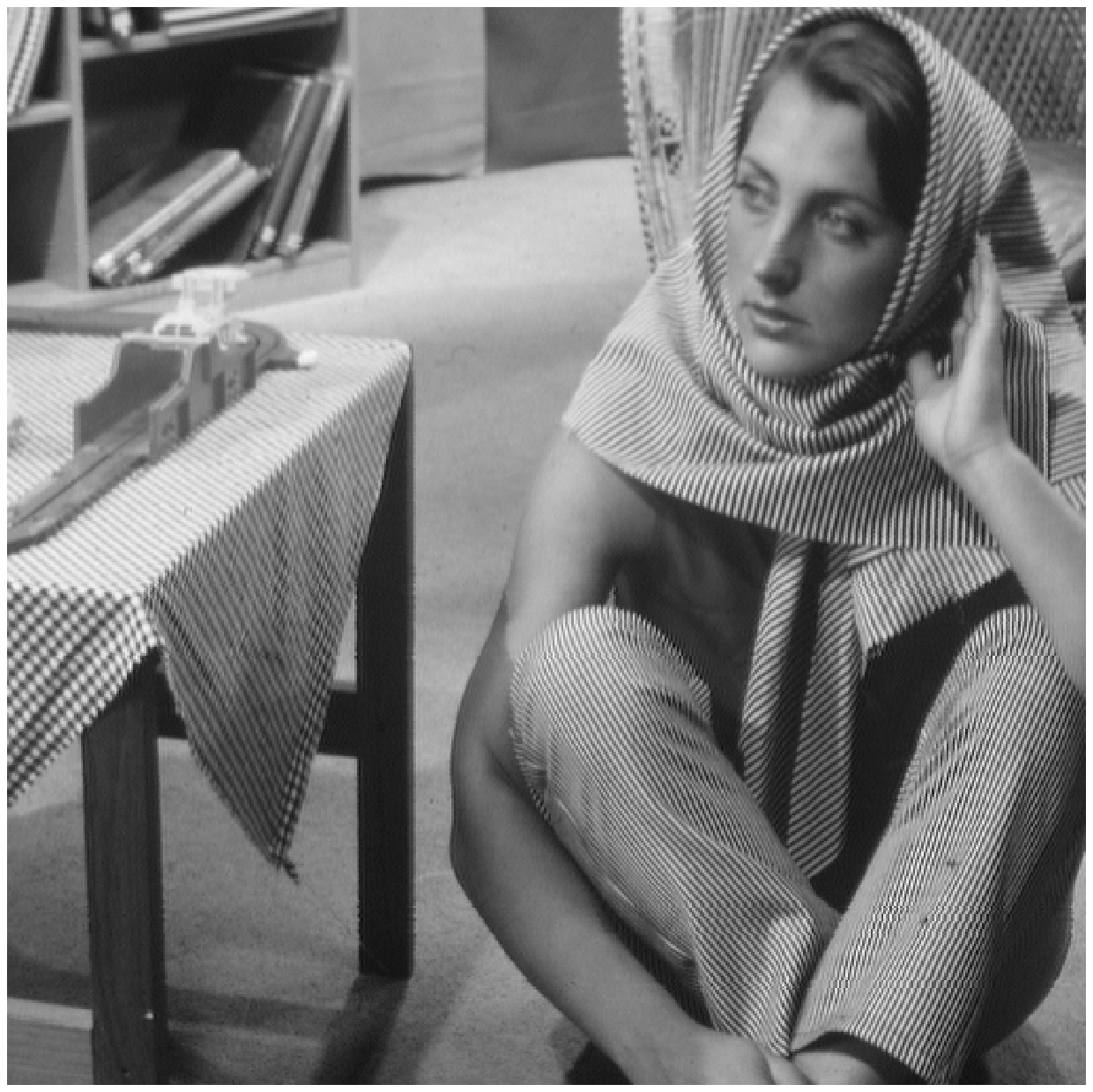}  &
 \includegraphics[width=0.42\linewidth]{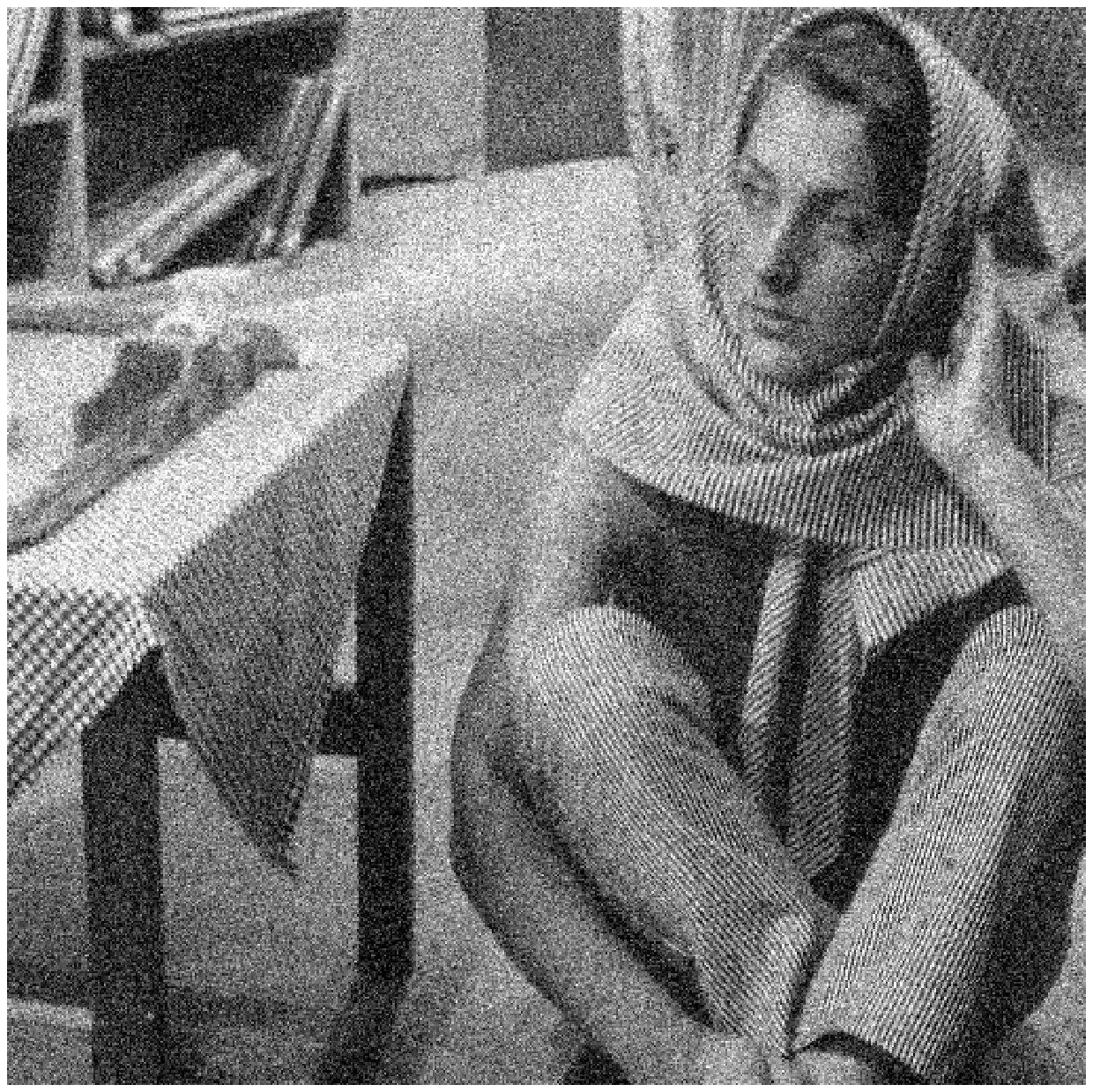} \\
(a) Original image "Barbara"&
(b) Noisy image with $\sigma =30$, $PSNR=18.60db$\\
\includegraphics[width=0.42\linewidth]{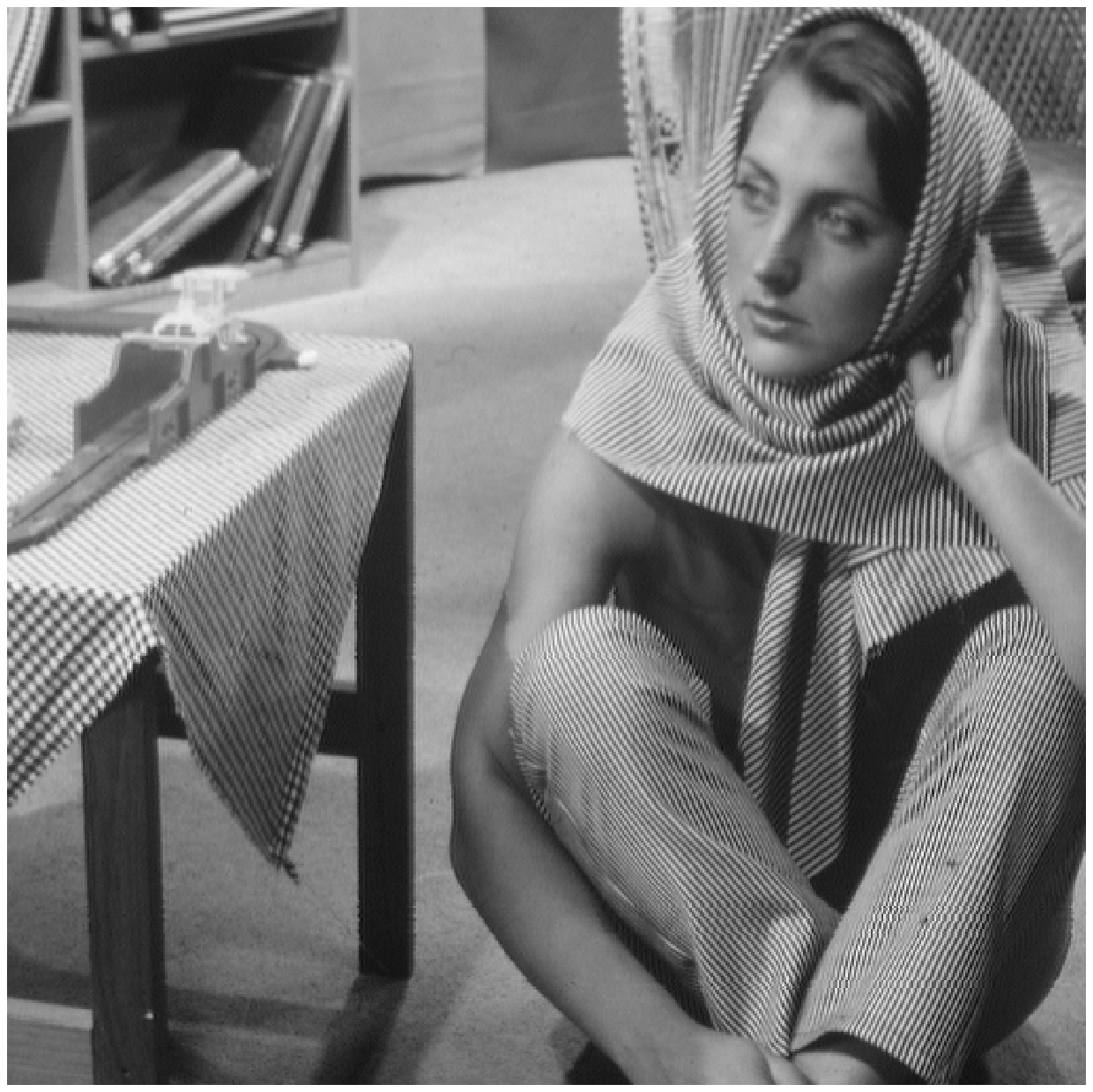} &
\includegraphics[width=0.42\linewidth]{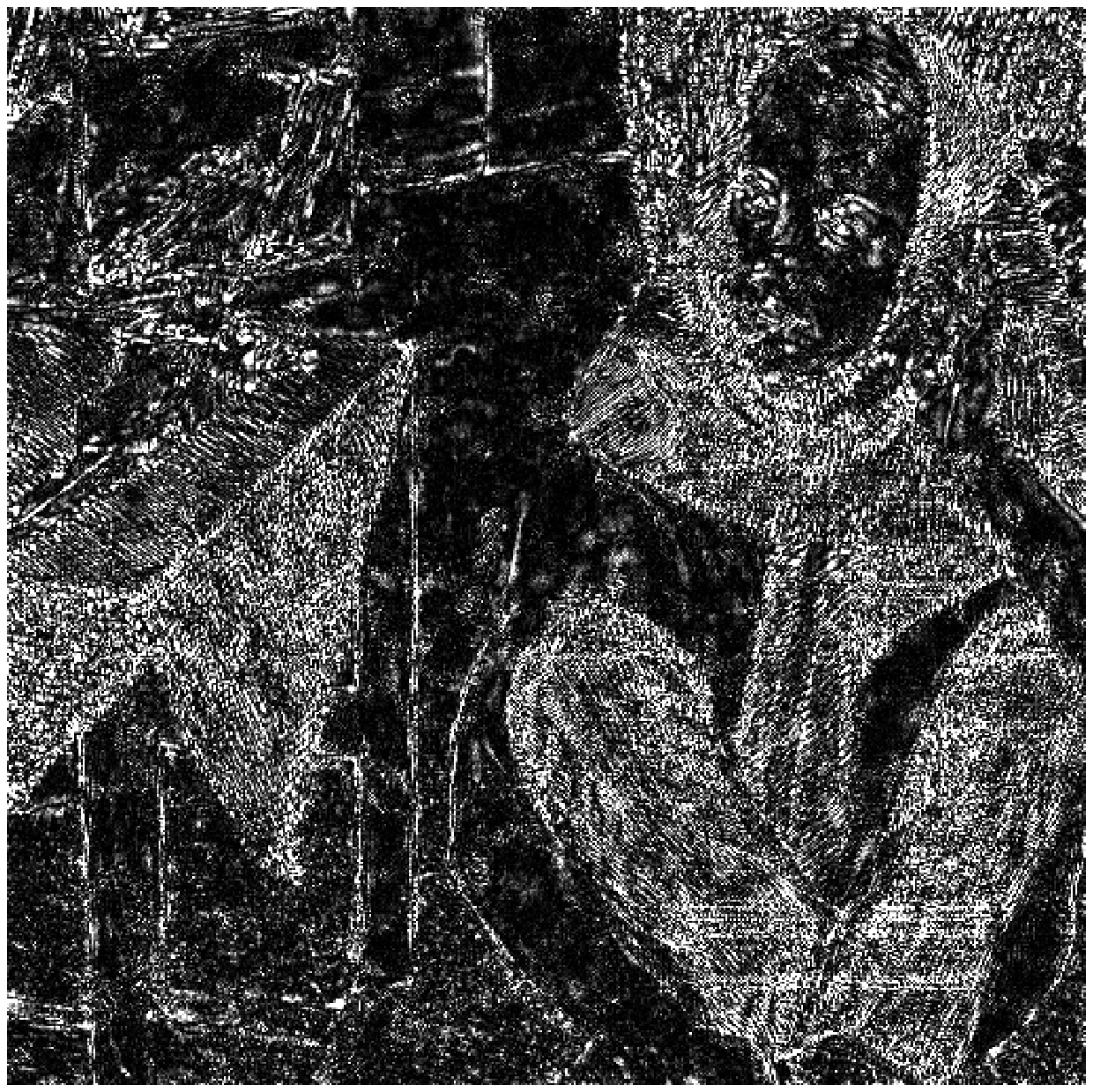} \\
(c) Restored with OWF, PSNR$=28.89db$  &
(d) Square error with OWF  \\
\includegraphics[width=0.42\linewidth]{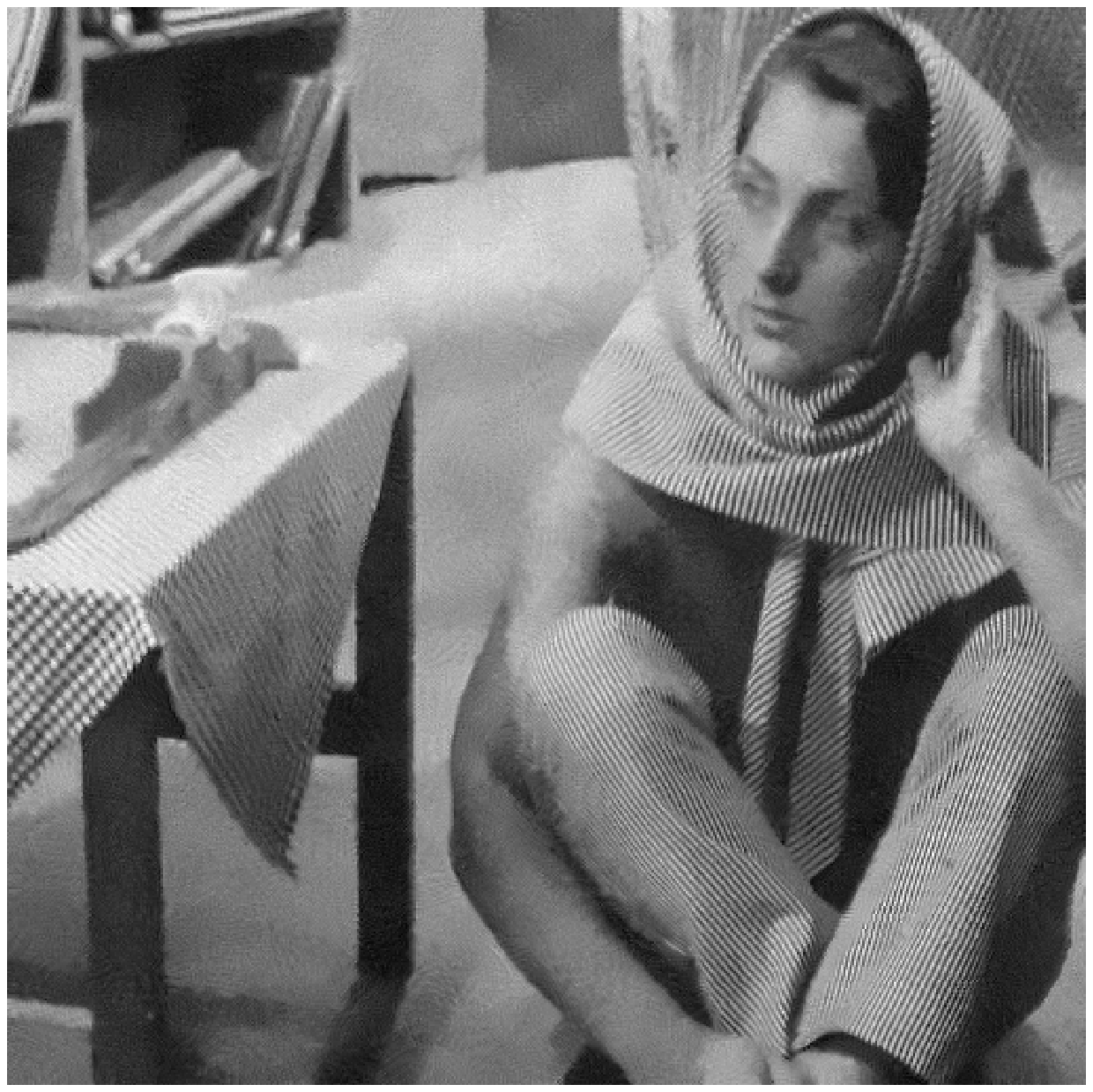} &
\includegraphics[width=0.42\linewidth]{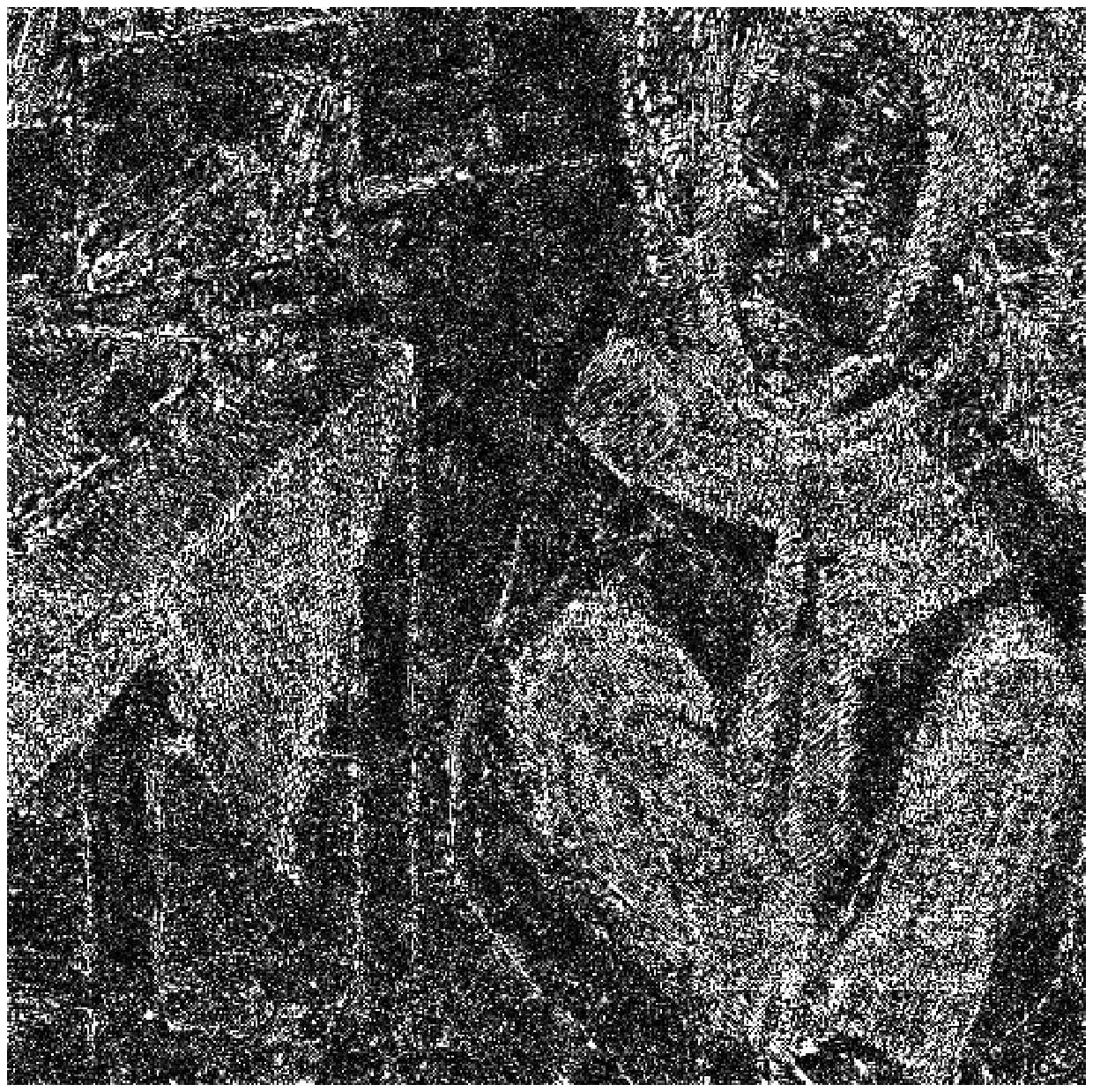}\\
(e) Restored   with NLMF,  PSNR$=27.88db$ &
(f) Square error with NLMF
\end{tabular}
} \vskip1mm
\par
\rule{0pt}{-0.2pt}%
\par
\vskip1mm 
\end{center}
\caption{\small Results of denoising "Barbara" $512\times 512$ image. Comparing (d) and (f) we see that the Optimal Weights Filter (OWF) captures more details than the Non-Local Means Filter (NLMF).}
\label{fig barbara}
\end{figure}

The Optimal Weights Filter seems to provide a feasible and rational method
to detect automatically the details of images and take the proper weights
for every possible geometric configuration of the image. For illustration purposes, we have chosen a series of search windows $\mathbf{U}_{x_{0},h}$ with centers at
some testing pixels $x_{0}$ on the noisy image, see Figure\ \ref{Fig selected}
The distribution of the weights inside the search window $\mathbf{U}%
_{x_{0},h}$ depends on the estimated brightness variation function $\widehat{%
\rho }_{K,x_{0}}\left( x\right) ,$ $x\in \mathbf{U}_{x_{0},h}.$ If the
estimated brightness variation $\widehat{\rho }_{K,x_{0}}\left( x\right) $
is less than $\widehat{a}$ (see Theorem \ref{Th weights 001}), the similarity
between pixels is measured by a linear decreasing function of $\widehat{\rho
}_{K,x_{0}}\left( x\right) ;$ otherwise it is zero. Thus $\widehat{a}$ acts
as an automatic threshold. In Figure\ \ref{Fig weights}, it is shown how
the Optimal Weights Filter chooses in each case a proper weight
configuration.

\begin{figure}[tbp]
\begin{center}
\includegraphics[width=0.6\linewidth]{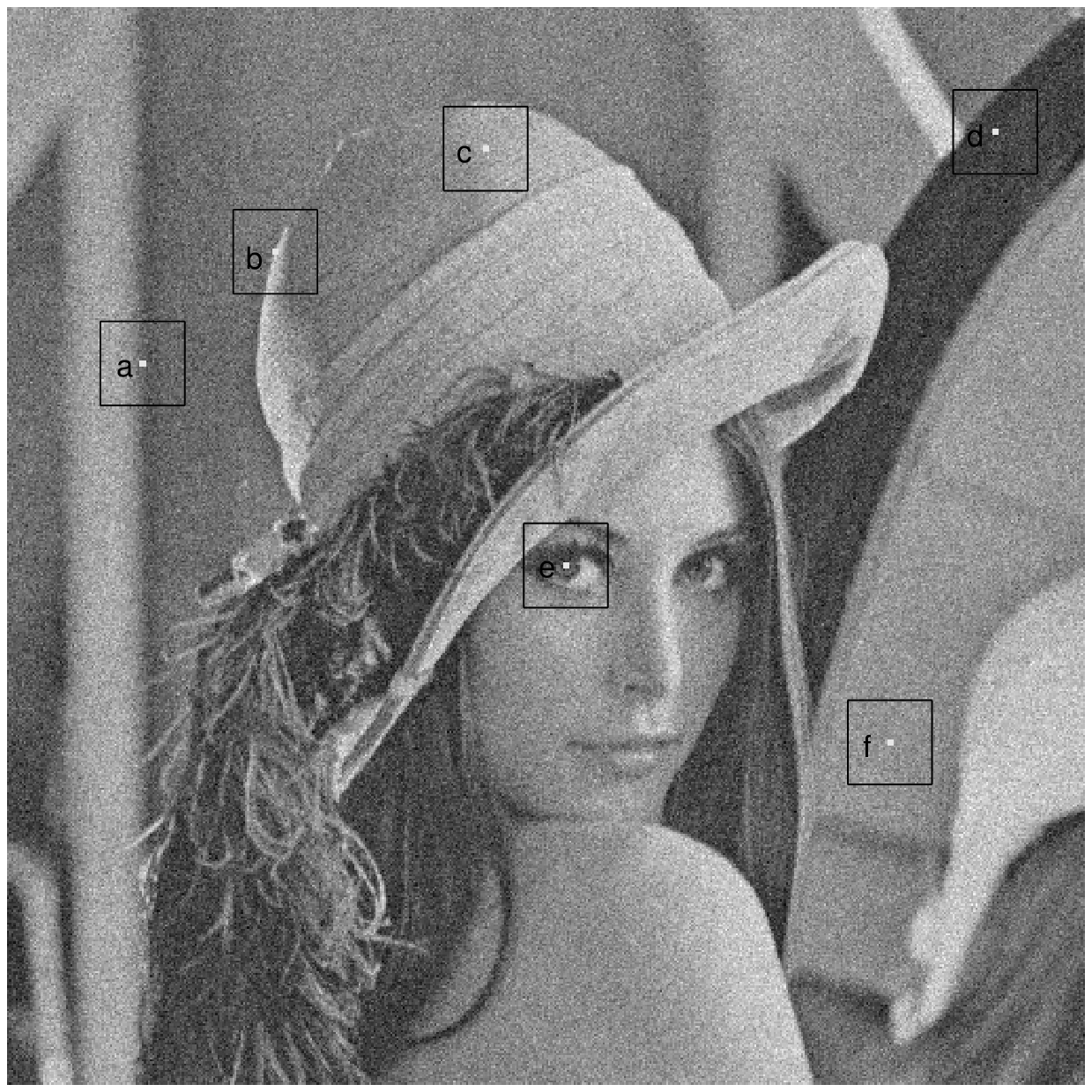}
\end{center}
\caption{\small The noisy image with six selected search windows with centers at
pixels a, b, c, d, e, f.}
\label{Fig selected}
\end{figure}

\begin{figure}[tbp]
\begin{center}
\renewcommand{\arraystretch}{0.5} \addtolength{\tabcolsep}{-6pt} \vskip3mm {%
\fontsize{8pt}{\baselineskip}\selectfont
\begin{tabular}{ccccc}
Original image& Noisy image& 2D representation & 3D representation & Restored  image\\
& & of the weights&of the weights&
\\
\includegraphics[width=0.18\linewidth]{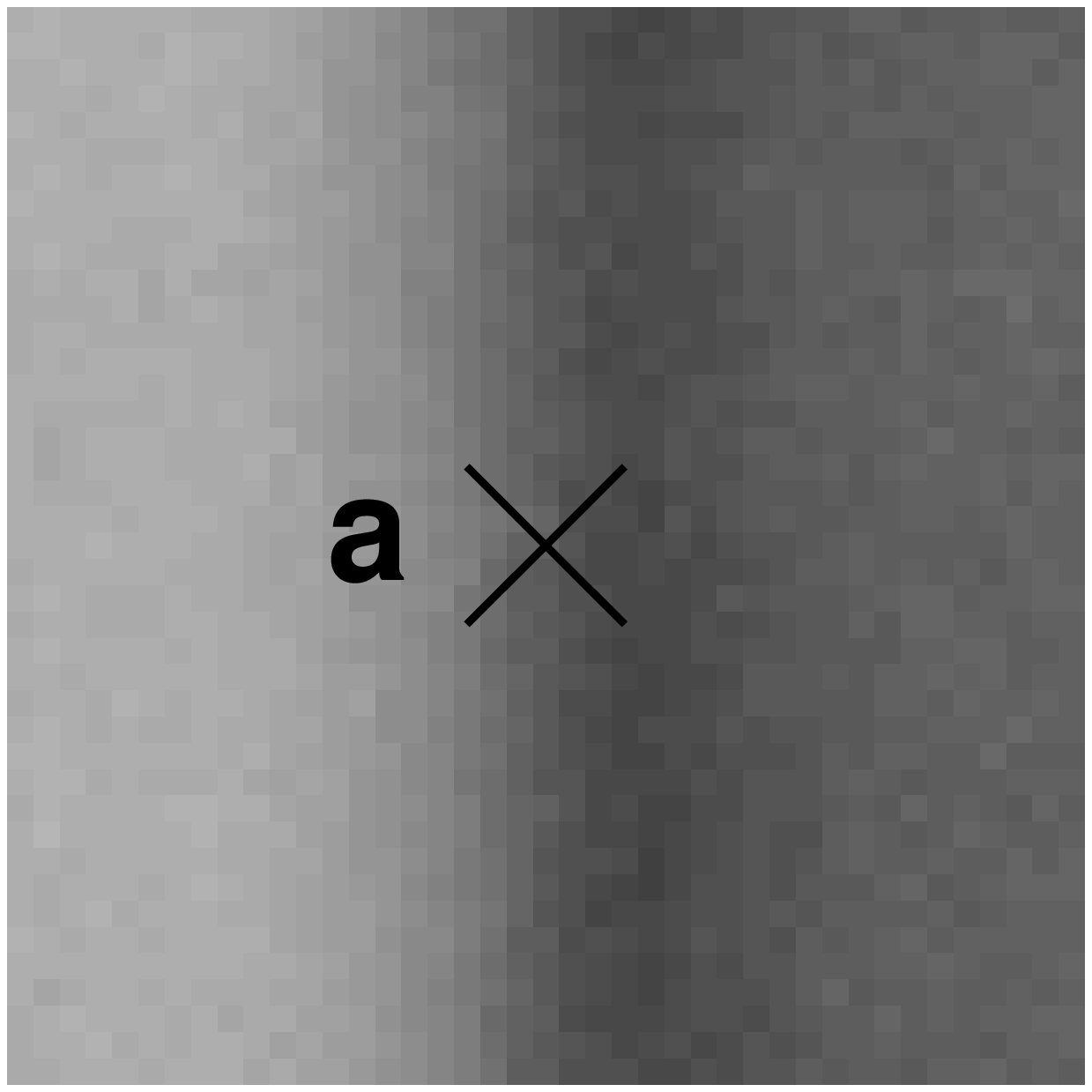} & %
\includegraphics[width=0.18\linewidth]{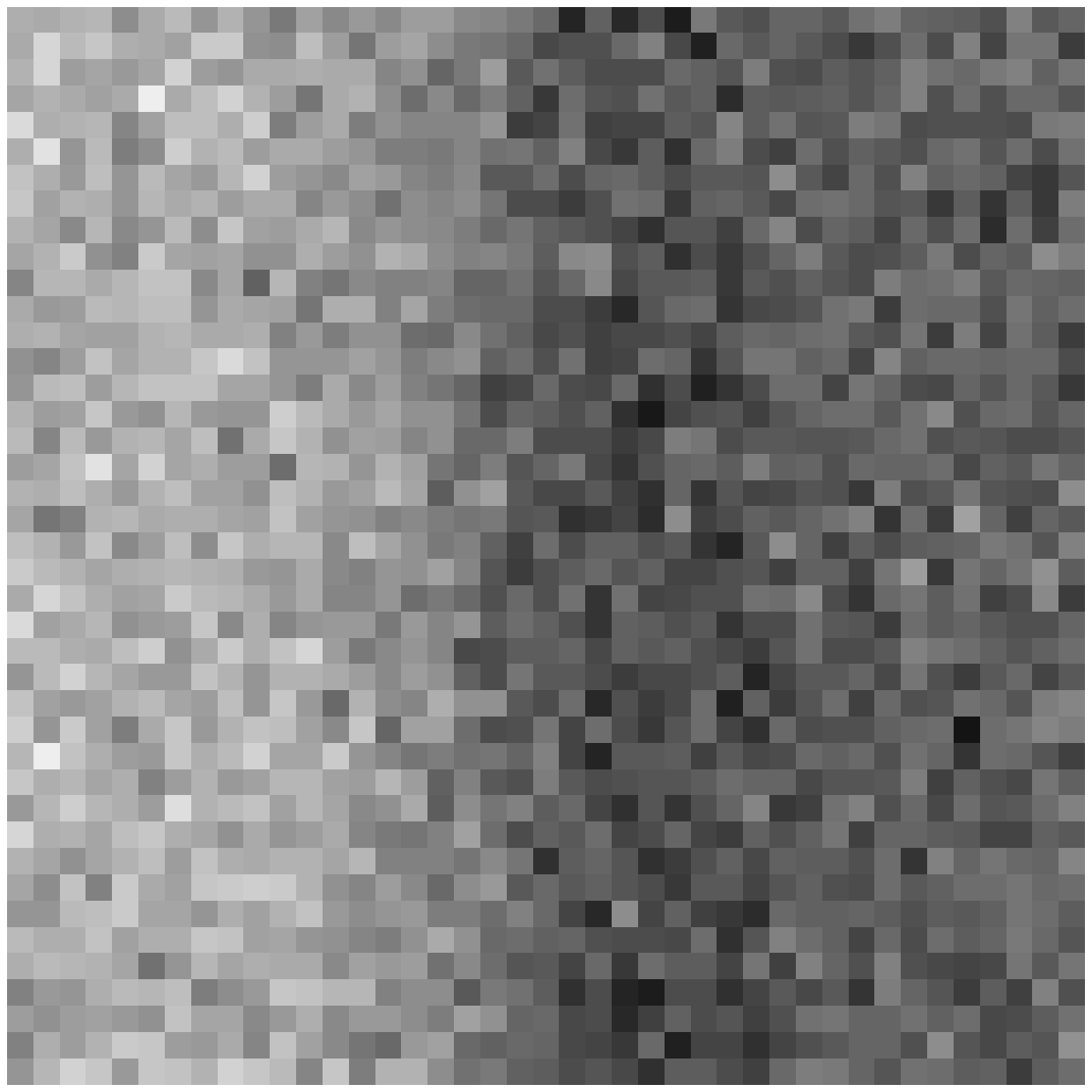} & %
\includegraphics[width=0.18\linewidth]{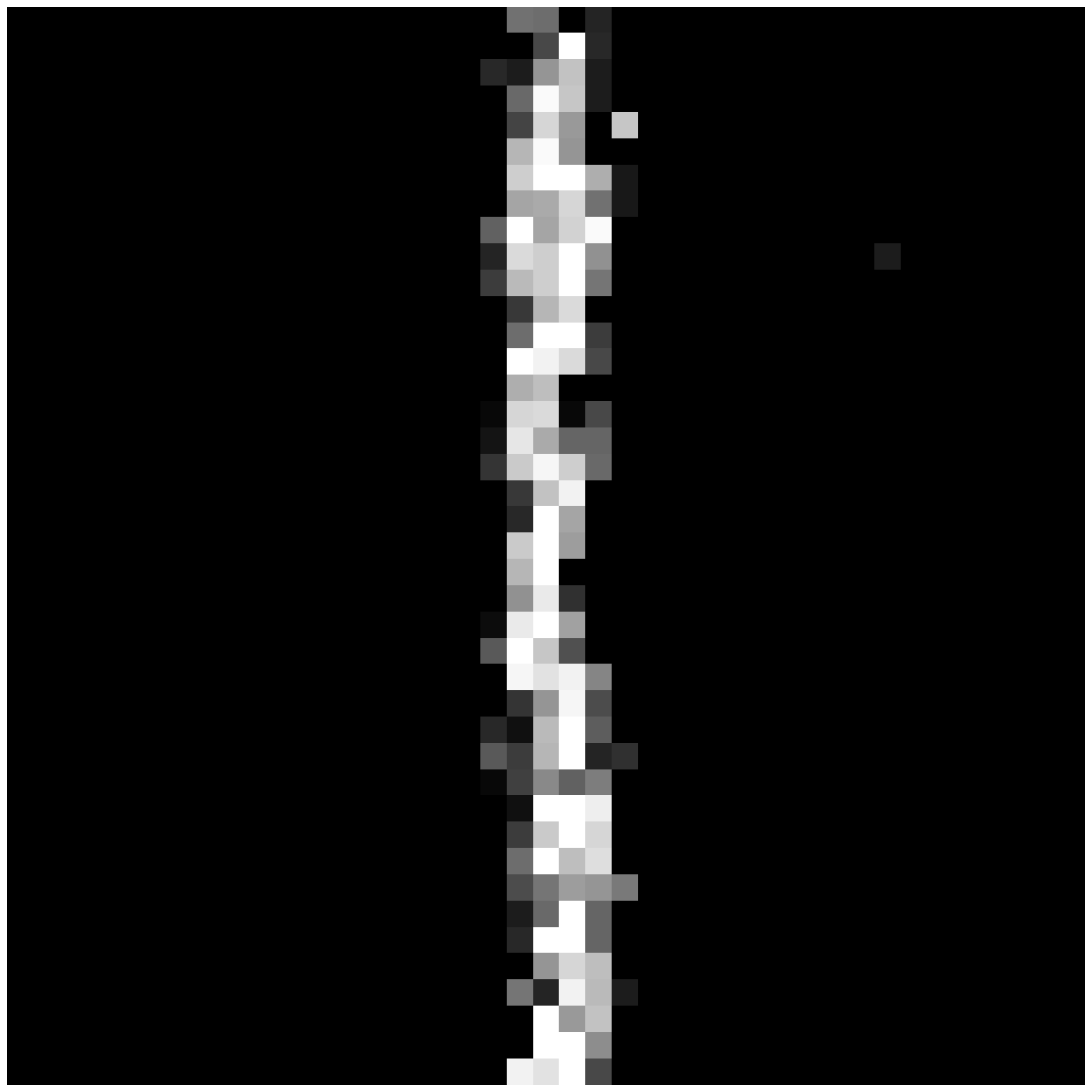} & %
\includegraphics[width=0.18\linewidth]{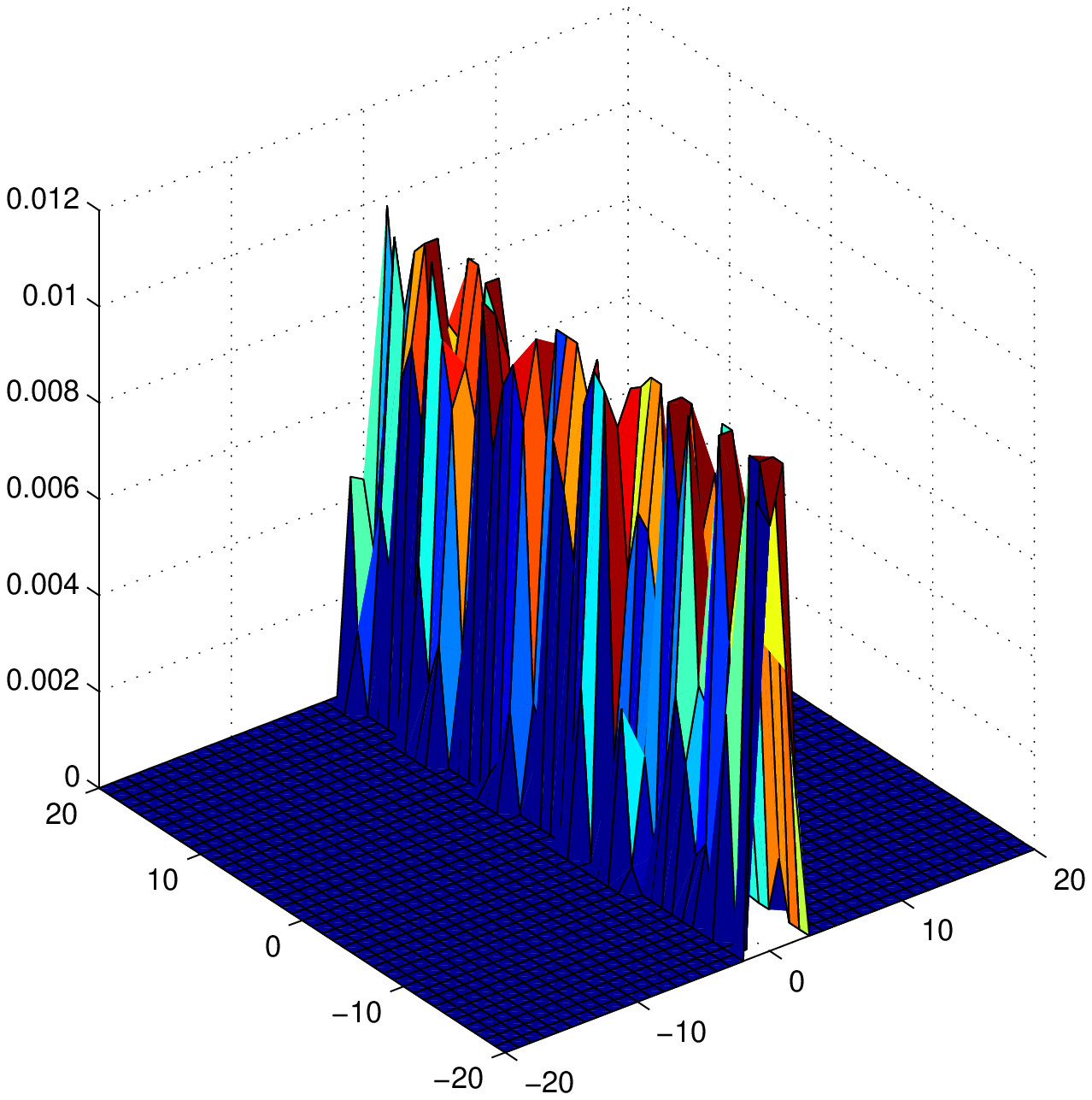} & %
\includegraphics[width=0.18\linewidth]{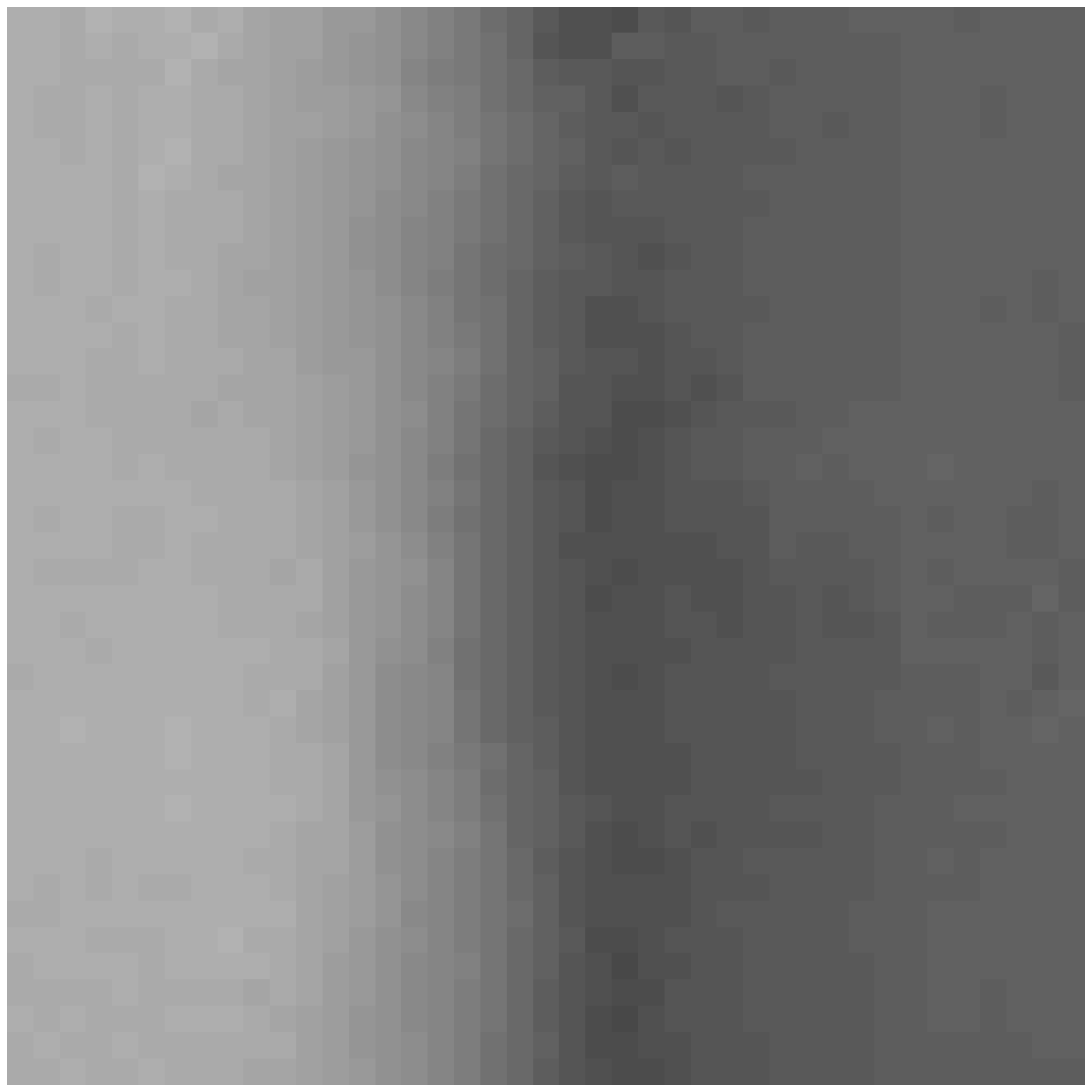} \\
\includegraphics[width=0.18\linewidth]{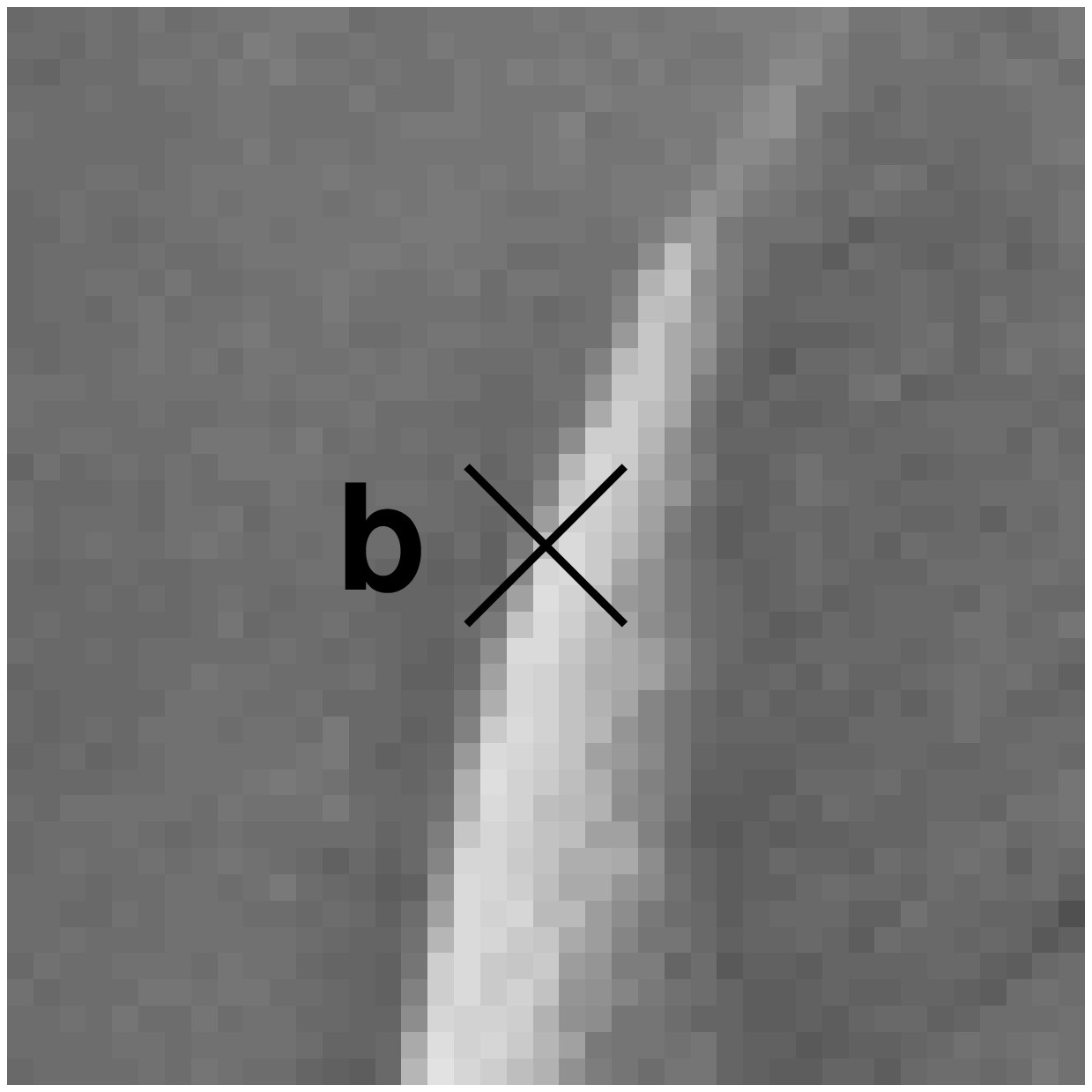} & %
\includegraphics[width=0.18\linewidth]{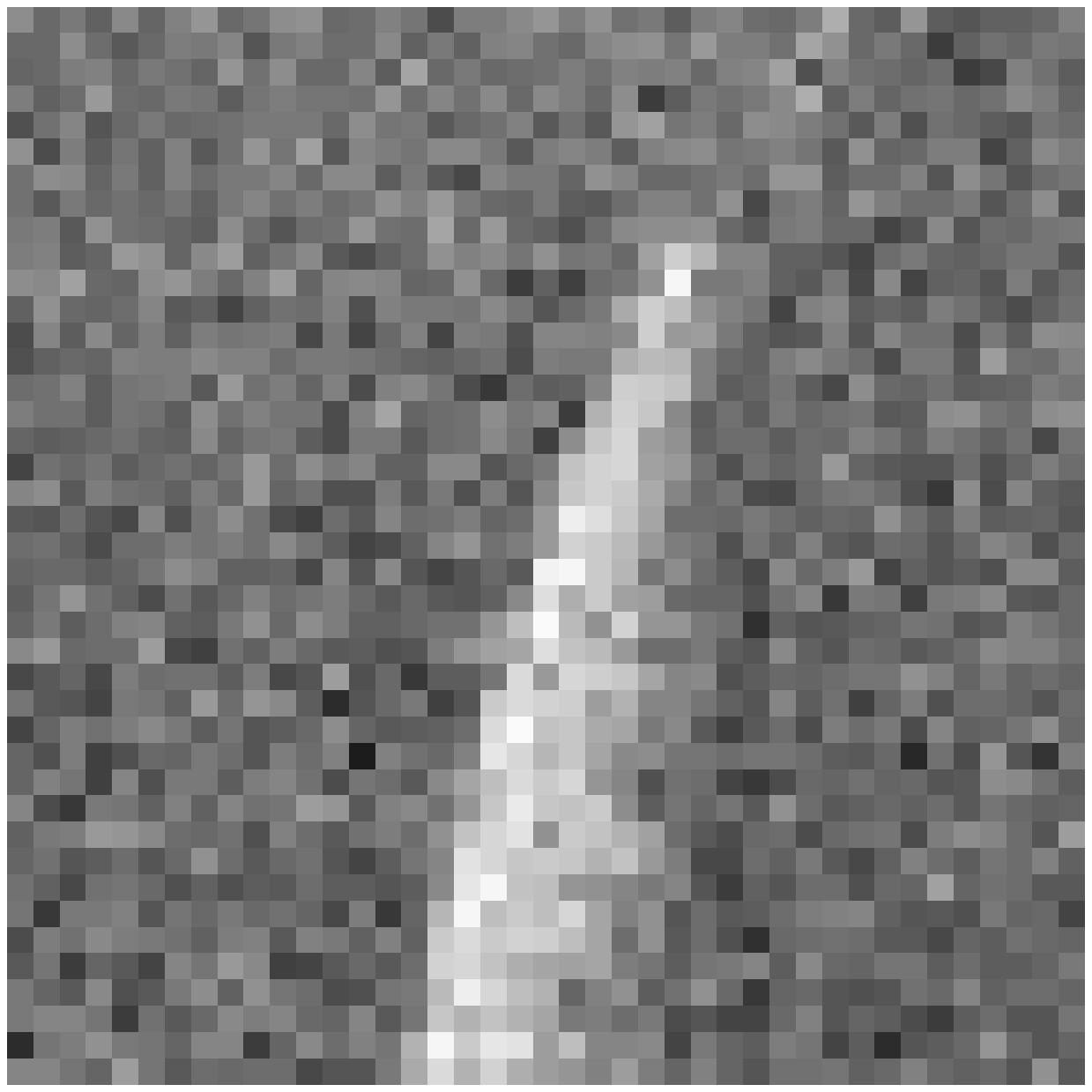} & %
\includegraphics[width=0.18\linewidth]{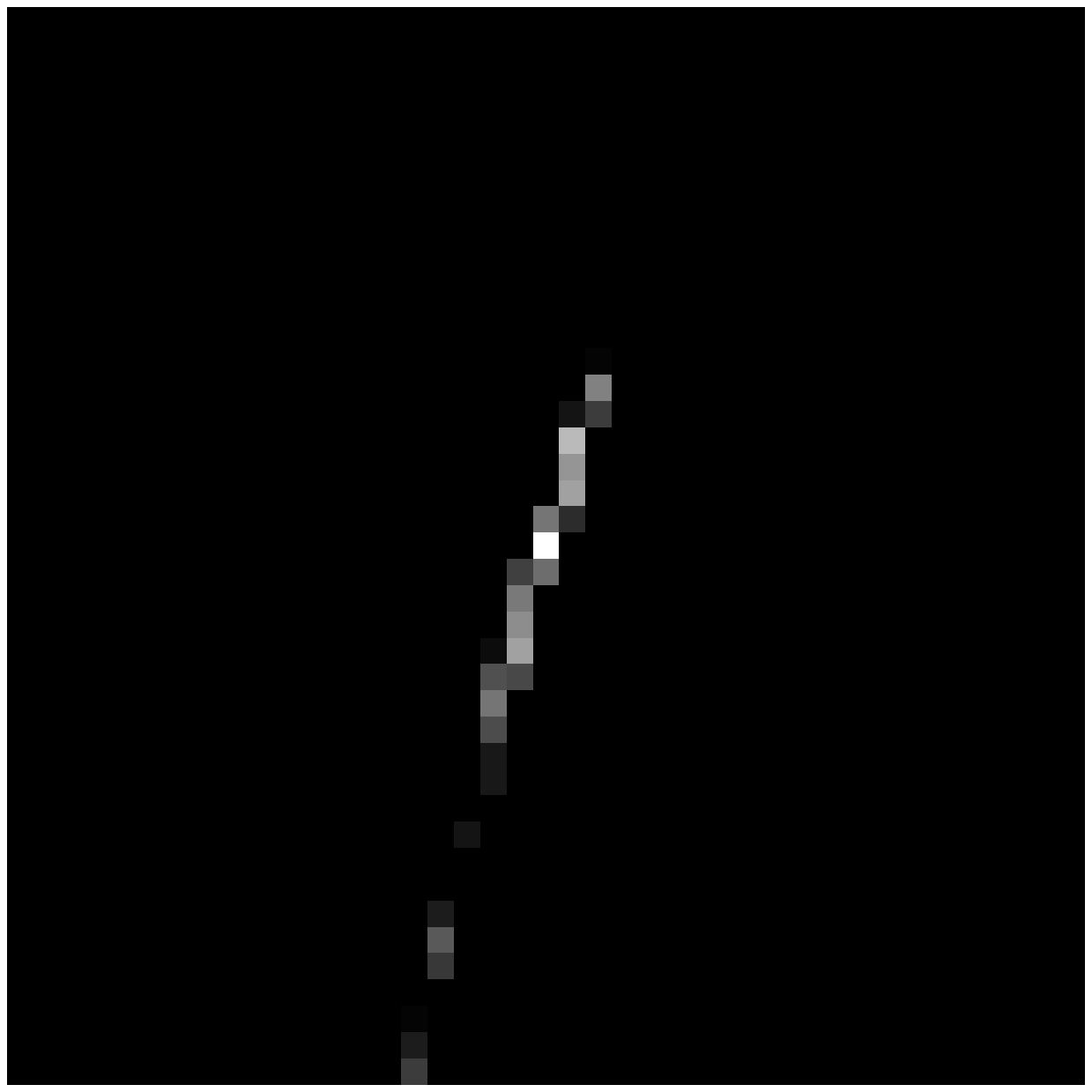} & %
\includegraphics[width=0.18\linewidth]{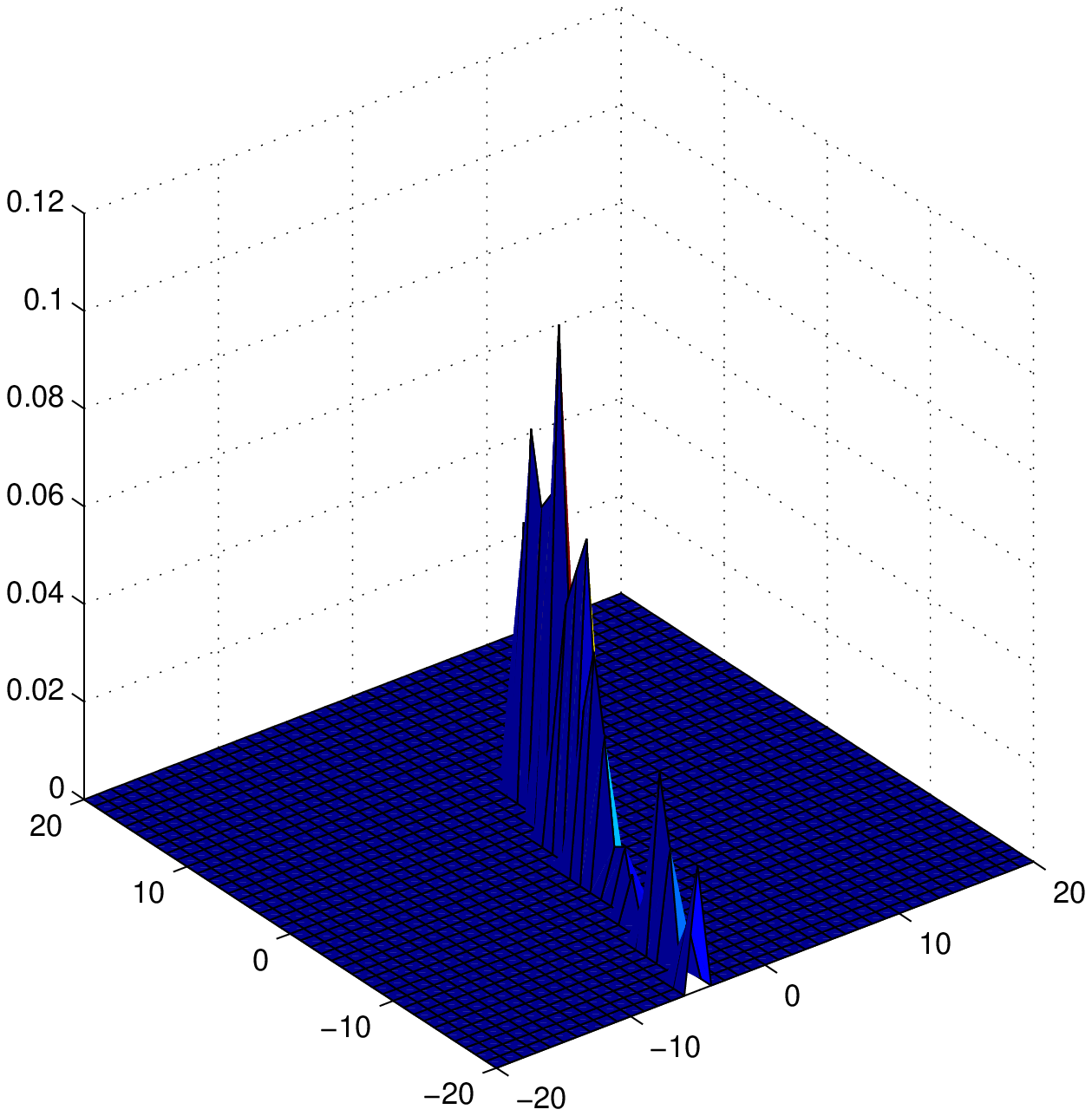} & %
\includegraphics[width=0.18\linewidth]{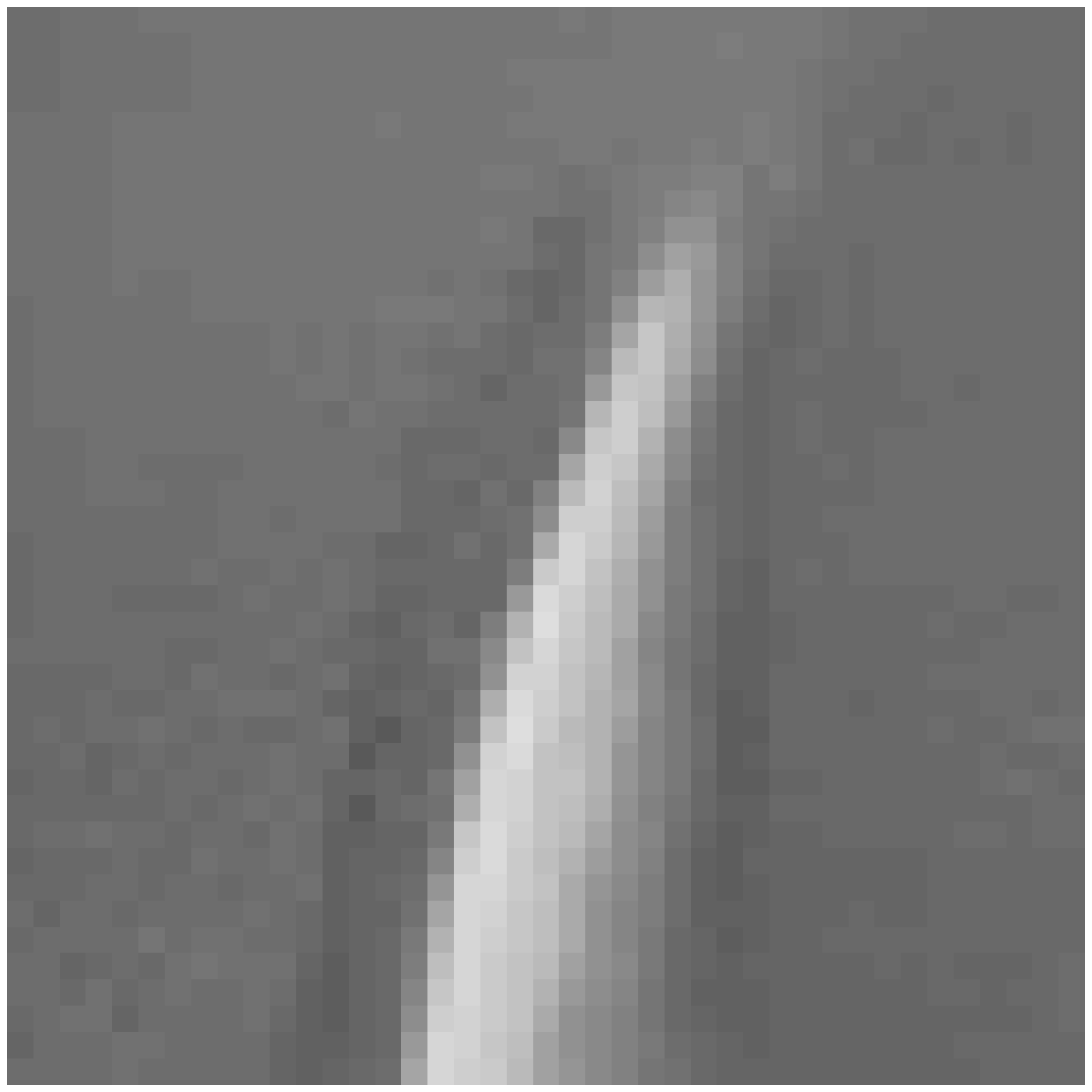} \\
\includegraphics[width=0.18\linewidth]{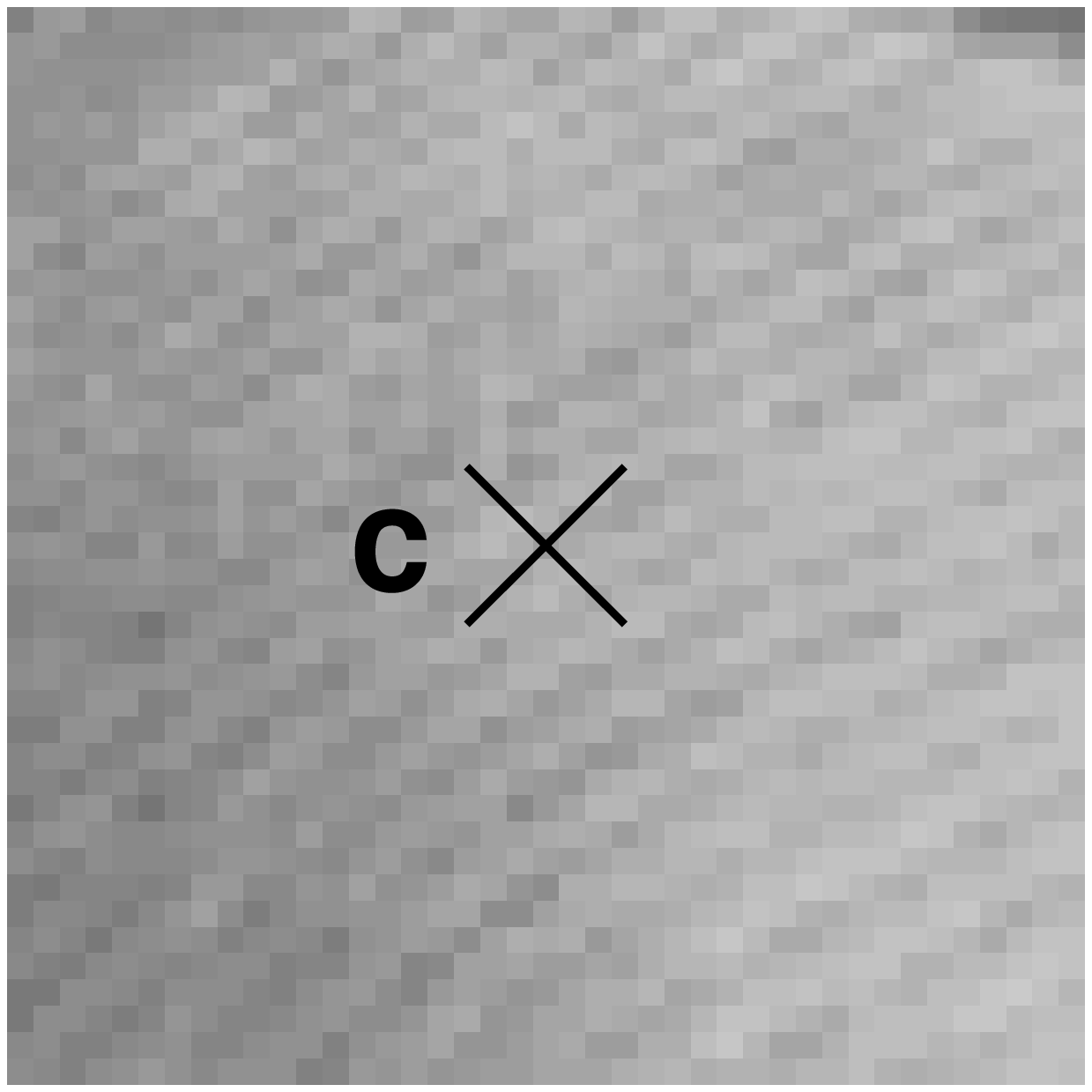} & %
\includegraphics[width=0.18\linewidth]{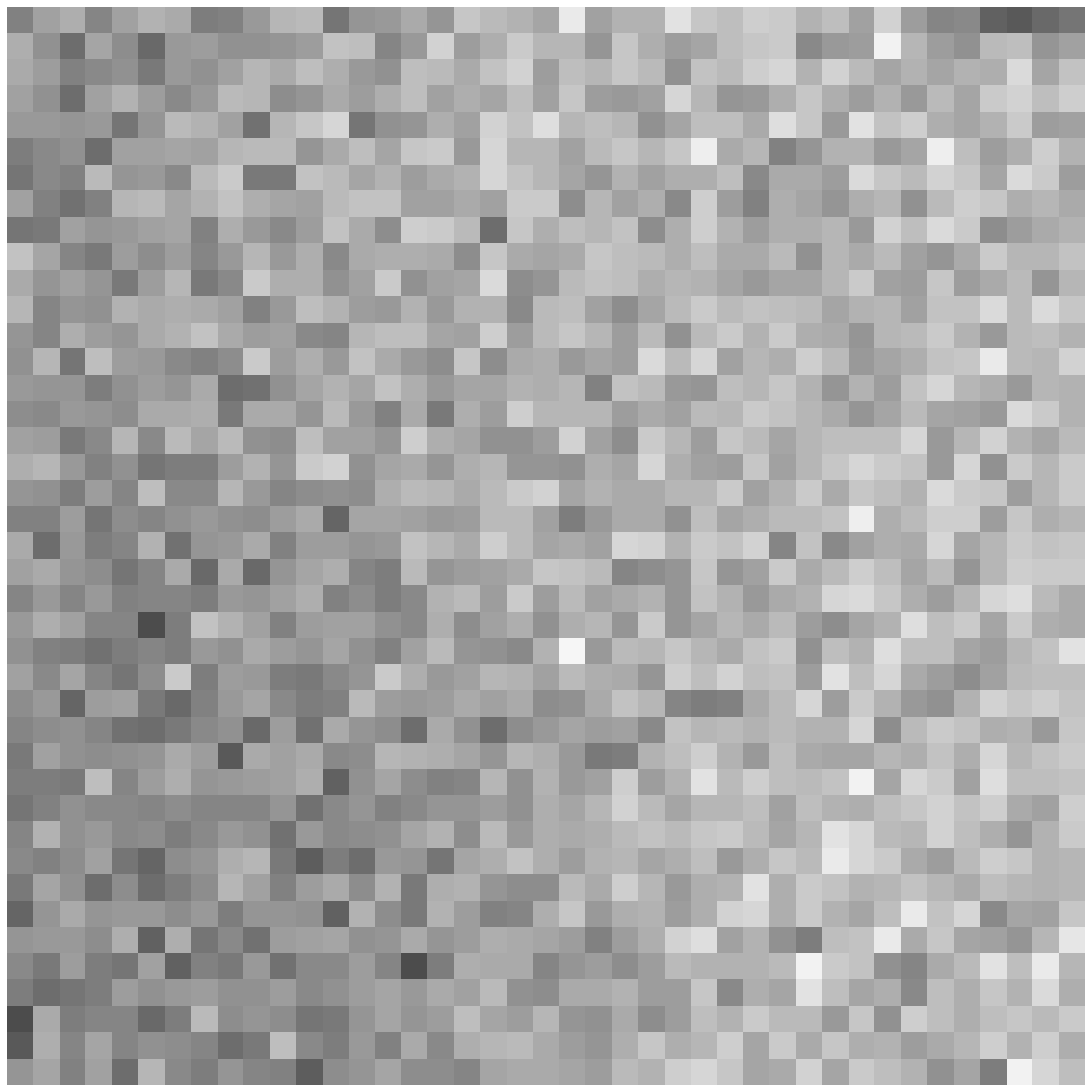} & %
\includegraphics[width=0.18\linewidth]{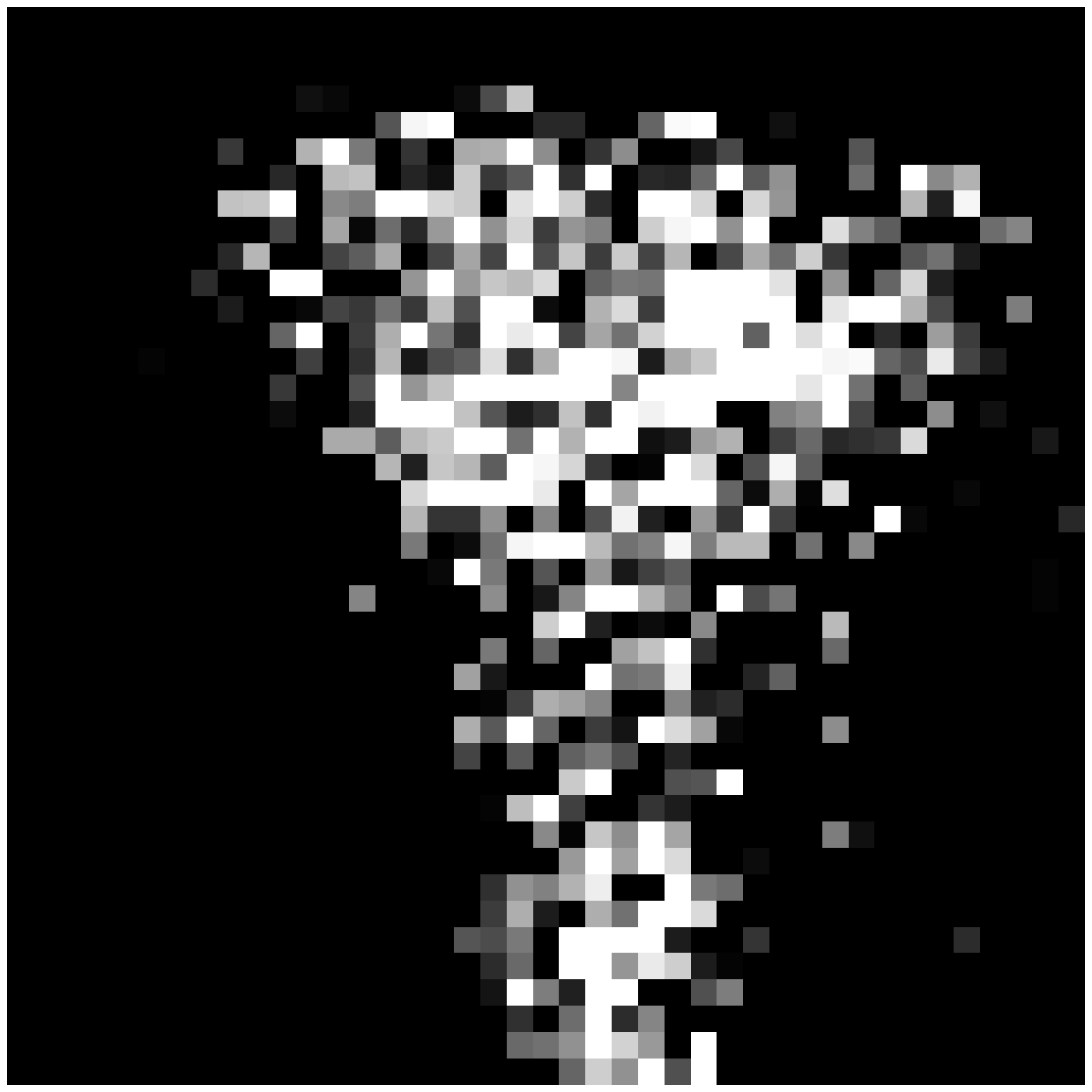} & %
\includegraphics[width=0.18\linewidth]{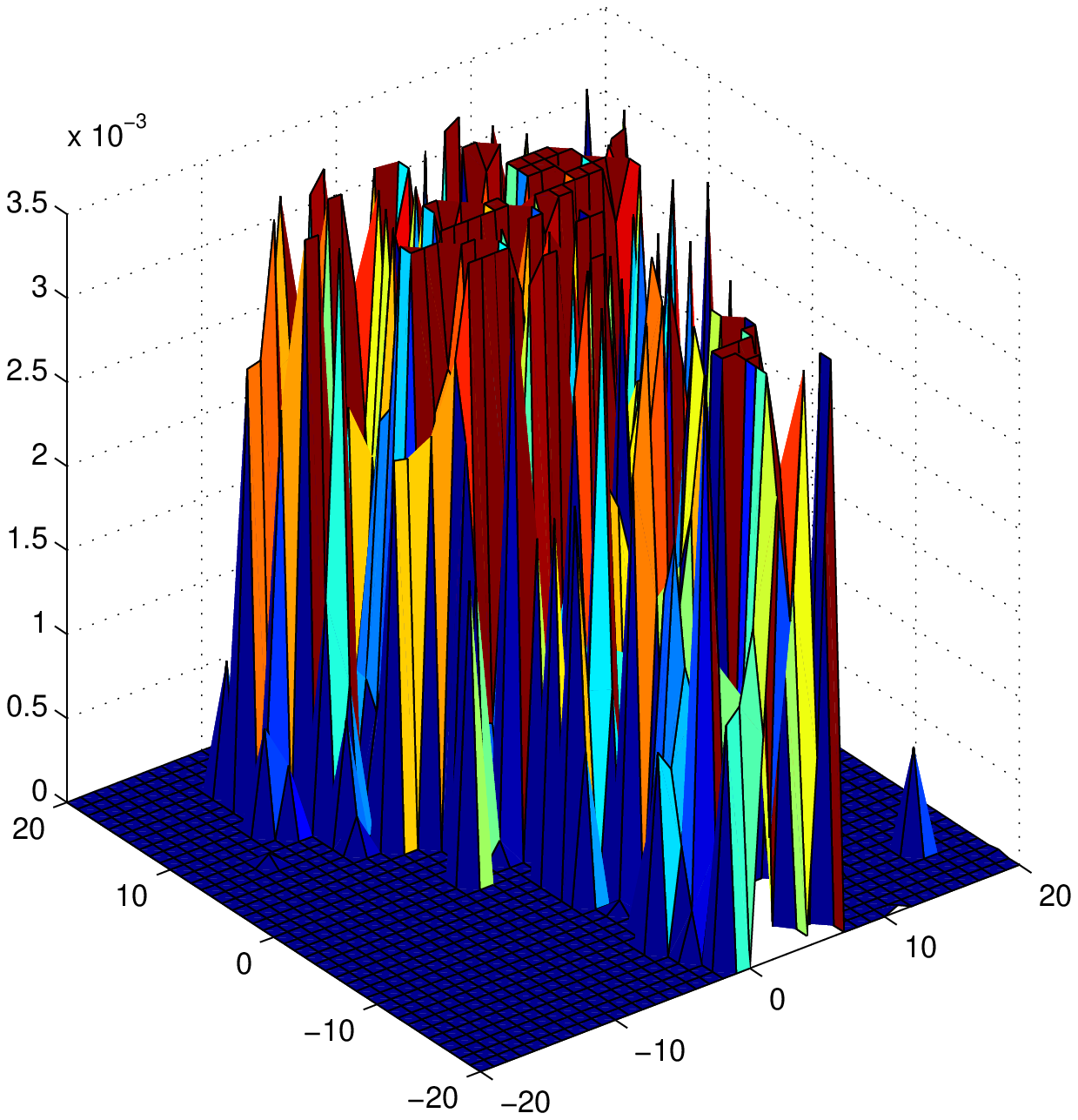} & %
\includegraphics[width=0.18\linewidth]{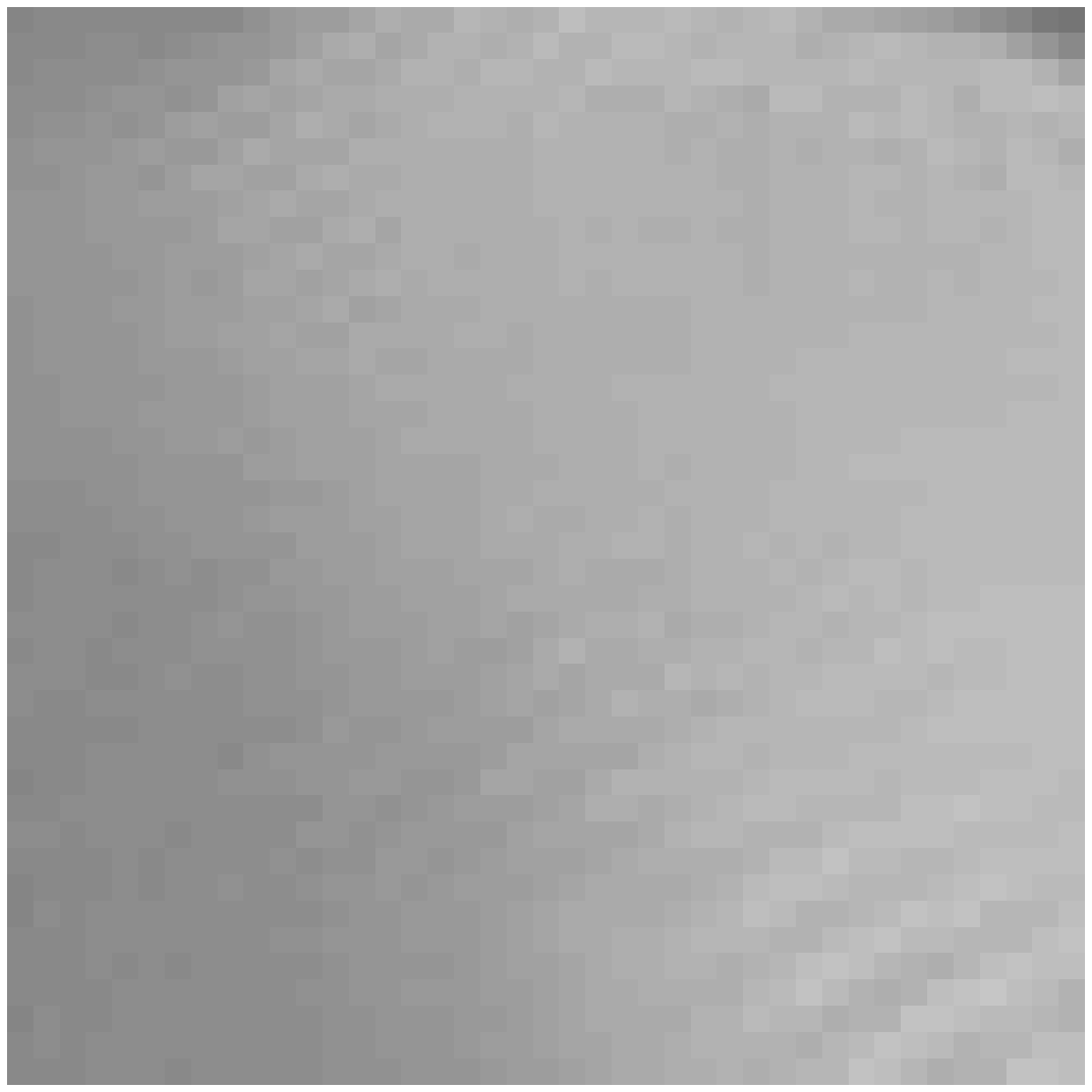} \\
\includegraphics[width=0.18\linewidth]{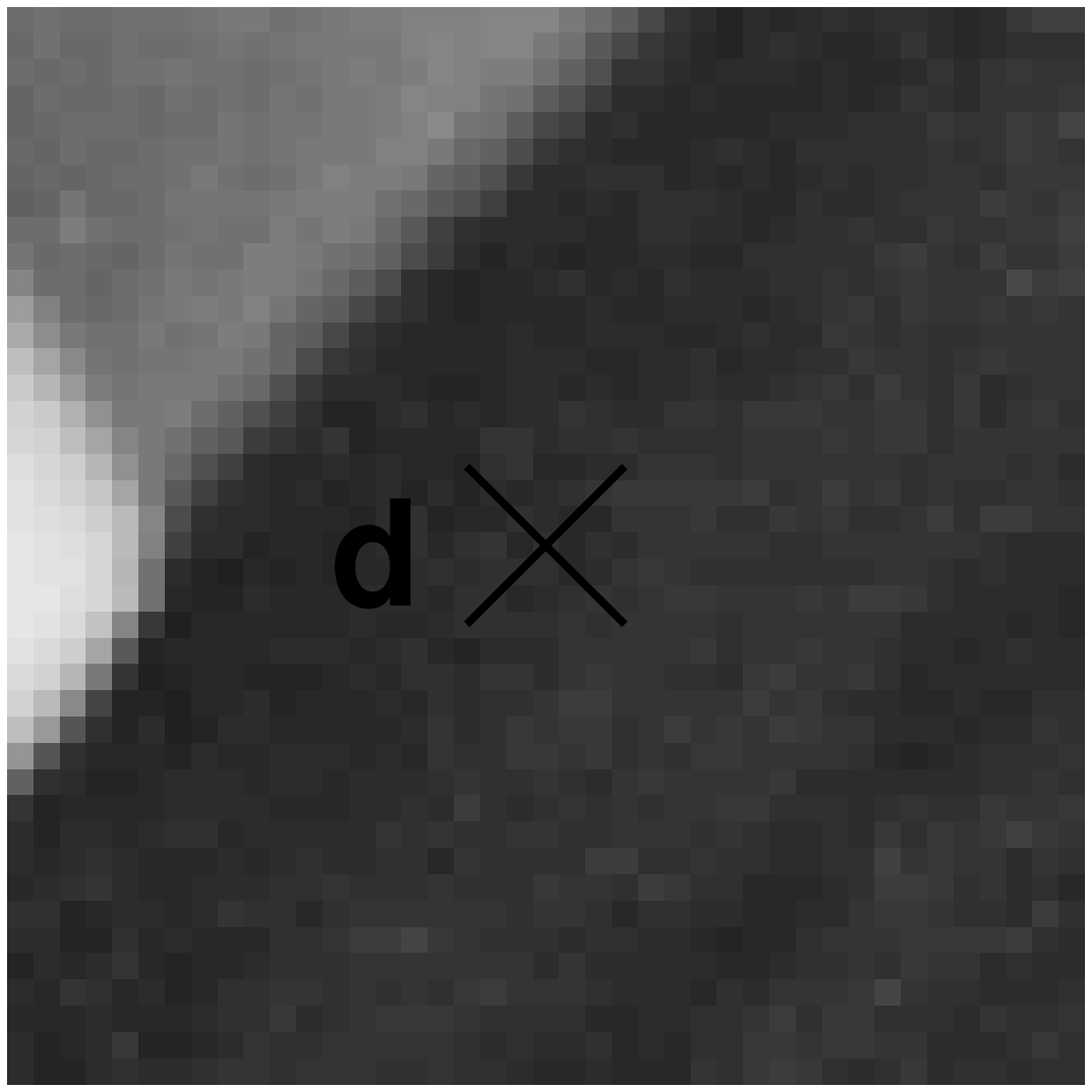} & %
\includegraphics[width=0.18\linewidth]{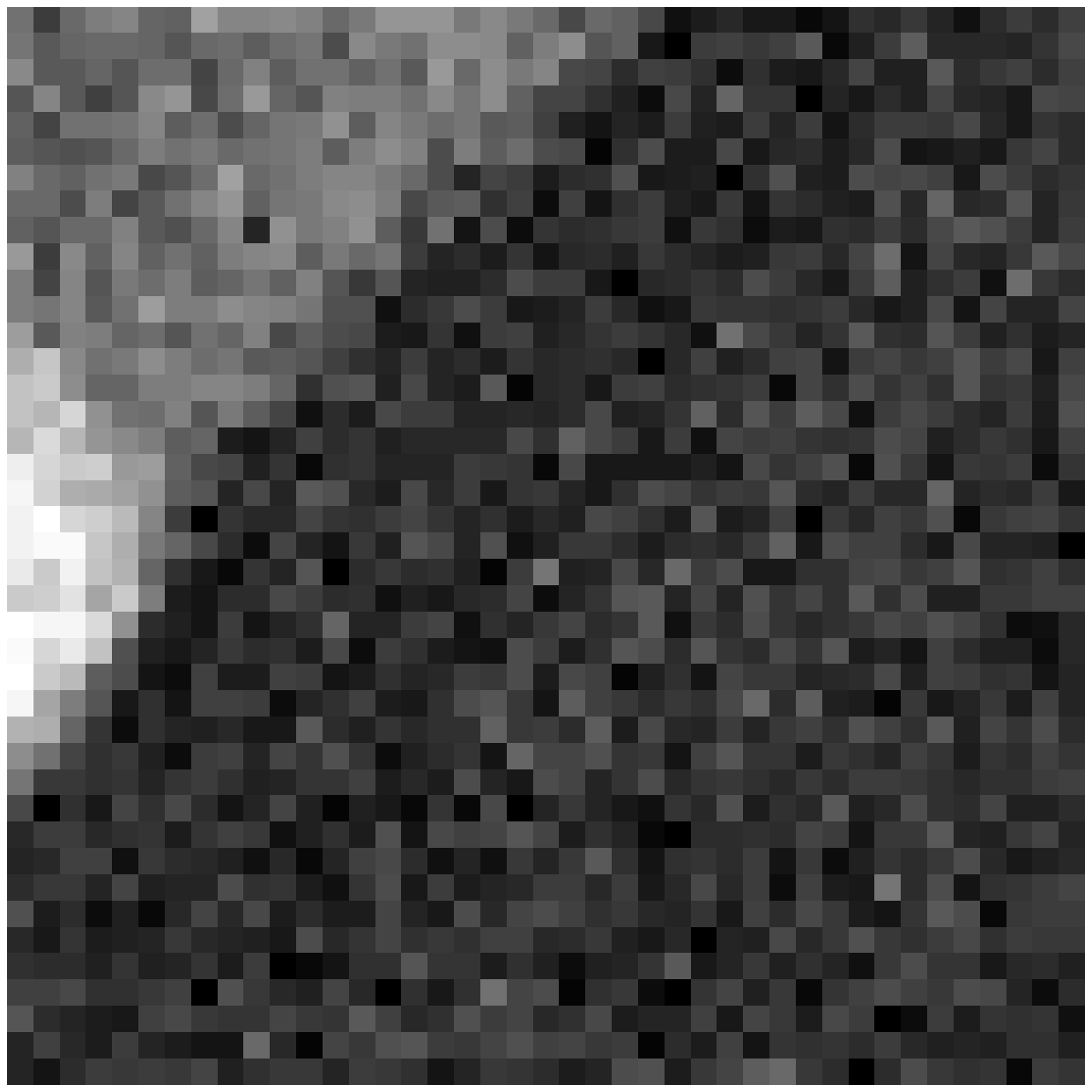} & %
\includegraphics[width=0.18\linewidth]{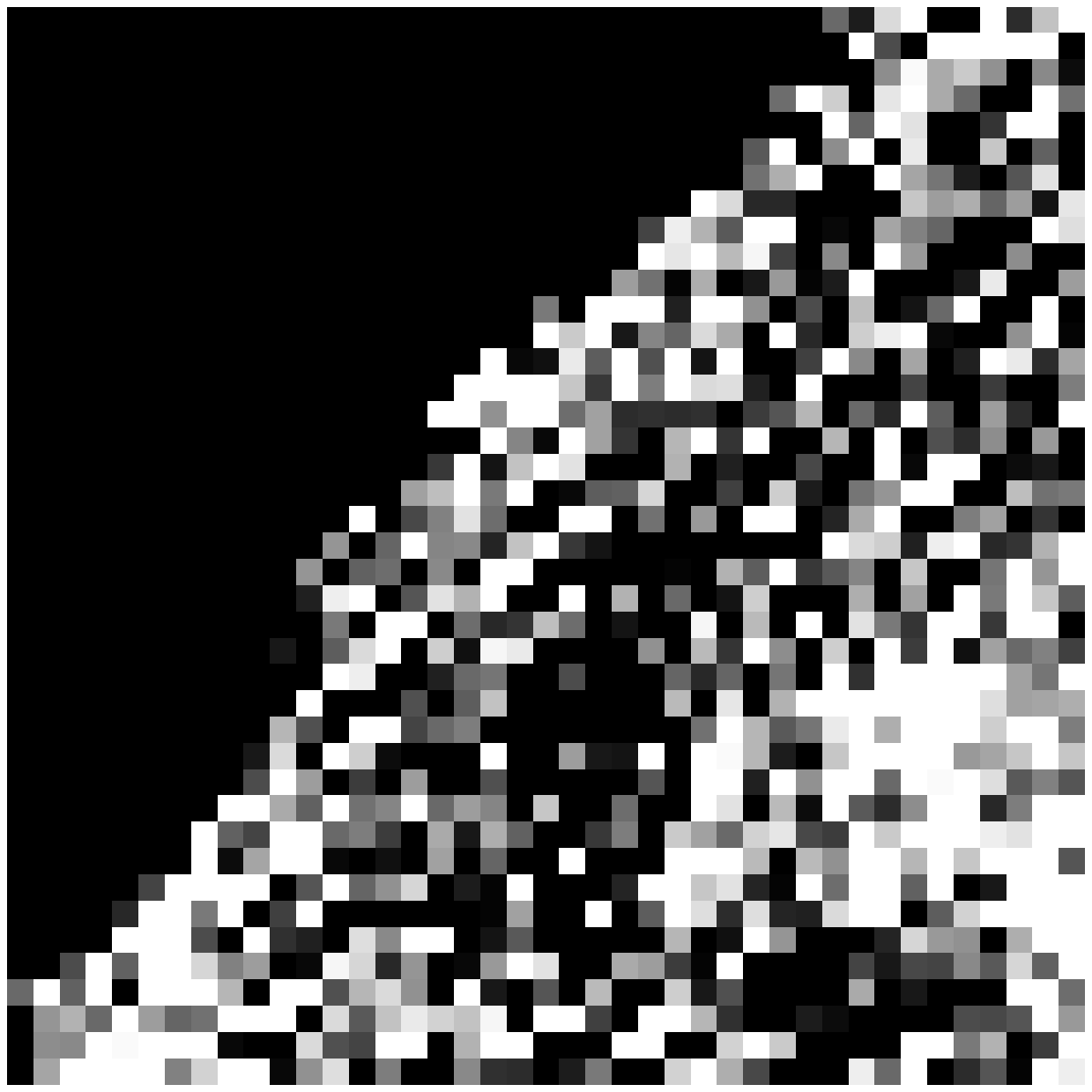} & %
\includegraphics[width=0.18\linewidth]{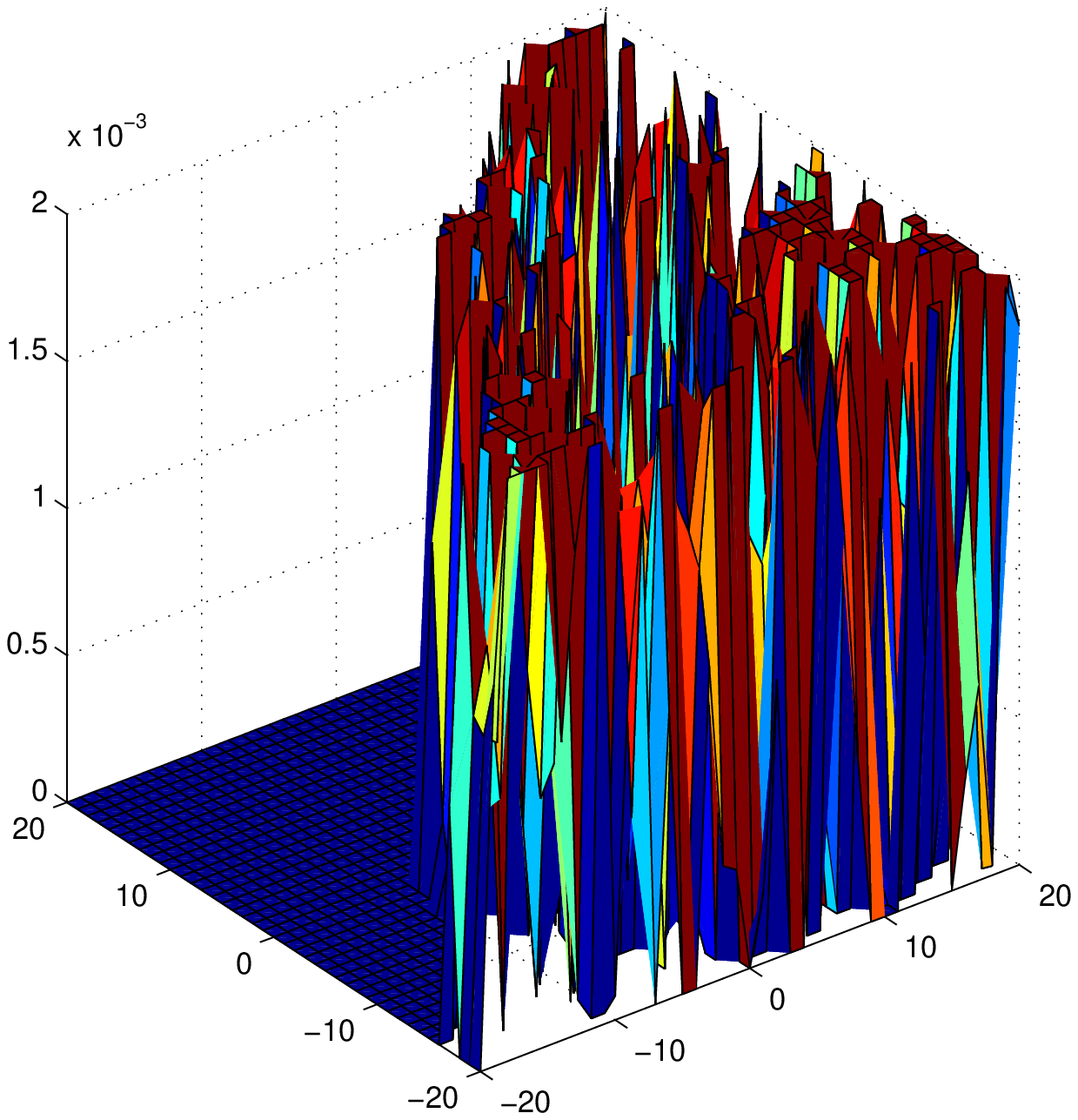} & %
\includegraphics[width=0.18\linewidth]{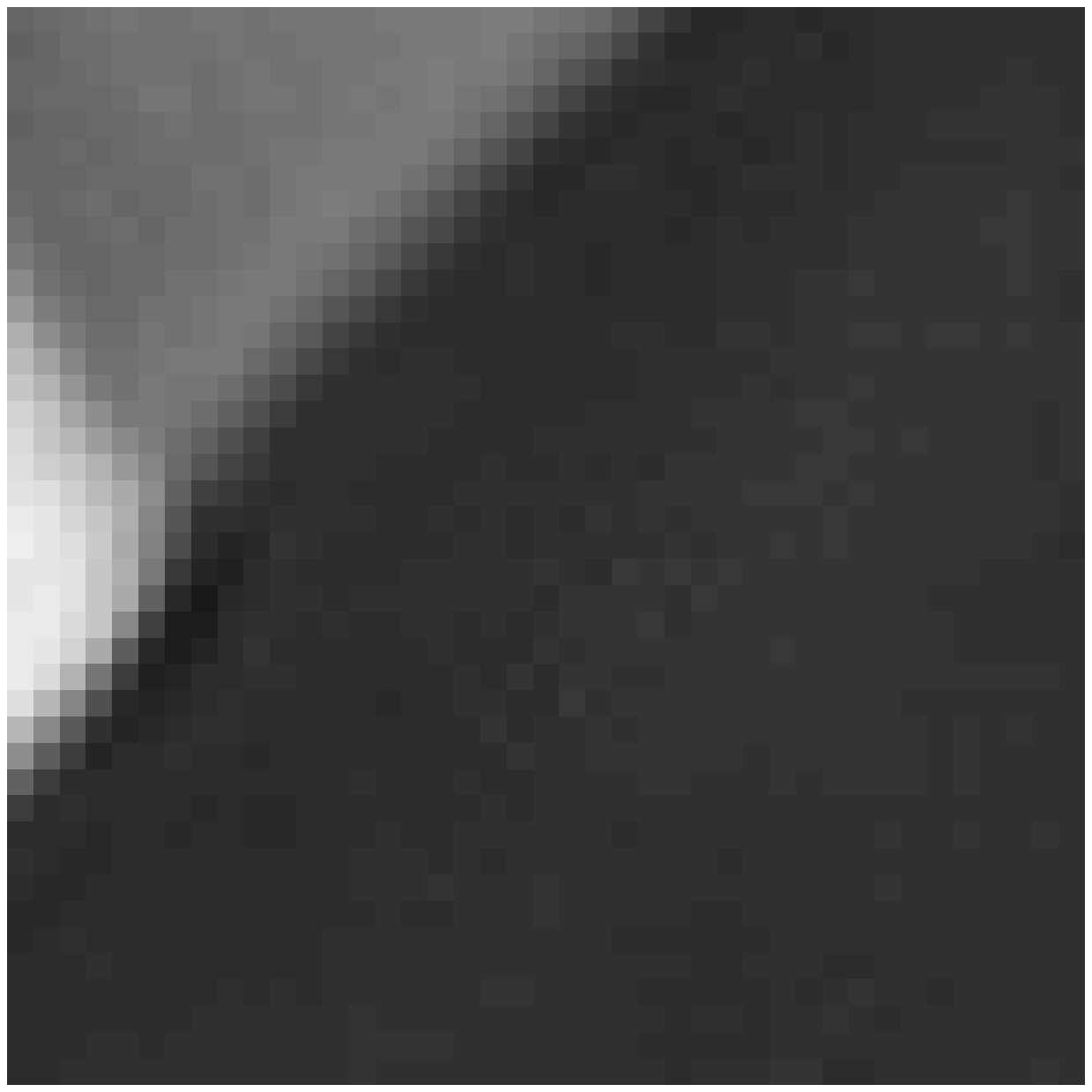} \\
\includegraphics[width=0.18\linewidth]{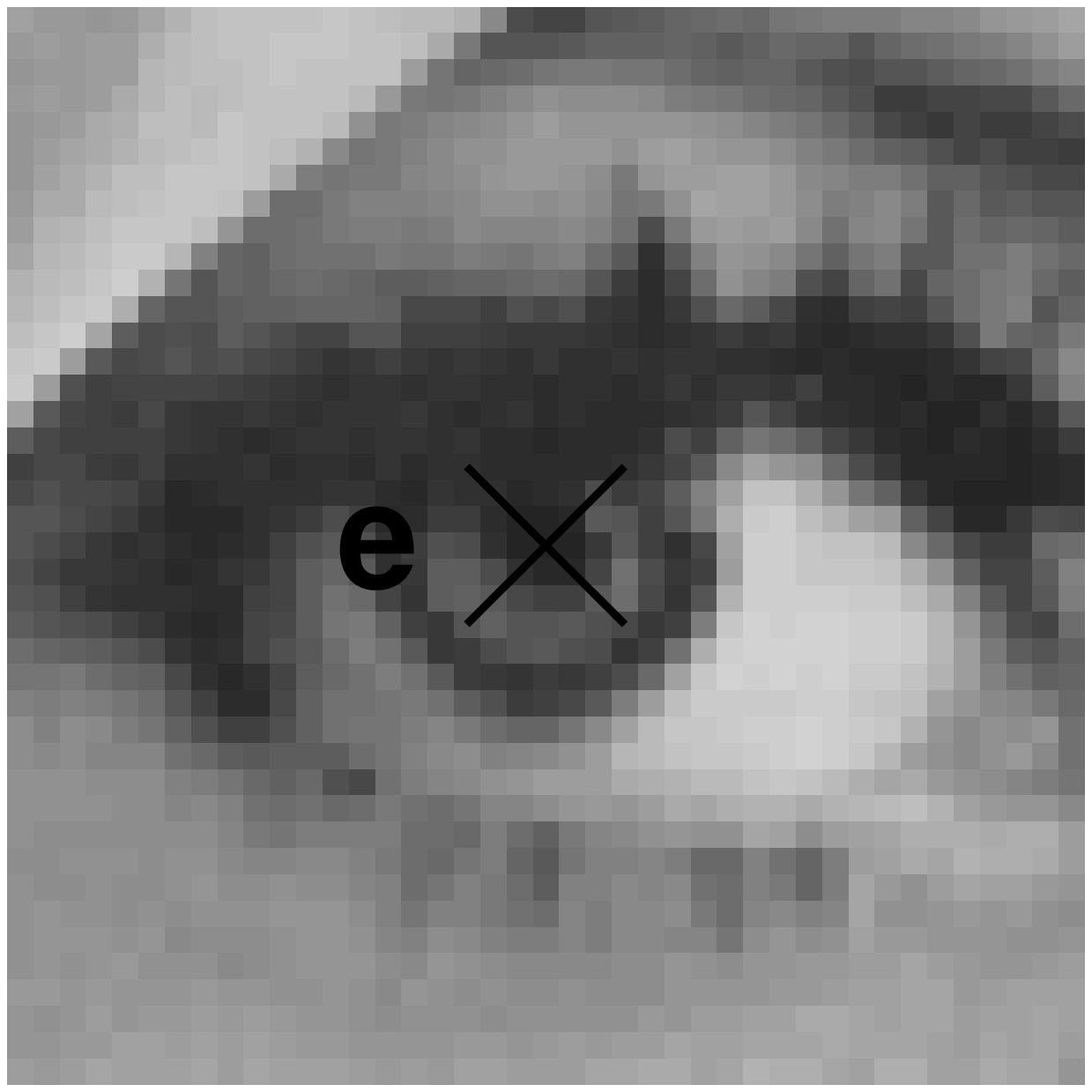} & %
\includegraphics[width=0.18\linewidth]{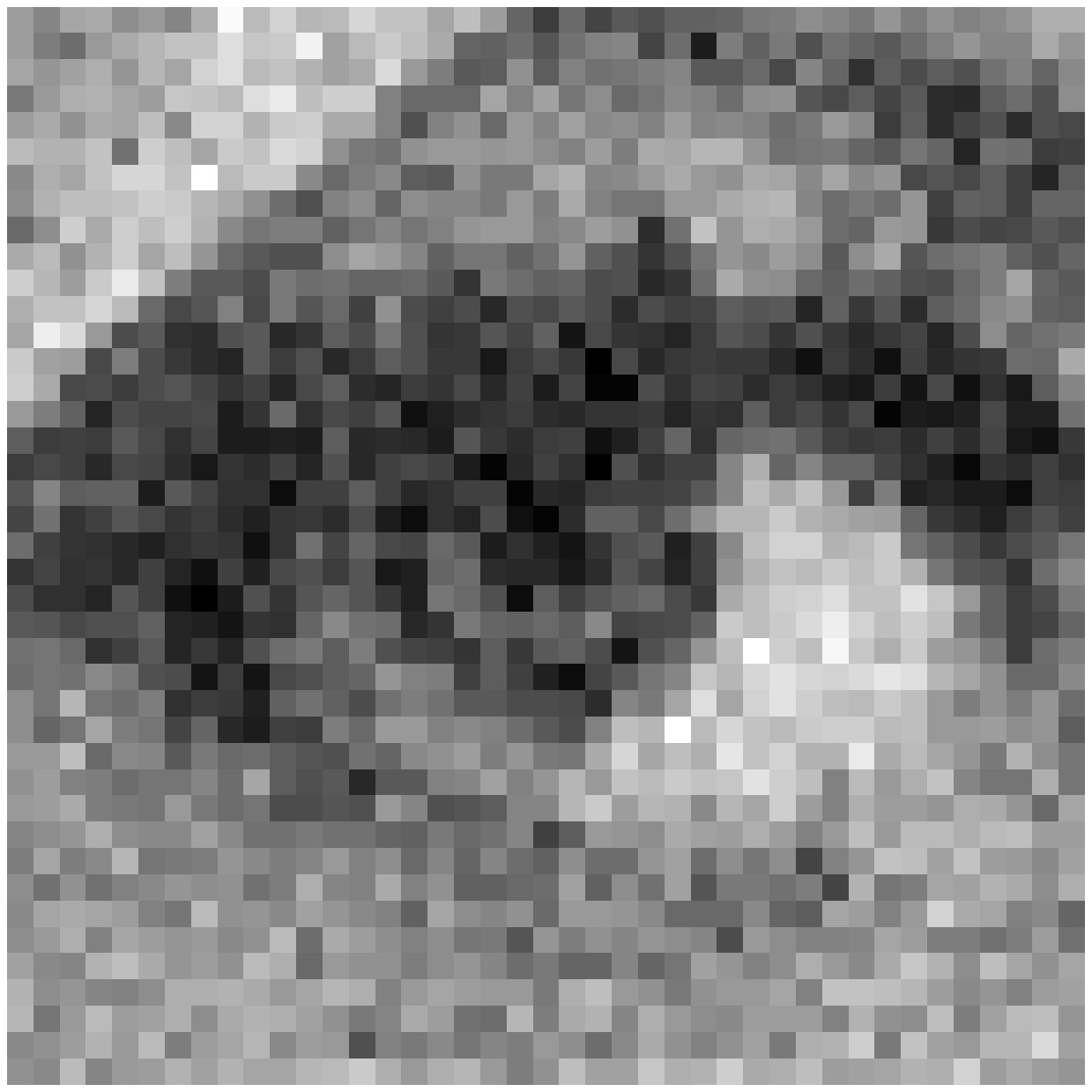} & %
\includegraphics[width=0.18\linewidth]{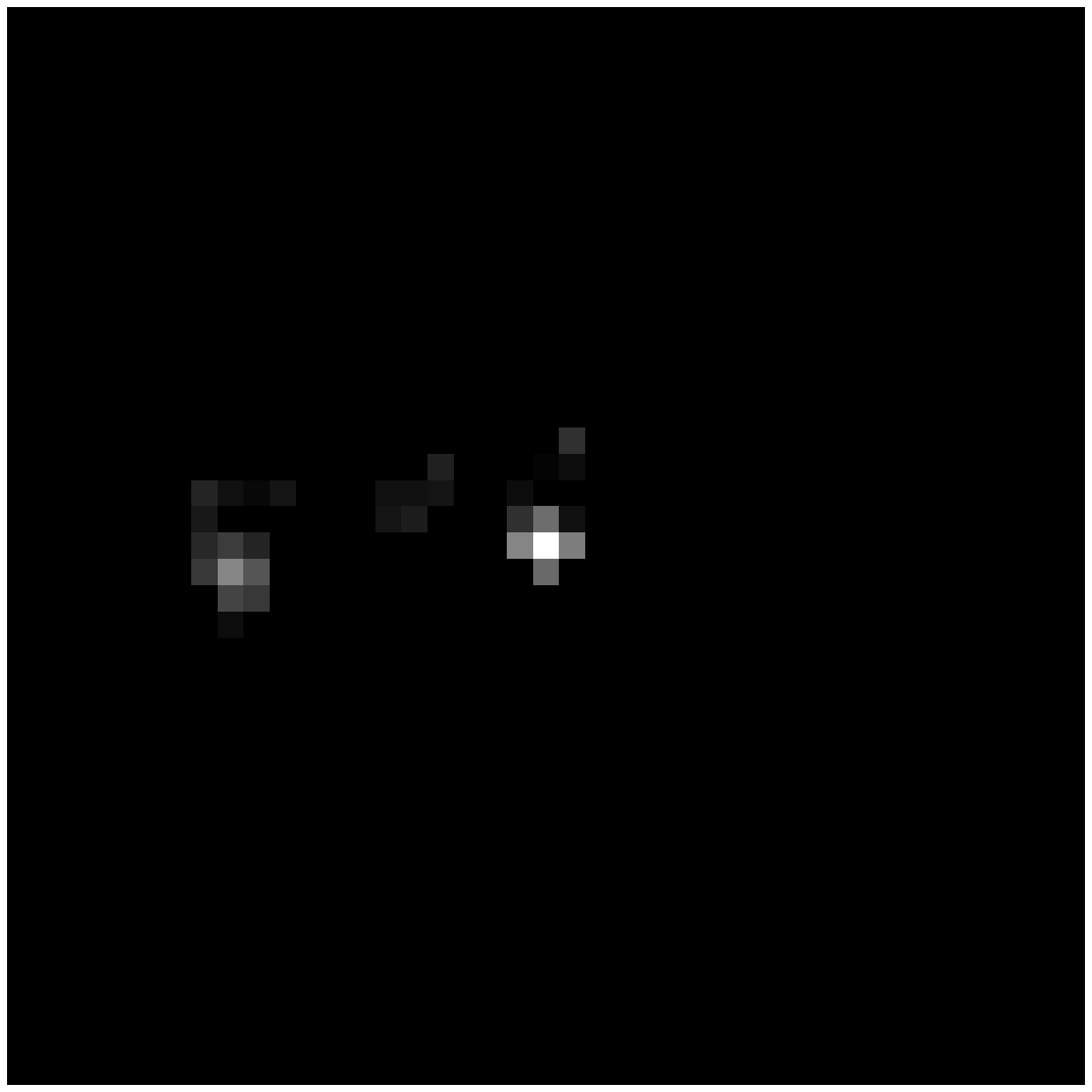} & %
\includegraphics[width=0.18\linewidth]{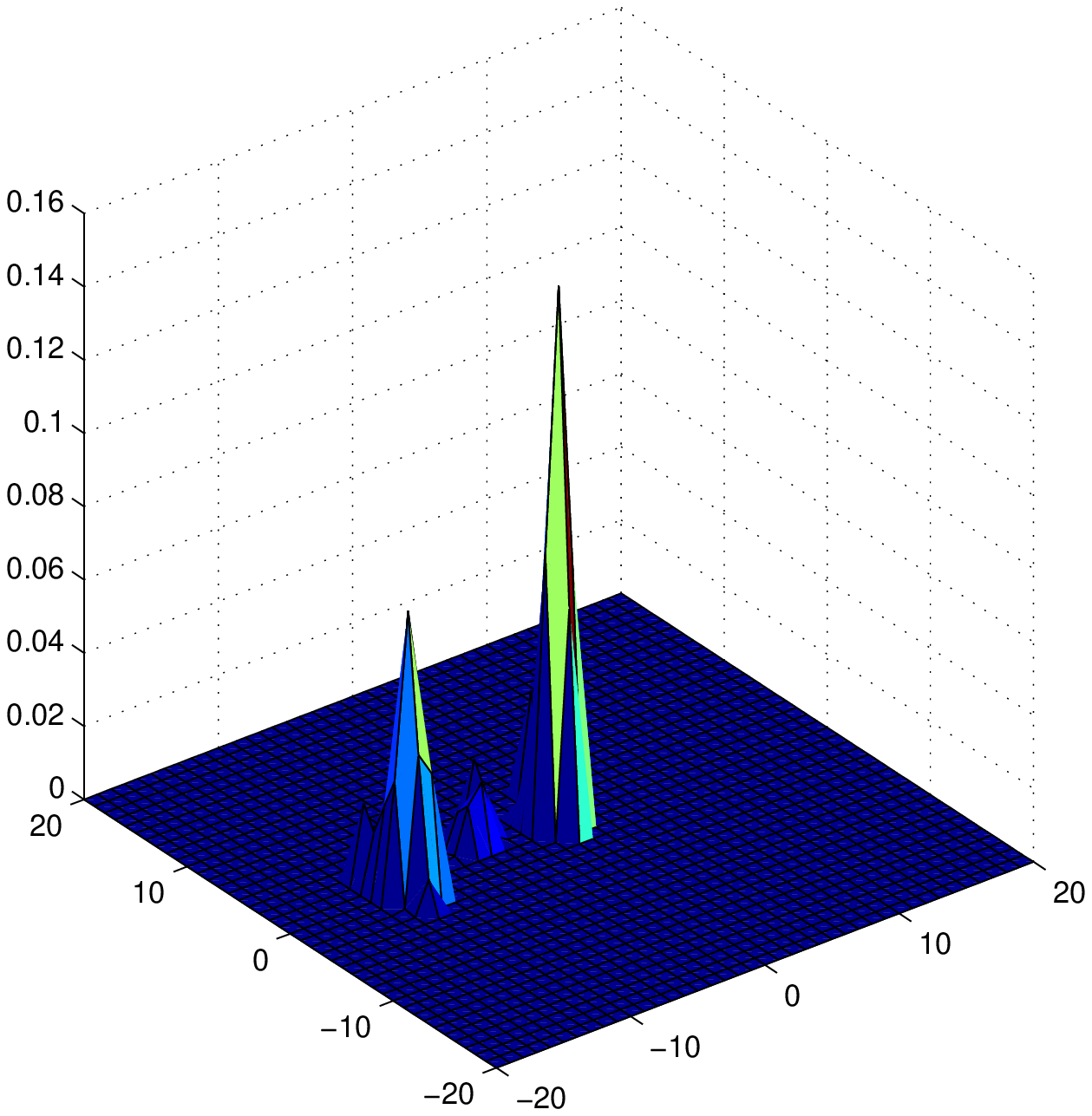} & %
\includegraphics[width=0.18\linewidth]{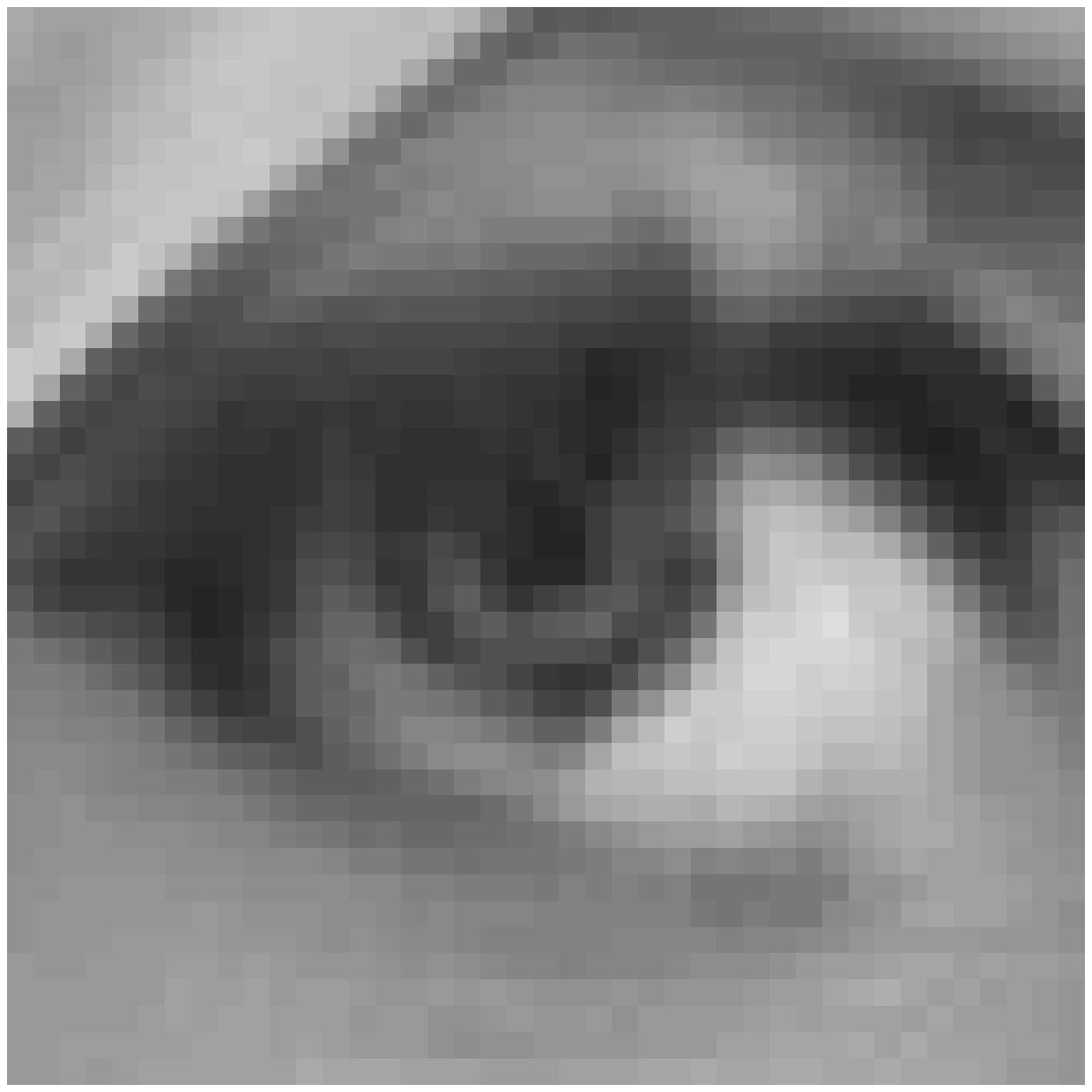} \\
\includegraphics[width=0.18\linewidth]{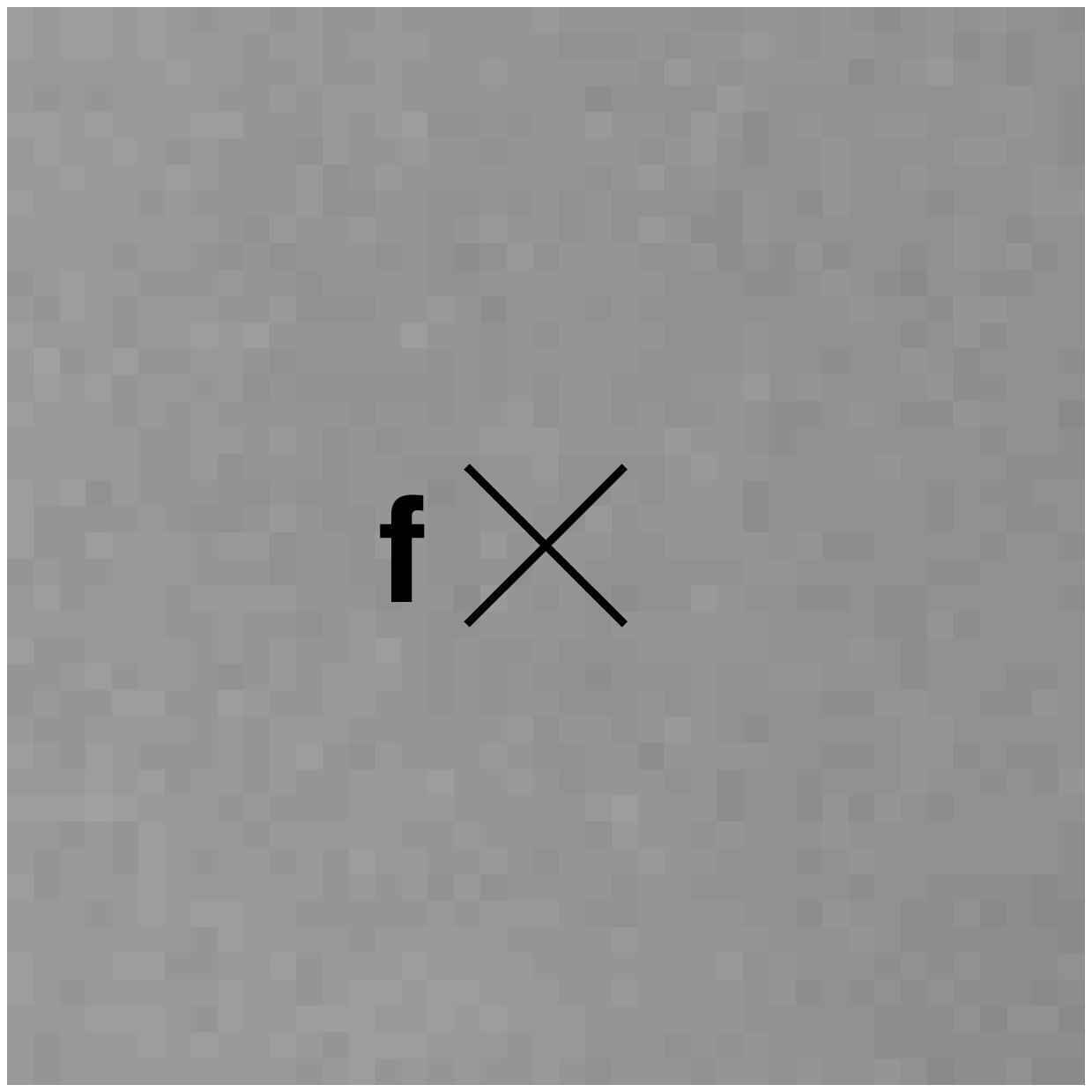} & %
\includegraphics[width=0.18\linewidth]{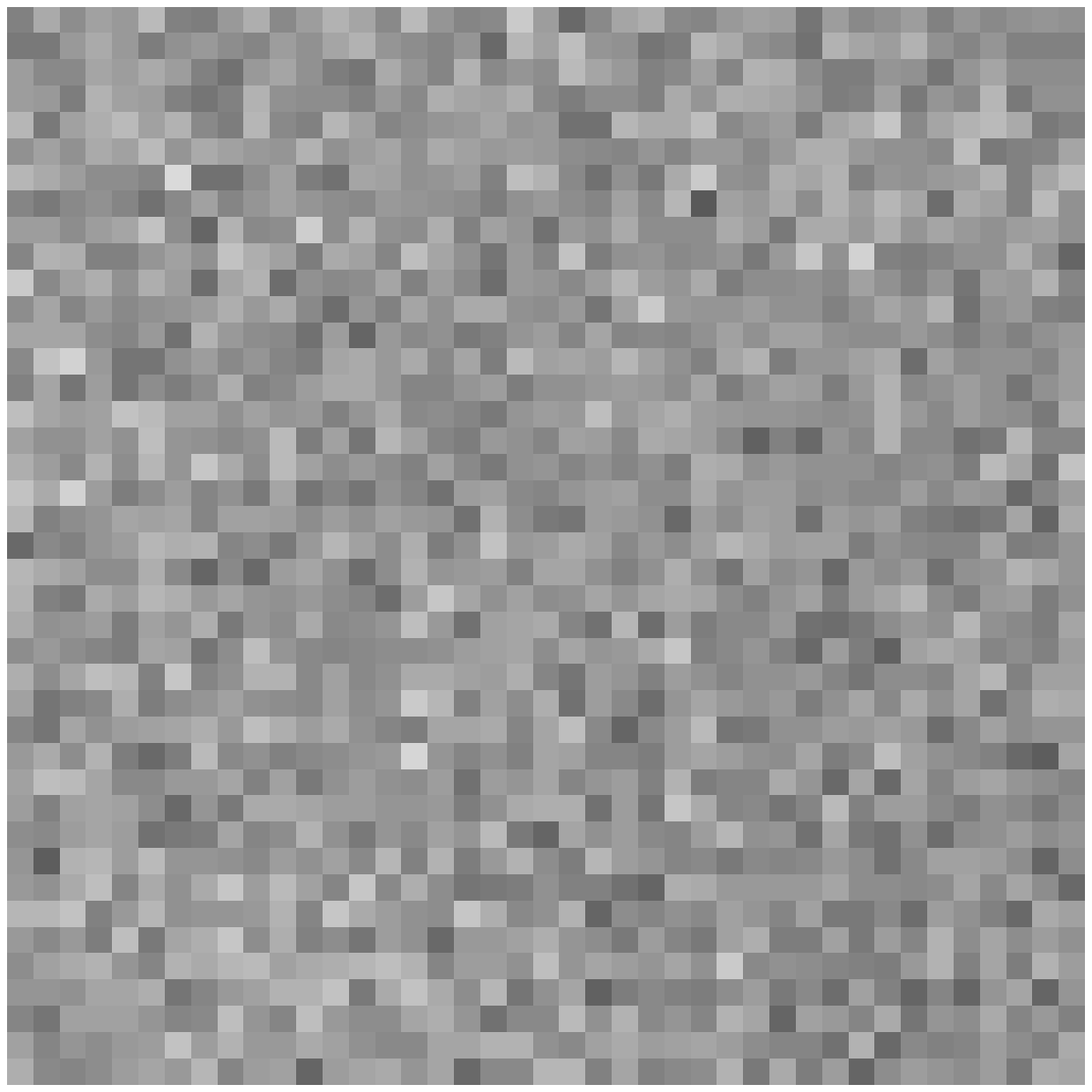} & %
\includegraphics[width=0.18\linewidth]{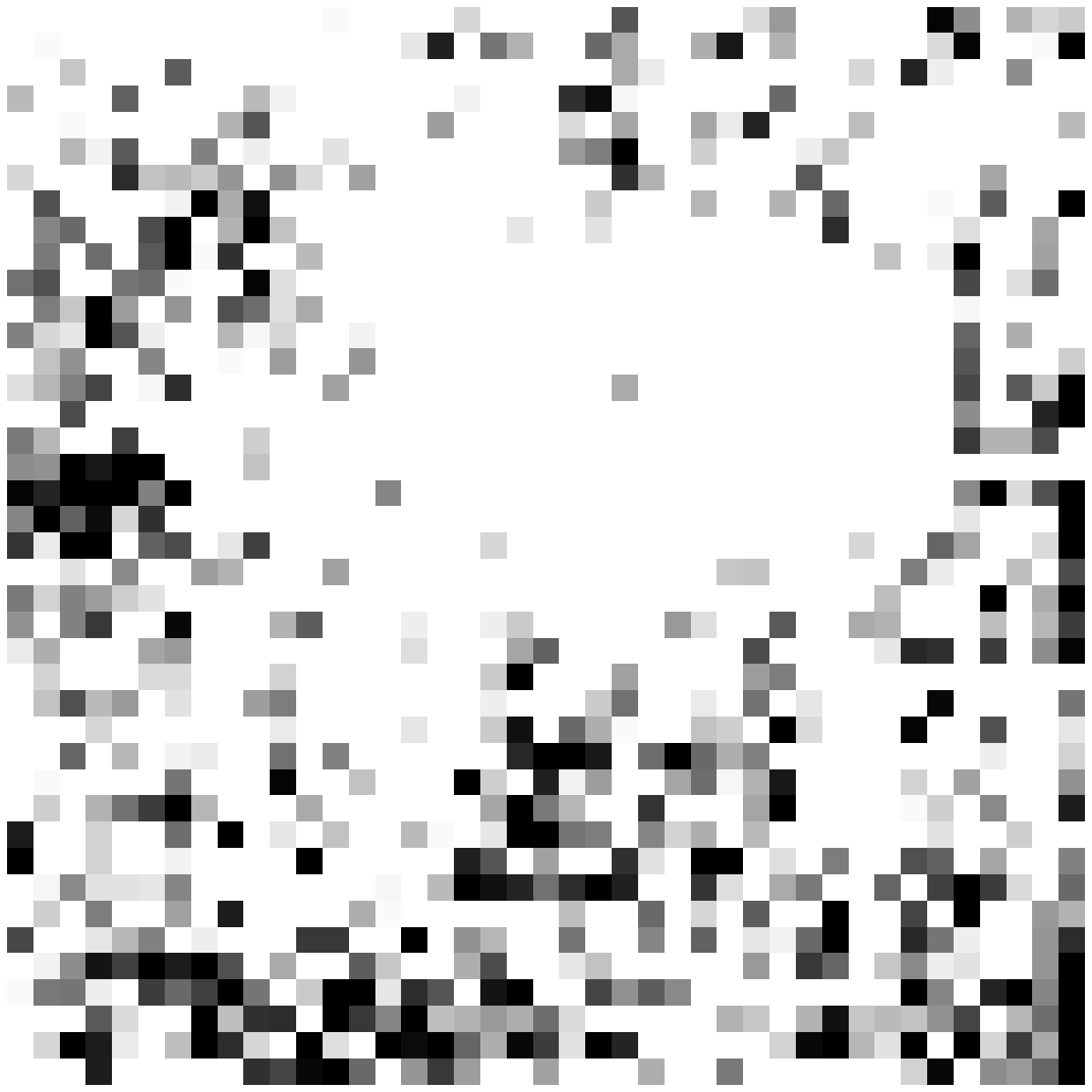} & %
\includegraphics[width=0.18\linewidth]{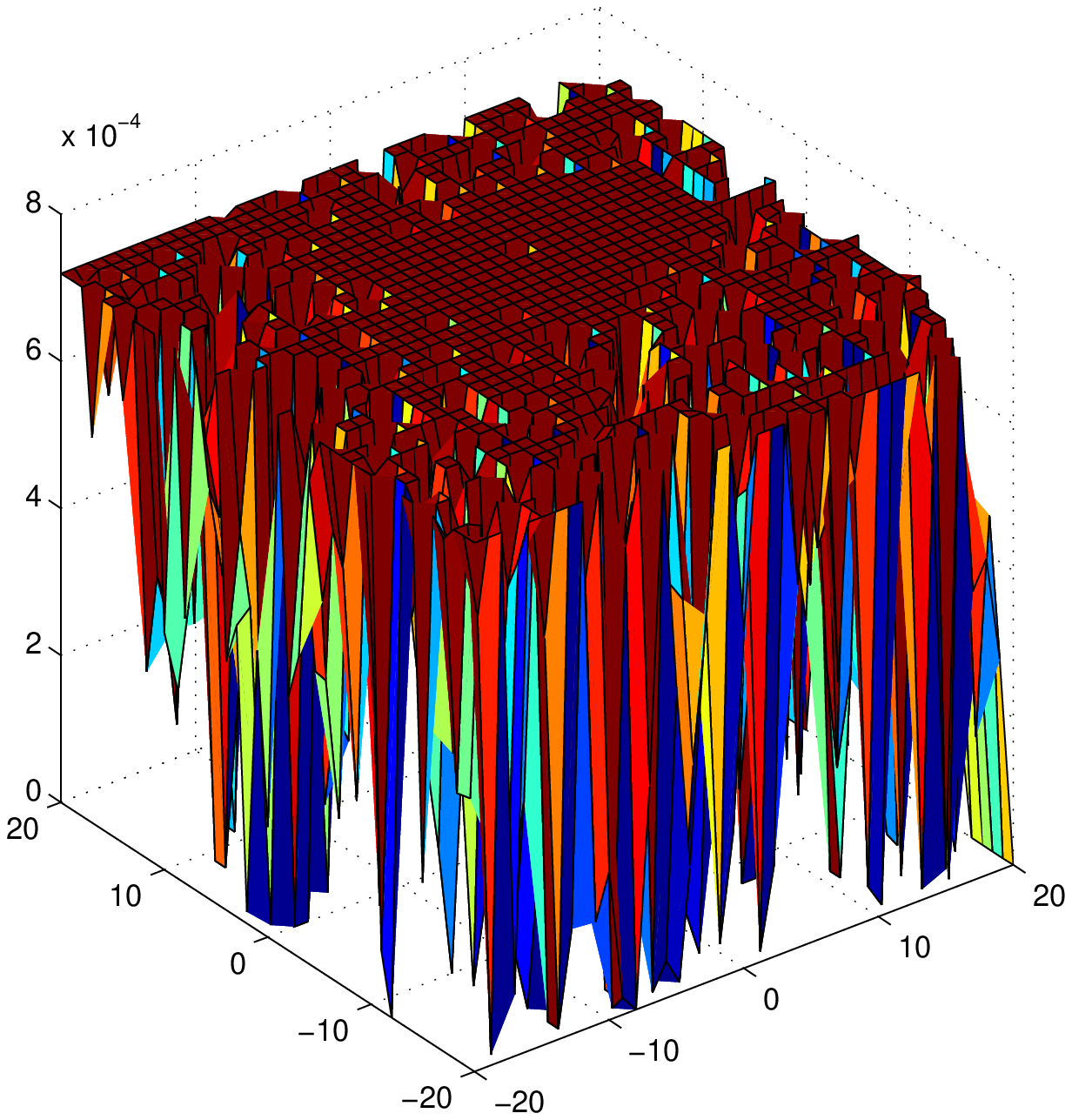} & %
\includegraphics[width=0.18\linewidth]{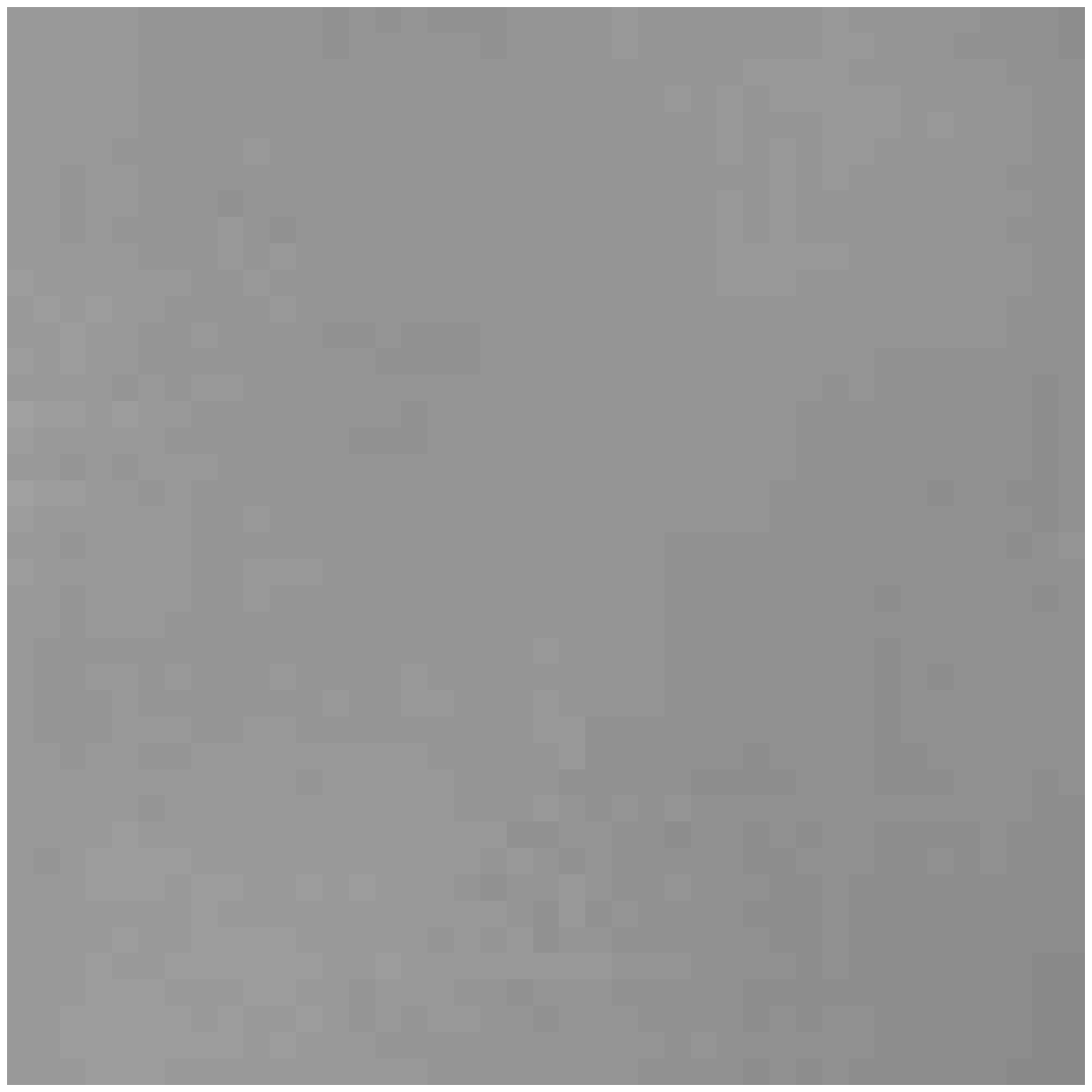}%
\end{tabular}
} \vskip1mm
\par
\rule{0pt}{-0.2pt}%
\par
\vskip1mm 
\end{center}
\caption{\small These pictures show how the Optimal Weights Filter detects the
features of the image by choosing  appropriate weights. The first column
displays six selected search windows used to estimate the image at the
corresponding central pixels a, b, c, d, e and f. The second column displays the corresponding
search windows corrupted by a Gaussian noise with standard deviation $%
\protect\sigma =20.$ The third column displays the two-dimensional representation of the weights used to estimate
central pixels. The fourth column gives the three-dimensional representation
of the weights. The fifth column gives the restored images.}
\label{Fig weights}
\end{figure}

The best numerical results are obtained using $K=K_{g}$ and $K=K_{0}$ in the
definition of $\widehat{\rho }_{K,x_0}.$ In Table\ \ref{Table compar},  we compare the Non-Local Mean Filter and the Optimal Weights filter with different choices of the kernel: $K=K_{g},K_{0},K_{r}.$ The best PSNR values we obtained by varying
the size $m$ of the similarity windows and the size $M$ of the search windows are
reported in Tables\ \ref{table11} ($\sigma =10$), \ref{table12} ($\sigma =20$%
) and \ref{table13} ($\sigma =30$) for $K=K_{0}.$ Note that the PSNR values
are close for every $m$ and $M$ and the optimal $m$ and $M$ depend on the
image content. The values $m=21\times 21$ and $M=13\times 13$ seem
appropriate in most cases and a smaller patch size $m$ can be considered for
processing piecewise smooth images.

\begin{table}[tbp]
\begin{center}
\renewcommand{\arraystretch}{0.6} \vskip3mm {\fontsize{8pt}{\baselineskip}%
\selectfont
\begin{tabular}{c|ccccc}
\hline
Images & Lena & Barbara & Boat & House & Peppers \\
Sizes & $512 \times 512$ & $512 \times 512$ & $512 \times 512$ & $256 \times
256$ & $256 \times 256$ \\ \hline\hline
$\sigma /PSNR$ & 10/28.12db & 10/28.12db & 10/28.12db & 10/28.11db &
10/28.11db \\ \hline
OWF with $K_{r}$ & 35.23db & 33.89db & 33.07db & 35.57db & 33.74db \\
OWF with $K_{g}$ & 35.49db & 34.13db & 33.40db & 35.83db & 33.97db \\
OWF with $K_{0}$ & 35.52db & 34.10db & 33.48db & 35.80db & 33.96db \\
NLMF & 35.03db & 33.77db & 32.85db & 35.43db & 33.27db \\ \hline\hline
$\sigma /PSNR$ & 20/22.11db & 20/22.11db & 20/22.11db & 20/28.12db &
20/28.12db \\ \hline
OWF with $K_{r}$ & 32.24db & 30.71db & 29.65db & 32.59db & 30.17db \\
OWF with $K_{g}$ & 32.61db & 31.01db & 30.05db & 32.88db & 30.44db \\
OWF with $K_{0}$ & 32.52db & 31.00db & 30.20db & 32.90db & 30.66db \\
NLMF & 31.73db & 30.36db & 29.58db & 32.51db & 30.11db \\ \hline\hline
$\sigma /PSNR$ & 30/18.60db & 30/18.60db & 30/18.60db & 30/18.61db &
30/18.61db \\ \hline
OWF with $K_{r}$ & 30.26db & 28.59db & 27.69db & 30.49db & 27.93db \\
OWF with $K_{g}$ & 30.66db & 28.97db & 28.05db & 30.81db & 28.16db \\
OWF with $K_{0}$ & 30.50db & 28.89db & 28.23db & 30.80db & 28.49db \\
NLMF & 29.56db & 27.88db & 27.50db & 30.02db & 27.77db \\ \hline
\end{tabular}
} \vskip1mm
\end{center}
\caption{
\small Comparison between the Non-Local Means Filter (NLMF) and the
Optimal Weights Filter (OWF). }
\label{Table compar}
\end{table}

\begin{table}[tbp]
\begin{center}
\renewcommand{\arraystretch}{0.6} \vskip3mm {\fontsize{8pt}{\baselineskip}%
\selectfont
\begin{tabular}{c|ccccc}
\hline
$\sigma = 10$ & Lena & Barbara & Boat & House & Peppers \\
$m/M$ & $512 \times 512$ & $512 \times 512$ & $512 \times 512$ & $256 \times
256$ & $256 \times 256$ \\ \hline
$11 \times 11/11 \times 11$ & 35.35db & 34.03db & 33.43db & 35.69db & 34.16db
\\
$13 \times 13/11 \times 11 $ & 35.40db & 34.06db & 33.45db & 35.72db &
34.14db \\
$15 \times 15/11 \times 11 $ & 35.44db & 34.07db & 33.47db & 35.73db &
34.10db \\
$17 \times 17/11 \times 11 $ & 35.47db & 34.08db & 33.47db & 35.74db &
34.06db \\
$19 \times 19/11 \times 11 $ & 35.50db & 34.07db & 33.48db & 35.74db &
34.02db \\
$21 \times 21/11 \times 11 $ & 35.52db & 34.06db & 33.47db & 35.73db &
33.97db \\ \hline
$11 \times 11/13 \times 13 $ & 35.35db & 34.08db & 33.43db & 35.77db &
34.15db \\
$13 \times 13/13 \times 13 $ & 35.40db & 34.11db & 33.46db & 35.79db &
34.12db \\
$15 \times 15/13 \times 13 $ & 35.44db & 34.12db & 33.47db & 35.80db &
34.09db \\
$17 \times 17/13 \times 13 $ & 35.47db & 34.12db & 33.48db & 35.81db &
34.05db \\
$19 \times 19/13 \times 13 $ & 35.50db & 34.12db & 33.48db & 35.81db &
34.01db \\
$21 \times 21/13 \times 13 $ & 35.52db & 34.10db & 33.48db & 35.80db &
33.96db \\ \hline
$11 \times 11/15 \times 15 $ & 35.33db & 34.11db & 33.43db & 35.82db &
34.14db \\
$13 \times 13/15 \times 15 $ & 35.39db & 34.13db & 33.45db & 35.84db &
34.11db \\
$15 \times 15/15 \times 15 $ & 35.43db & 34.14db & 33.47db & 35.85db &
34.08db \\
$17 \times 17/15 \times 15 $ & 35.47db & 34.14db & 33.48db & 35.86db &
34.04db \\
$19 \times 19/15 \times 15 $ & 35.49db & 34.14db & 33.48db & 35.85db &
34.00db \\
$21 \times 21/15 \times 15 $ & 35.52db & 34.12db & 33.48db & 35.84db &
33.96db \\ \hline
$11 \times 11/17 \times 17 $ & 35.32db & 34.13db & 33.42db & 35.86db &
34.12db \\
$13 \times 13/17 \times 17 $ & 35.37db & 34.15db & 33.44db & 35.88db &
34.10db \\
$15 \times 15/17 \times 17 $ & 35.42db & 34.16db & 33.46db & 35.89db &
34.07db \\
$17 \times 17/17 \times 17 $ & 35.46db & 34.16db & 33.47db & 35.89db &
34.03db \\
$19 \times 19/17 \times 17 $ & 35.48db & 34.15db & 33.47db & 35.88db &
34.00db \\
$21 \times 21/17 \times 17 $ & 35.51db & 34.14db & 33.47db & 35.87db &
33.95db \\ \hline
\end{tabular}
} \vskip1mm
\end{center}
\caption{\small PSNR values when Optimal Weights Filter with $K=K_0$ is applied
with different values of $m$ and $M$ ($\protect\sigma=10$).}
\label{table11}
\end{table}

\begin{table}[tbp]
\begin{center}
\renewcommand{\arraystretch}{0.6} \vskip3mm {\fontsize{8pt}{\baselineskip}%
\selectfont
\begin{tabular}{c|ccccc}
\hline
$\sigma = 20$ & Lena & Barbara & Boat & House & Peppers \\
$m/M$ & $512 \times 512$ & $512 \times 512$ & $512 \times 512$ & $256 \times
256$ & $256 \times 256$ \\ \hline
$11 \times 11/11 \times 11 $ & 32.08db & 30.60db & 30.00db & 32.56db &
30.65db \\
$13 \times 13/11 \times 11 $ & 32.20db & 30.70db & 30.06db & 32.64db &
30.68db \\
$15 \times 15/11 \times 11 $ & 32.30db & 30.78db & 30.11db & 32.71db &
30.70db \\
$17 \times 17/11 \times 11 $ & 32.39db & 30.84db & 30.15db & 32.76db &
30.70db \\
$19 \times 19/11 \times 11 $ & 32.47db & 30.88db & 30.18db & 32.79db &
30.70db \\
$21 \times 21/11 \times 11 $ & 32.53db & 30.91db & 30.21db & 32.81db &
30.69db \\ \hline
$11 \times 11/13 \times 13 $ & 32.06db & 30.67db & 29.99db & 32.63db &
30.61db \\
$13 \times 13/13 \times 13 $ & 32.18db & 30.78db & 30.05db & 32.71db &
30.64db \\
$15 \times 15/13 \times 13 $ & 32.29db & 30.86db & 30.10db & 32.79db &
30.66db \\
$17 \times 17/13 \times 13 $ & 32.38db & 30.92db & 30.14db & 32.84db &
30.67db \\
$19 \times 19/13 \times 13 $ & 32.46db & 30.97db & 30.18db & 32.88db &
30.67db \\
$21 \times 21/13 \times 13 $ & 32.52db & 31.00db & 30.20db & 32.90db &
30.66db \\ \hline
$11 \times 11/15 \times 15 $ & 32.02db & 30.71db & 29.97db & 32.67db &
30.56db \\
$13 \times 13/15 \times 15 $ & 32.15db & 30.82db & 30.03db & 32.76db &
30.59db \\
$15 \times 15/15 \times 15 $ & 32.26db & 30.90db & 30.08db & 32.83db &
30.62db \\
$17 \times 17/15 \times 15 $ & 32.35db & 30.96db & 30.12db & 32.89db &
30.63db \\
$19 \times 19/15 \times 15 $ & 32.43db & 31.01db & 30.16db & 32.92db &
30.63db \\
$21 \times 21/15 \times 15 $ & 32.50db & 31.04db & 30.19db & 32.94db &
30.63db \\ \hline
$11 \times 11/17 \times 17 $ & 31.97db & 30.72db & 29.94db & 32.70db &
30.52db \\
$13 \times 13/17 \times 17 $ & 32.10db & 30.83db & 30.00db & 32.79db &
30.56db \\
$15 \times 15/17 \times 17 $ & 32.22db & 30.92db & 30.05db & 32.86db &
30.58db \\
$17 \times 17/17 \times 17 $ & 32.32db & 30.98db & 30.10db & 32.92db &
30.59db \\
$19 \times 19/17 \times 17 $ & 32.40db & 31.02db & 30.13db & 32.96db &
30.60db \\
$21 \times 21/17 \times 17 $ & 32.47db & 31.06db & 30.17db & 32.98db &
30.60db \\ \hline
\end{tabular}
} \vskip1mm
\end{center}
\caption{\small PSNR values when Optimal Weights Filter with $K=K_0$ is applied
with different values of $m$ and $M$ ($\protect\sigma=20$).}
\label{table12}
\end{table}

\begin{table}[tbp]
\begin{center}
\renewcommand{\arraystretch}{0.6} \vskip3mm {\fontsize{8pt}{\baselineskip}%
\selectfont
\begin{tabular}{c|ccccc}
\hline
$\sigma = 30$ & Lena & Barbara & Boat & House & Peppers \\
$m/M$ & $512 \times 512$ & $512 \times 512$ & $512 \times 512$ & $256 \times
256$ & $256 \times 256$ \\ \hline
$11 \times 11/11 \times 11 $ & 29.96db & 28.38db & 27.96db & 30.26db &
28.36db \\
$13 \times 13/11 \times 11 $ & 30.10db & 28.53db & 28.03db & 30.39db &
28.43db \\
$15 \times 15/11 \times 11 $ & 30.23db & 28.65db & 28.10db & 30.50db &
28.47db \\
$17 \times 17/11 \times 11 $ & 30.34db & 28.75db & 28.15db & 30.58db &
28.50db \\
$19 \times 19/11 \times 11 $ & 30.43db & 28.83db & 28.20db & 30.65db &
28.51db \\
$21 \times 21/11 \times 11$ & 30.50db & 28.81db & 28.23db & 30.70db & 28.52db
\\ \hline
$11 \times 11/13 \times 13 $ & 29.94db & 28.42db & 27.95db & 30.35db &
28.30db \\
$13 \times 13/13 \times 13 $ & 30.08db & 28.58db & 28.02db & 30.49db &
28.37db \\
$15 \times 15/13 \times 13 $ & 30.21db & 28.70db & 28.09db & 30.60db &
28.42db \\
$17 \times 17/13 \times 13 $ & 30.32db & 28.80db & 28.14db & 30.68db &
28.46db \\
$19 \times 19/13 \times 13 $ & 30.42db & 28.88db & 28.19db & 30.75db &
28.48db \\
$21 \times 21/13 \times 13$ & 30.50db & 28.89db & 28.23db & 30.80db & 28.49db
\\ \hline
$11 \times 11/15 \times 15 $ & 29.89db & 28.43db & 27.92db & 30.39db &
28.23db \\
$13 \times 13/15 \times 15 $ & 30.04db & 28.58db & 27.99db & 30.53db &
28.30db \\
$15 \times 15/15 \times 15 $ & 30.17db & 28.71db & 28.06db & 30.64db &
28.36db \\
$17 \times 17/15 \times 15 $ & 30.28db & 28.81db & 28.11db & 30.73db &
28.40db \\
$19 \times 19/15 \times 15 $ & 30.38db & 28.89db & 28.16db & 30.80db &
28.43db \\ \hline
$11 \times 11/17 \times 17 $ & 29.82db & 28.40db & 27.89db & 30.39db &
28.18db \\
$13 \times 13/17 \times 17 $ & 29.98db & 28.56db & 27.96db & 30.54db &
28.26db \\
$15 \times 15/17 \times 17 $ & 30.11db & 28.69db & 28.02db & 30.66db &
28.31db \\
$17 \times 17/17 \times 17 $ & 30.22db & 28.79db & 28.08db & 30.76db &
28.36db \\
$19 \times 19/17 \times 17 $ & 30.33db & 28.87db & 28.13db & 30.84db &
28.39db \\
$21 \times 21/17 \times 17 $ & 30.42db & 28.96db & 28.17db & 30.89db &
28.41db \\ \hline
\end{tabular}
} \vskip1mm
\end{center}
\caption{\small PSNR values when Optimal Weights Filter with $K=K_0$ is applied
with different values of $m$ and $M$ ($\protect\sigma=30$).}
\label{table13}
\end{table}

\section{\label{Sec:Appendix Proofs} Proofs of the main results}

\subsection{\label{Sec: proof of Th weights 001} Proof of Theorem \protect
\ref{Th weights 001}}

We begin with some preliminary results. The following lemma can be obtained from Theorem 1 of Sacks and Ylvisaker \citep{Sacks1978linear}. For the convenience of readers, we prefer to give a direct proof adapted to our situation.

\begin{lemma}
\label{Lemma weights}Let $g_{\rho}(w)$ be defined by (\ref{def gw}). Then there are
unique weights $w_{\rho}$ which minimize $g_{\rho}(w)$ subject to (\ref{s2wx}),
given by
\begin{equation}
w_{\rho}(x)=\frac{1}{\sigma ^{2}}(b-\lambda \rho (x))^{+},
\label{eq Lemma W000}
\end{equation}%
where $b$ and $\lambda $ are determined by%
\begin{eqnarray}
\sum_{x\in \mathbf{U}_{x_{0},h}}\frac{1}{\sigma ^{2}}(b-\lambda \rho
(x))^{+} &=&1,  \label{eq Lemma W001} \\
\sum_{x\in \mathbf{U}_{x_{0},h}}\frac{1}{\sigma ^{2}}(b-\lambda \rho
(x))^{+}\rho (x) &=&\lambda .  \label{eq Lemma W002}
\end{eqnarray}
\end{lemma}

\begin{proof}
Let $w^{\prime }$ be a minimizer of $g_{\rho}\left( w\right) $ under the constraint
(\ref{s2wx}). According to Theorem 3.9 of Whittle (1971 \citep{Wh}), there
are Lagrange multipliers $b\geq 0$ and $b_{0}(x)\geq 0,$ $x\in \mathbf{U}%
_{x_{0},h},$ such that the function%
\begin{equation*}
G(w)=g_{\rho}(w)-2b(\sum_{x\in \mathbf{U}_{x_{0},h}}w(x)-1)-2\sum_{x\in \mathbf{U}%
_{x_{0},h}}b_{0}(x)w(x)
\end{equation*}%
is minimized at the same point $w^{\prime }.$ Since the function $G$ is
strictly convex it admits a unique point of minimum. This implies that
there is also a unique minimizer of $g_{\rho}\left( w\right) $ under the constraint
(\ref{s2wx}) which coincides with the unique minimizer of $G.$

Let $w_{\rho}$ be the unique minimizer of $G$ satisfying the constraint (%
\ref{s2wx}). Again, using the fact that $G$ is strictly convex, for any $%
x\in \mathbf{U}_{x_{0},h},$
\begin{equation}
\frac{\partial }{\partial w\left( x\right) }G\left( w\right)\bigg|_{w=w_{\rho}}
=2\left( \sum_{y\in \mathbf{U}_{x_{0},h}}w_{\rho}(y)\rho (y)\right) \rho
(x)+2\sigma ^{2}w_{\rho}(x)-2b-2b_{0}(x)\geq 0.  \label{s5kw1}
\end{equation}%
Note that in general we do not have an equality in (\ref{s5kw1}). In
addition, by the Karush-Kuhn-Tucker condition,
\begin{equation}
b_{0}(x)w_{\rho}(x)=0.  \label{s5bx}
\end{equation}
\par
Let
\begin{equation}
\lambda =\sum_{y\in \mathbf{U}_{x_{0},h}}w_{\rho}(y)\rho (y).
\label{lambda def}
\end{equation}%
Then (\ref{s5kw1}) becomes
\begin{equation}
\frac{\partial }{\partial w\left( x\right) }G\left( w\right)\bigg|_{w=w_{\rho}}
=\lambda \rho (x)+\sigma ^{2}w_{\rho}(x)-b-b_{0}(x)\geq 0,\quad x\in \mathbf{U}%
_{x_{0},h}.  \label{s5lr}
\end{equation}

If $b_{0}(x)=0,$ then, with respect to  the single variable $w(x)$ the function $G(w)$
attains its minimum at an interior point $w_{\rho}\left( x\right) \geq 0$,
so that we have
\begin{equation*}
\frac{\partial }{\partial w\left( x\right) }G\left( w\right)\bigg|_{w=w_{\rho}}
=\lambda \rho (x)+\sigma ^{2}w_{\rho}(x)-b=0.
\end{equation*}%
From this we obtain $b-\lambda \rho (x)=\sigma w_{\rho}(x)\geq 0$, so
\begin{equation*}
w_{\rho}(x)=\frac{(b-\lambda \rho (x))^{+}}{\sigma }.
\end{equation*}

If $b_{0}(x)>0$, by (\ref{s5bx}), we have $w_{\rho}(x)=0$. Consequently, from (%
\ref{s5lr}) we have
\begin{equation}
b-\lambda \rho (x)\leq -b_{0}(x)\leq 0,  \label{s5bl2}
\end{equation}
so that we get again%
\begin{equation*}
w_{\rho}(x)=0=\frac{(b-\lambda \rho (x))^{+}}{\sigma }.
\end{equation*}%
As to the conditions (\ref{eq Lemma W001}) and (\ref{eq Lemma W002}), they
follow immediately from  the constraint (\ref{s2wx}) and  the equation (%
\ref{lambda def}).
\end{proof}
\\
\\
{ \bf Proof of Theorem \ref{Th weights 001}}. Applying Lemma \ref%
{Lemma weights} with $b=\lambda a $, we see that the unique optimal weights $%
w$ minimizing $g_{\rho}(w)$ subject to (\ref{s2wx}), are given by
\begin{equation}
w_{\rho}=\frac{\lambda }{\sigma ^{2}}(a-\rho (x))^{+},  \label{s5wl}
\end{equation}%
where $a$ and $\lambda $ satisfy%
\begin{equation}
\lambda \sum_{x\in \mathbf{U}_{x_{0},h}}(a-\rho (x))^{+}=\sigma ^{2}
\label{eq proof  Th w 001}
\end{equation}%
and%
\begin{equation}
\sum_{x\in \mathbf{U}_{x_{0},h}}(a-\rho (x))^{+}\rho (x)=\sigma ^{2}.
\label{equation of a}
\end{equation}%
Since the function
\begin{equation*}
M_{\rho }\left( t\right) =\sum_{x\in \mathbf{U}_{x_{0},h}}(t-\rho (x))^{+}\rho (x)
\end{equation*}%
is strictly increasing and continuous with $M_{\rho }\left( 0\right) =0$ and $\lim\limits_{t%
\rightarrow \infty }M_{\rho }\left( t\right) =+\infty ,$ the equation
\begin{equation*}
M_{\rho }\left( a\right) =\sigma ^{2}
\end{equation*}%
has a unique solution on $(0,\infty)$. By (\ref{eq proof Th w 001}),
\begin{equation*}
\frac{\sigma ^{2}}{\lambda }=\sum_{x\in \mathbf{U}_{x_{0},h}}(a-\rho
(x))^{+},
\end{equation*}%
which together with (\ref{s5wl}) imply (\ref{eq th weights 001}) and (\ref%
{eq th weights 002}).

\subsection{\label{Sec: proof of remark}Proof of Remark \protect\ref%
{calculate a}}

Expression (\ref{def mt}) can be rewritten as
\begin{equation}
M_{\rho}(t)=\sum_{i=1}^{M}\rho_i(t-\rho_i)^+.
\label{def remt}
\end{equation}
Since  function  $M_{\rho}(t)$ is strictly increasing with  $M_{\rho}(0)=0$ and $M_{\rho}(+\infty)=+\infty$, equation (\ref{eq th weights 002}) admits a unique solution $a$ on $(0,+\infty)$, which must be located in some interval $[\rho_{k_0},\rho_{k_0+1})$, $1\leq k_0\leq M$, where $\rho_{M+1}=\infty$ (see Figure \ref{Fig rho}).  Hence the equation (\ref{eq th weights 002}) becomes
\begin{equation}
\sum_{i=1}^{k_0}\rho_i(a-\rho_i)=\sigma^2,
\label{function M}
\end{equation}
where $ \quad \rho_{k_0}\leq a <\rho_{k_0+1}$.
From (\ref{function M}), it follows that
\begin{equation}
a=\frac{\sigma^2+\sum\limits_{i=1}^{k_0}\rho_i^2}
{\sum\limits_{i=1}^{k_0}\rho_i}, \quad \rho_{k_0}\leq a <\rho_{k_0+1}.
\label{solution mt}
\end{equation}
\par
We now show that $k_0=k^*$ (so that $a=k_0=k^*$), where $k^*:=\max\{1\leq k \leq M\, |\, a_k\geq \rho_{k} \}$. To this end, it suffices to verify that
$a_{k_0}\geq \rho_{k_0}$ and $a_k < \rho_k$ if $k_0<k\leq M$. We have already seen that $a_{k_0}\geq \rho_{k_0}$; if $k_0<k\leq M$, then $a_{k_0}<\rho_{k_0+1}\leq \rho_k$, so that
\begin{equation}
a_{k}=\frac{(\sigma^2+\sum\limits_{i=1}^{k_0}\rho_i^2)
+\sum\limits_{i=k_0+1}^{k}\rho_i^2}
{\sum\limits_{i=1}^{k}\rho_i}
=\frac{a_{k_0}\sum\limits_{i=1}^{k_0}\rho_i+\sum\limits_{k_0+1}^{k}\rho_i^2}
{\sum\limits_{i=1}^{k}\rho_i}
<\frac{\rho_{k}\sum\limits_{i=1}^{k_0}\rho_i+\sum\limits_{k_0+1}^{k}\rho_{k}\rho_i}
{\sum\limits_{i=1}^{k}\rho_i}=\rho_{k}.
\end{equation}
\par
We finally prove that if $1\leq k <M$ and $a_k<\rho_k$, then $a_{k+1}<\rho_{k+1}$, so that the last equality in (\ref{k star}) holds and that $k^*$ is the unique integer $k\in \{1,\cdots,M\}$ such that $a_k \geq \rho_k$ and $a_{k+1}<\rho_{k+1}$ if $1\leq k<M$.
 In fact, for $1\leq k<M$, the inequality $a_k<\rho_k$ implies that
 $$
 \sigma^2 + \sum_{i=1}^{k}\rho_i^2<\rho_k\sum_{i=1}^{k}\rho_i.
 $$
This, in turn, implies that
\begin{equation*}
a_{k+1}
=\frac{\sigma^2+\sum\limits_{i=1}^{k}\rho_i^2+\rho_{k+1}^2}
{\sum\limits_{i=1}^{k+1}\rho_i}
< \frac{\rho_k\sum\limits_{i=1}^{k}\rho_i+\rho_{k+1}^2}
{\sum\limits_{i=1}^{k+1}\rho_i}\leq\rho_{k+1}.
\end{equation*}
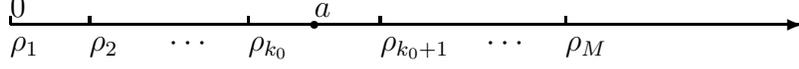
\begin{figure}
\begin{center}
\begin{picture}(0,0)
\thicklines
\put(-150,0){\vector(1, 0){300}}
\put(-150,0){\line(0, 1){3}}
\put(-150,3){0}
\put(-150,-10){$\rho_1$}
\put(-120,0){\line(0, 1){3}}
\put(-120,-10){$\rho_2$}
\put(-90,-10){$\cdots$}
\put(-60,0){\line(0, 1){3}}
\put(-60,-10){$\rho_{k_0}$}
\put(-35,0){\circle*{2.5}}
\put(-35,3){$a$}
\put(-10,0){\line(0, 1){3}}
\put(-10,-10){$\rho_{k_0+1}$}
\put(30,-10){$\cdots$}
\put(60,0){\line(0, 1){3}}
\put(60,-10){$\rho_{M}$}
\end{picture}
\end{center}
\caption{\small The  number axis of $\rho_i$, $i=1,2,\cdots,M$.}
\label{Fig rho}
\end{figure}

\subsection{\label{Sec: proof of Th oracle 001}Proof of Theorem \protect\ref%
{Th oracle 001}}

First assume that $\rho \left( x\right) =\rho_{f,x_0}(x)=|f(x)-f(x_0)|.$ Recall that $g_{\rho}$ and $w_{\rho}$ were defined by (\ref{def gw}) and (\ref{eq th weights 001}).  Using H\"{o}lder's condition (\ref%
{Local Holder cond}) we have, for any $w$,
\begin{equation*}
g_{\rho}(w_{\rho})\leq g_{\rho}(w)\leq \overline{g}(w),
\end{equation*}%
where
\begin{equation*}
\overline{g}(w)=\left( \sum_{x\in \mathbf{U}_{x_{0},h}}w(x)L\Vert
x-x_{0}\Vert _{\infty }^{\beta }\right) ^{2}+\sigma ^{2}\sum_{x\in \mathbf{U}%
_{x_{0},h}}w^{2}(x).
\end{equation*}%
In particular, denoting $\overline{w}=\arg \min_{w}\overline{g}(w),$ we get%
\begin{equation*}
g_{\rho}(w_{\rho})\leq \overline{g}(\overline{w}).
\end{equation*}%
By Theorem \ref{Th weights 001},
\begin{equation*}
\overline{w}(x)=\left( a-L\Vert x-x_{0}\Vert _{\infty }^{\beta }\right) ^{+}%
\Big/\sum\limits_{y\in \mathbf{U}_{x_{0},h}}\left( a-L\Vert x-x_{0}\Vert
_{\infty }^{\beta }\right) ^{+},
\end{equation*}%
where $a>0$ is the unique solution on $(0,\infty)$ of the equation $\overline{M}_h(a)=\sigma^2$, with
\begin{equation*}
\overline{M}_{h}\left( t\right) =\sum_{x\in \mathbf{U}_{x_{0},h}}L\Vert
x-x_{0}\Vert _{\infty }^{\beta }(t-L\Vert x-x_{0}\Vert _{\infty }^{\beta
})^{+}, \quad t\geq 0.
\end{equation*}%
Theorem \ref{Th oracle 001} will be  a consequence of the following lemma.

\begin{lemma}
\label{lm5_2}Assume that $\rho(x)=L\|x-x_0\|_{\infty}^{\beta}$ and that $h\geq c_{1}n^{-\alpha }$ with $0\leq \alpha <\frac{%
1}{2\beta +2}$, or $h=c_{1}n^{-\frac{1}{2\beta +2}}$ with $c_{1}>c_{0}=\left( \frac{\sigma ^{2}(2\beta +2)(\beta +2)}{8L^{2}%
}\right) ^{\frac{1}{2\beta +2}}.$ Then
\begin{equation}
a=c_{3}n^{-\beta /(2\beta +2)}(1+o(1))  \label{s5hk}
\end{equation}%
and
\begin{equation}
\overline{g}(\overline{w})\leq c_{4}n^{-\frac{2\beta }{2+2\beta }},
\label{s5gw4}
\end{equation}%
where $c_{3}$ and $c_{4}$ are positive constants depending only on $\beta ,$ $L$ and $%
\sigma .$
\end{lemma}

\begin{proof}
 We first prove (\ref{s5hk}) in the case where $h=1$, i.e. $\mathbf{U}_{x_0,h}=\mathbf{I}$.  Then by the
definition of $a,$ we have
\begin{equation}
\overline{M}_{1}\left( a\right) =\sum\limits_{x\in I}(a-L\Vert x-x_{0}\Vert _{\infty
}^{\beta })^{+}L\Vert x-x_{0}\Vert _{\infty }^{\beta }=\sigma ^{2}.
\label{eq - M1 001}
\end{equation}%
Let $\overline{h}=\left( a/L\right) ^{1/\beta }$. Then $a-L\|x-x_0\|_{\infty}^{\beta}\geq 0$ if and only if $\|x-x_0\|_{\infty}\leq \overline{h}$. So from (\ref{eq - M1 001}) we get%
\begin{equation}
L^{2}\overline{h}^{\beta }\sum_{\Vert x-x_{0}\Vert _{\infty }\leq \overline{h%
}}\Vert x-x_{0}\Vert _{\infty }^{\beta }-L^{2}\sum_{\Vert x-x_{0}\Vert
_{\infty }\leq \overline{h}}\Vert x-x_{0}\Vert _{\infty }^{2\beta }=\sigma
^{2}.  \label{eq - M1 002}
\end{equation}%
By the definition of  the neighborhood $U_{x_{0},\overline{h}}$ it
is easily seen that%
\begin{equation*}
\sum_{\Vert x-x_{0}\Vert _{\infty }\leq \overline{h}}\Vert x-x_{0}\Vert
_{\infty }^{\beta }=8N^{-\beta }\sum_{k=1}^{N\overline{h}}k^{\beta +1}=8N^{2}%
\frac{\overline{h}^{\beta +2}}{\beta +2}\left( 1+o\left( 1\right) \right)
\end{equation*}%
and%
\begin{equation*}
\sum_{\Vert x-x_{0}\Vert _{\infty }\leq \overline{h}}\Vert x-x_{0}\Vert
_{\infty }^{2\beta }=8N^{-2\beta }\sum_{k=1}^{N\overline{h}}k^{2\beta
+1}=8N^{2}\frac{\overline{h}^{2\beta +2}}{2\beta +2}\left( 1+o\left(
1\right) \right) .
\end{equation*}%
Therefore, (\ref{eq - M1 002}) implies
\begin{equation*}
\frac{8L^{2}\beta }{\left( \beta +2\right) \left( 2\beta +2\right) }N^{2}%
\overline{h}^{2\beta +2}(1+o{(1)})=\sigma ^{2},
\end{equation*}%
from which we infer that%
\begin{equation}
\overline{h}=c_{0}n^{-\frac{1}{2\beta +2}}(1+o(1))  \label{eq - h bar}
\end{equation}%
with $c_{0}=\left( \frac{\sigma ^{2}\left( \beta +2\right) \left( 2\beta
+2\right) }{8L^{2}\beta }\right) ^{\frac{1}{2\beta +2}}.$ From (\ref{eq - h bar}) and the definition of $\overline{h}$, we obtain%
\begin{equation*}
a=L\overline{h}^{\beta }=Lc_{0}^{\beta }n^{-\frac{\beta }{2\beta +2}%
}(1+o(1)),
\end{equation*}
which prove (\ref{s5hk}) in the case when $h=1$.
\par
We next prove  (\ref{s5hk})  under the conditions of the lemma. If $h\geq c_{0}n^{-\alpha },$ where $0\leq \alpha <\frac{1}{2\beta +2},$
then it is clear that $h\geq \overline{h}$ for $n$ sufficiently large.
Therefore $\overline{M}_{h}\left( a\right) =\overline{M}_{1}\left( a\right) $, thus we arrive at equation (\ref{eq - M1 001}) from which we deduce (\ref{eq - h bar}). If
$h\geq c_{0}n^{-\frac{1}{2\beta +2}}$ and $c_{0}>c_{1},$ then again $h\geq
\overline{h}$ for $n$ sufficiently large. Therefore $\overline{M}_{h}\left( a\right)
=\overline{M}_{1}\left( a\right) $, and we arrive again at (\ref{eq - h bar}).
\par
 We finally prove (\ref{s5gw4}). Denote for brevity
\begin{equation*}
G_{h}=\sum_{\Vert x-x_{0}\Vert _{\infty }\leq h}(\overline{h}^{\beta }-\Vert
x-x_{0}\Vert _{\infty }^{\beta })^{+}.
\end{equation*}%
Since $h\geq \overline{h}$ for $n$ sufficiently large, we have $\overline{M}%
_{h}\left( a\right) =\overline{M}_{\overline{h}}\left( a\right) =\sigma ^{2}$
and $G_{h}=G_{\overline{h}}.$ Then it is easy to see that%
\begin{eqnarray*}
\overline{g}(\overline{w}) &=&\sigma ^{2}\frac{\overline{M}_{\overline{h}%
}\left( a\right) +\sum_{\Vert x-x_{0}\Vert _{\infty }\leq \overline{h}%
}\left( \left( \overline{a}-L\Vert x-x_{0}\Vert _{\infty }^{\beta }\right)
^{+}\right) ^{2}}{L^{2}G_{\overline{h}}^{2}} \\
&=&\frac{\sigma ^{2}}{L}\frac{a}{G_{\overline{h}}}.
\end{eqnarray*}%
Since
\begin{align*}
G_{h}& =\sum_{\Vert x-x_{0}\Vert _{\infty }\leq \overline{h}}(\overline{h}%
^{\beta }-\Vert x-x_{0}\Vert _{\infty }^{\beta }) \\
& =\overline{h}^{\beta }\sum_{1\leq k\leq N\overline{h}}8k-\frac{8}{N^{\beta
}}\sum_{1\leq k\leq N\overline{h}}k^{\beta +1} \\
& =\frac{4\beta }{\beta +2}N^{2}\overline{h}^{\beta +2}\left( 1+o\left(
1\right) \right) \\
& =\frac{4\beta }{\left( \beta +2\right) L^{1/\beta }}N^{2}a^{\left( \beta
+2\right) /\beta }\left( 1+o\left( 1\right) \right) ,
\end{align*}%
we obtain
\begin{equation*}
\overline{g}\left( \overline{w}\right) =\sigma ^{2}\frac{\left( \beta
+2\right) }{4\beta }L^{1/\beta -1}\frac{\overline{a}^{-\frac{2}{\beta }}}{%
N^{2}}\left( 1+o\left( 1\right) \right) =c_{5}n^{-\frac{2\beta }{2\beta +2}%
}\left( 1+o\left( 1\right) \right) ,
\end{equation*}%
where $c_{4}$ is a constant depending on $\beta ,$ $L$ and $\sigma .$
\end{proof}
\\
\\
{\bf{Proof of Theorem \ref{Th oracle 001}}}. As $\rho
\left( x\right) =\left\vert f\left( x\right) -f\left( x_{0}\right)
\right\vert +\delta _{n},$
we have
\begin{equation*}
\begin{split}
\left( \sum_{x\in \mathbf{U}_{x_0,h}} w(x)\rho(x)\right)^2
&\leq
\left( \sum_{x\in \mathbf{U}_{x_0,h}} w(x) |f(x)-f(x_0)|+\delta_n\right)^2
\\&\leq
2\left(\sum_{x\in \mathbf{U}_{x_0,h}} w(x) |f(x)-f(x_0)|  \right)^2+2\delta_n^2.
\end{split}
\end{equation*}
Hence
\begin{equation*}
g_{\rho}(w)\leq 2 \overline{g}(w) + 2\delta_n^2.
\end{equation*}
So
\begin{equation*}
g_{\rho}(w_{\rho})\leq g_{\rho}(\overline{w})\leq 2\overline{g}(\overline{w})+2\delta
_{n}^{2}.
\end{equation*}
Therefore, by Lemma \ref{lm5_2} and the condition that $\delta_n=O\left( n^{-\frac{\beta }{2\beta +2}}\right) $,
we obtain
\begin{equation*}
g_{\rho}(w_{\rho})=O\left( n^{-\frac{2\beta }{2\beta +2}}\right).
\end{equation*}
This gives (\ref{s2ef2}).

\subsection{\label{Sec: proof of Th adapt 001}Proof of Theorem \protect\ref%
{Th adapt 001}}

We begin with a decomposition of $\widehat{\rho }_{x_{0}}^{\prime \prime
}(x) $. Note that
\begin{equation}
\widehat{\rho }_{x_{0}}^{\prime \prime }(x)=\left( d\left( \mathbf{Y}%
_{x,\eta }^{\prime \prime },\mathbf{Y}_{x_{0},\eta }^{\prime \prime }\right)
-\sigma \sqrt{2}\right) ^{+}\leq \left\vert d\left( \mathbf{Y}_{x,\eta
}^{\prime \prime },\mathbf{Y}_{x_{0},\eta }^{\prime \prime }\right) -\sigma
\sqrt{2}\right\vert .  \label{rho hat bound}
\end{equation}%
Recall that $M^{\prime }=\mathrm{card\ }\mathbf{U}_{x_{0},h}^{\prime
}=nh^{2}/2$, $m^{\prime \prime }=\mathrm{card\ }\mathbf{V}_{x_{0},\eta
}^{\prime \prime }=n\eta ^{2}/2.$ Let $T_{x_{0},x}$ be the translation
mapping $T_{x_{0},x}y=x+(y-x_{0}).$ Denote $\Delta _{x_{0},x}\left( y\right)
=f(y)-f(T_{x_{0},x}y)$ and $\zeta \left( y\right) =\varepsilon
(y)-\varepsilon (T_{x_{0},x}y).$ Since%
\begin{equation*}
Y(y)-Y(T_{x_{0},x}y)=\Delta _{x_{0},x}\left( y\right) +\zeta \left( y\right)
,
\end{equation*}%
it is easy to see that%
\begin{equation*}
d\left( \mathbf{Y}_{x,\eta }^{\prime \prime },\mathbf{Y}_{x_{0},\eta
}^{\prime \prime }\right) ^{2}=\frac{1}{m^{\prime \prime }}\sum_{y\in
\mathbf{V}_{x_{0},\eta }^{\prime \prime }}\left( \Delta _{x_{0},x}\left(
y\right) +\zeta \left( y\right) \right) ^{2}=\Delta ^{2}\left( x\right)
+S\left( x\right) +2\sigma ^{2},
\end{equation*}%
where%
\begin{eqnarray}
\Delta ^{2}\left( x\right) &=&\frac{1}{m^{\prime \prime }}\sum_{y\in \mathbf{%
V}_{x_{0},\eta }^{\prime \prime }}\Delta _{x_{0},x}^{2}\left( y\right) ,
\label{delta bound} \\
S\left( x\right) &=&-2S_{1}\left( x\right) +S_{2}\left( x\right)
\label{centering}
\end{eqnarray}%
with%
\begin{eqnarray*}
S_{1}(x) &=&\frac{1}{m^{\prime \prime }}\sum_{y\in \mathbf{V}_{x_{0},\eta
}^{\prime \prime }}\Delta _{x_{0},x}\left( y\right) \zeta \left( y\right) ,
\\
S_{2}(x) &=&\frac{1}{m^{\prime \prime }}\sum_{y\in \mathbf{V}_{x_{0},\eta
}^{\prime \prime }}\left( \zeta \left( y\right) ^{2}-2\sigma ^{2}\right) .
\end{eqnarray*}
Notice that $\mathbb{E}S_{1}(x)=\mathbb{E}S_{2}\left( x\right) =\mathbb{E}%
S\left( x\right) =0.$ Then obviously
\begin{eqnarray}
d\left( \mathbf{Y}_{x,\eta }^{\prime \prime },\mathbf{Y}_{x_{0},\eta
}^{\prime \prime }\right) -\sigma \sqrt{2} &=&\sqrt{\Delta ^{2}(x)+S\left(
x\right) +2\sigma ^{2}}-\sqrt{2\sigma ^{2}}  \notag \\
&=&\frac{\Delta ^{2}(x)+S(x)}{\sqrt{\Delta ^{2}(x)+S\left( x\right) +2\sigma
^{2}}+\sqrt{2\sigma ^{2}}}.  \label{d bound}
\end{eqnarray}

First we prove the following lemma.

\begin{lemma}
\label{lemma-ad2-001} Suppose that the function $f$ satisfies the local H%
\"{o}lder condition (\ref{Local Holder cond}). Then, for any $x\in \mathbf{U}%
_{x_{0},h}^{\prime },$%
\begin{equation*}
\frac{1}{3}\rho _{f,x_{0}}^{2}\left( x\right) -2L^{2}\eta ^{2\beta }\leq
\Delta ^{2}\left( x\right) \leq 3\rho _{f,x_{0}}^{2}\left( x\right)
+6L^{2}\eta ^{2\beta }.
\end{equation*}
\end{lemma}

\begin{proof}
By the decomposition%
\begin{equation*}
f\left( y\right) -f\left( T_{x_{0},x}\left( y\right) \right) =\left[ f\left(
x_{0}\right) -f\left( x\right) \right] +\left[ f\left( y\right) -f\left(
x_{0}\right) \right] +\left[ f\left( x\right) -f\left( T_{x_{0},x}\left(
y\right) \right) \right]
\end{equation*}

and the inequality $\left( a+b+c\right) ^{2}\leq 3\left(
a^{2}+b^{2}+c^{3}\right) $ we obtain%
\begin{eqnarray*}
\Delta ^{2}\left( x\right) &=&\frac{1}{m^{\prime \prime }}\sum_{y\in \mathbf{%
V}_{x_{0},\eta }^{\prime \prime }}\left( f\left( y\right) -f\left(
T_{x_{0},x}\left( y\right) \right) \right) ^{2} \\
&\leq &\frac{3}{m^{\prime \prime }}\sum_{y\in \mathbf{V}_{x_{0},\eta
}^{\prime \prime }}\left( f\left( x_{0}\right) -f\left( x\right) \right) ^{2}
\\
&&\frac{3}{m^{\prime \prime }}\sum_{y\in \mathbf{V}_{x_{0},\eta }^{\prime
\prime }}\left( f\left( y\right) -f\left( x_{0}\right) \right) ^{2} \\
&&\frac{3}{m^{\prime \prime }}\sum_{y\in \mathbf{V}_{x_{0},\eta }^{\prime
\prime }}\left( f\left( x\right) -f\left( T_{x_{0},x}\left( y\right) \right)
\right) ^{2}.
\end{eqnarray*}

By the local H\"{o}lder condition (\ref{Local Holder cond}) this implies%
\begin{equation*}
\Delta ^{2}\left( x\right) \leq 3\left( f\left( x_{0}\right) -f\left(
x\right) \right) ^{2}+3L^{2}\eta ^{2\beta }+3L^{2}\eta ^{2\beta },
\end{equation*}%
which gives the upper bound. The lower bound can be proved similarly using
the inequality $\left( a+b+c\right) ^{2}\geq \frac{1}{3}a^{2}-b^{2}-c^{2}.$
\end{proof}

We first prove a large deviation inequality for $S(x)$.

\begin{lemma}
\label{Lemma s}Let $S(x)$ be defined by (\ref{centering}). Then there are
two constants $c_{1}$ and $c_{2}$ such that for any $0\leq z\leq c_{1}\left(
m^{\prime \prime }\right) ^{1/2},$
\begin{equation*}
\mathbb{P}\left( \left\vert S(x)\right\vert \geq \frac{z}{\sqrt{m^{\prime
\prime }}}\right) \leq 2\exp \left( -c_{2}z^{2}\right) .
\end{equation*}
\end{lemma}

\begin{proof}
Denote $\xi \left( y\right) =\zeta \left( y\right) ^{2}-2\sigma ^{2}-2\Delta
_{x_{0},x}\left( y\right) \zeta \left( y\right) .$ Since $\zeta \left(
y\right) =\varepsilon (y)-\varepsilon (T_{x_{0},x}y)$ is a normal random
variable with mean $0$ and variance $2\sigma ^{2}$, the random variable $\xi
\left( y\right) $ has an exponential moment, i.e. there exist two positive
constants $t_{0}$ and $c_{3}$ depending only on $\beta ,$ $L$ and $\sigma
^{2}$ such that $\phi _{y}\left( t\right) =\mathbb{E}e^{t\xi (y)}\leq c_{3},$
for any $\left\vert t\right\vert \leq t_{0}.$ Let $\psi _{y}(t)=\ln \phi
_{y}\left( t\right) $ be the cumulate generating function. By Chebyshev's
exponential inequality we get,%
\begin{equation}
\mathbb{P}\{S(x)>z\sqrt{m^{\prime \prime }}\}\leq \exp \left\{-tz\sqrt{m^{\prime
\prime }}+\sum_{y\in \mathbf{\mathbf{V}^{\prime \prime }}_{x_{0},\eta }}\psi
_{y}(t)\right\},  \notag
\end{equation}%
for any $\left\vert t\right\vert \leq t_{0}$ and for any $z>0.$ By the-three
terms Taylor expansion, for $\left\vert t\right\vert \leq t_{0},$%
\begin{equation*}
\psi _{y}(t)=\psi _{y}(0)+t\psi _{y}^{\prime }(0)+\frac{t^{2}}{2}\psi
_{y}^{\prime \prime }{(\theta t)},
\end{equation*}%
where $\left\vert \theta \right\vert \leq 1,$ $\psi _{y}(0)=0,$ $\psi
_{y}^{\prime }(0)=\mathbb{E}\xi (y)=0$ and%
\begin{equation*}
0\leq \psi _{y}^{\prime \prime }(t)=\frac{\phi _{y}^{\prime \prime }\left(
t\right) \phi _{y}\left( t\right) -\left( \phi _{y}^{\prime }\left( t\right)
\right) ^{2}}{\left( \phi _{y}\left( t\right) \right) ^{2}}\leq \frac{\phi
_{y}^{\prime \prime }\left( t\right) }{\phi _{y}\left( t\right) }.
\end{equation*}%
Since, by Jensen's inequality $\mathbb{E}e^{t\xi (y)}\geq e^{t\mathbb{E}\xi
(y)}=1,$ we obtain the following upper bound:%
\begin{equation*}
\psi _{y}^{\prime \prime }(t)\leq \phi _{y}^{\prime \prime }\left( t\right) =%
\mathbb{E}\xi ^{2}(y)e^{t\xi (y)}.
\end{equation*}%
Using the elementary inequality $x^{2}e^{x}\leq e^{3x},$ $x\geq 0,$ we have,
for $\left\vert t\right\vert \leq t_{0}/3,$
\begin{equation*}
\psi _{y}^{\prime \prime }(t)\leq \frac{9}{t_{0}^{2}}\mathbb{E}\left( \frac{%
t_{0}}{3}\xi (y)\right) ^{2}e^{\frac{t_{0}}{3}\xi (y)}\leq \frac{9}{t_{0}^{2}%
}\mathbb{E}e^{t_{0}\xi (y)}\leq \frac{9}{t_{0}^{2}}c_{3}.
\end{equation*}%
This implies that for $\left\vert t\right\vert \leq t_{0},$%
\begin{equation*}
0\leq \psi _{y}(t)\leq \frac{9c_{3}}{2t_{0}^{2}}t^{2}
\end{equation*}%
and%
\begin{equation}
\mathbb{P}\left( S(x)>z\sqrt{m^{\prime \prime }}\right) \leq \exp \{-tz\sqrt{%
m^{\prime \prime }}+\frac{9c_{3}}{2t_{0}^{2}}m^{\prime \prime 2}\}.  \notag
\end{equation}%
If $t=c_{4}z/\sqrt{m^{\prime \prime }}\leq t_{0}/3$, where $c_{4}$ is a
positive constant, we obtain%
\begin{equation}
\mathbb{P}\left( S(x)>z\sqrt{m^{\prime \prime }}\right) \leq \exp \left\{
-c_{4}z^{2}\left( 1-\frac{9c_{3}}{2t_{0}^{2}}c_{4}\right) \right\} .  \notag
\end{equation}%
Choosing $c_{4}>0$ sufficiently small we get%
\begin{equation*}
\mathbb{P}\left( S(x)>z\sqrt{m^{\prime \prime }}\right) \leq \exp \left(
-c_{5}z^{2}\right)
\end{equation*}%
for some constant $c_{5}>0.$ In the same way we show that%
\begin{equation*}
\mathbb{P}\left( S(x)<-z\sqrt{m^{\prime \prime }}\right) \leq \exp \left(
-c_{5}z^{2}\right) .
\end{equation*}%
This proves the lemma.
\end{proof}

We next prove that $\rho _{x_{0}}^{\prime \prime }(x)$ is uniformly of order
$O\left( n^{-\frac{\beta }{2\beta +2}}\sqrt{\ln n}\right) $ with probability
$1-O\left( n^{-2}\right) ,$ if $h$ has the order $n^{-\frac{1}{2\beta +2}}.$

\begin{lemma}
\label{lm5_4} Suppose that the function $f$ satisfies the local H\"{o}lder
condition (\ref{Local Holder cond}). Assume that $h=c_{1}n^{-\frac{1}{2\beta
+2}}$ with $c_{1}>c_{0}=\left( \frac{\sigma ^{2}\left( \beta +2\right)
\left( 2\beta +2\right) }{8L^{2}\beta }\right) ^{\frac{1}{2\beta +2}}$ and
that $\eta =c_{2}n^{-\frac{1}{2\beta +2}}.$ Then there exists a constant $%
c_{3}>0$ depending only on $\beta ,$ $L$ and $\sigma $, such that
\begin{equation}
\mathbb{P}\left\{ \max_{x\in \mathbf{U}_{x_{0},h}}\widehat{\rho }%
_{x_{0}}^{\prime \prime }(x)\geq c_{3}n^{-\frac{\beta }{2\beta +2}}\sqrt{\ln
n}\right\} =O\left( n^{-2}\right) .  \label{s5pr}
\end{equation}
\end{lemma}

\begin{proof}
Using Lemma \ref{Lemma s}, there are two constants $c_{4},$ $c_{5}$ such
that, for any $z$ satisfying $0\leq z\leq c_{4}\left( m^{\prime \prime
}\right) ^{1/2},$%
\begin{eqnarray*}
\mathbb{P}\left( \max_{x\in \mathbf{U}_{x_{0},h}^{\prime }}\left\vert
S(x)\right\vert \geq \frac{z}{\sqrt{m^{\prime \prime }}}\right) &\leq
&\sum_{x\in \mathbf{U}_{x_{0},h}^{\prime }}\mathbb{P}\left( \left\vert
S(x)\right\vert \geq \frac{z}{\sqrt{m^{\prime \prime }}}\right) \\
&\leq &2m^{\prime \prime }\exp \left( -c_{5}z^{2}\right) .
\end{eqnarray*}

Recall that $m^{\prime \prime }=n\eta ^{2}/2=c_{7}n^{\frac{2\beta }{2\beta +2%
}}.$ Leting $z=\sqrt{c_{6}\log m^{\prime \prime }}$ and choosing $c_{6}$
sufficiently large we obtain
\begin{equation}
\mathbb{P}\left( \max_{x\in \mathbf{U}_{x_{0},h}^{\prime }}\left\vert
S(x)\right\vert \geq c_{8}n^{-\frac{\beta }{2\beta +2}}\sqrt{\ln n}\right)
\leq \frac{c_{9}}{n^{2}}.  \label{ineq p b}
\end{equation}
Using Lemma \ref{lemma-ad2-001} and the local H\"{o}lder condition (\ref%
{Local Holder cond}) we have $\Delta ^{2}(x)\leq cL^{2}h^{2\beta },$ for $%
x\in \mathbf{U}_{x_{0},h}^{\prime }.$ From (\ref{rho hat bound}) and (\ref{d
bound}), with probability $1-O\left( n^{-2}\right) ,$ we have%
\begin{eqnarray*}
\max_{x\in \mathbf{U}_{x_{0},h}^{\prime }}\widehat{\rho }_{x_{0}}^{\prime
\prime }(x) &\leq &\max_{x\in \mathbf{U}_{x_{0},h}^{\prime }}\frac{\Delta
^{2}(x)+\left\vert S(x)\right\vert }{\sqrt{\Delta ^{2}(x)+S\left( x\right)
+2\sigma ^{2}}+\sqrt{2\sigma ^{2}}} \\
&\leq &\frac{cL^{2}h^{2\beta }+c_{8}n^{-\frac{\beta }{2\beta +2}}\sqrt{\ln n}%
}{\sqrt{2\sigma ^{2}}}.
\end{eqnarray*}%
Since $h=O\left( n^{-\frac{1}{2\beta +2}}\right) ,$ this gives the desired
result.
\end{proof}


We then prove that given $\{Y(x),x\in \mathbf{I}_{x_{0}}^{\prime \prime }\},$
the conditional expectation of $|\widehat{f}_{h,\eta }^{\prime
}(x_{0})-f(x_{0})|$ is of order $O\left( n^{-\frac{2\beta }{2\beta +2}}\sqrt{%
\ln n}\right) $ with probability $1-O\left( n^{-2}\right) .$

\begin{lemma}
\label{s5ef}Suppose that the conditions of Theorem \ref{Th adapt 001} are
satisfied. Then
\begin{equation*}
\mathbb{P}\left( \mathbb{E}\{|\widehat{f}_{h,\eta }^{\prime
}(x_{0})-f(x_{0})|^{2}\,\,\big|\,\,Y(x),x\in \mathbf{I}_{x_{0}}^{\prime
\prime }\}\geq cn^{-\frac{2\beta }{2\beta +2}}\ln n\right) =O(n^{-2}),
\end{equation*}%
where $c>0$ is a constant depending only on $\beta ,$ $L$ and $\sigma .$
\end{lemma}


\begin{proof}
By (\ref{OWFilter prim}) and the independence of $\varepsilon (x)$, we have
\begin{equation}
\begin{split}
& \mathbb{E}\{|\widehat{f}_{h,\eta }^{\prime }(x_{0})-f(x_{0})|^{2}\,\,\big|%
\,\,Y(x),x\in \mathbf{I}_{x_{0}}^{\prime \prime }\} \\
\leq \,\,& \left( \sum_{x\in \mathbf{U}_{x_{0},h}^{\prime }}\widehat{w}%
^{\prime \prime }(x)\rho _{f,x_{0}}(x)\right) ^{2}+\sigma ^{2}\sum_{x\in
\mathbf{U}_{x_{0},h}^{\prime }}\widehat{w}^{\prime \prime 2}(x).
\end{split}
\label{eq Lemma experance}
\end{equation}%
Since $\rho _{f,x_{0}}(x)<Lh^{\beta }$, from (\ref{eq Lemma experance}) we
get 
\begin{align}
& \mathbb{E}\{|\widehat{f}_{h,\eta }^{\prime }(x_{0})-f(x_{0})|^{2}\,\,\big|%
\,\,Y(x),x\in \mathbf{I}_{x_{0}}^{\prime \prime }\}  \notag \\*
\leq \,\,& \left( \sum_{x\in \mathbf{U}_{x_{0},h}^{\prime }}\widehat{w}%
^{\prime \prime \beta }\right) ^{2}+\sigma ^{2}\sum_{x\in \mathbf{U}%
_{x_{0},h}^{\prime }}\widehat{w}^{\prime \prime 2}(x)  \notag \\*
\leq \,\,& L^{2}h^{2\beta }+\sigma ^{2}\sum_{x\in \mathbf{U}%
_{x_{0},h}^{\prime }}\widehat{w}^{\prime \prime 2}(x)  \notag \\*
\leq \,\,& \left( \sum_{x\in \mathbf{U}_{x_{0},h}^{\prime }}\widehat{w}%
^{\prime \prime }(x)\widehat{\rho }_{x_{0}}^{\prime \prime }(x)\right)
^{2}+\sigma ^{2}\sum_{x\in \mathbf{U}_{x_{0},h}^{\prime }}\widehat{w}%
^{\prime \prime 2}(x)+L^{2}h^{2\beta }.  \label{espect mse}
\end{align}

Let $w_{1}^{\ast }=\arg \min\limits_{w}g_{1}(w)$, where
\begin{equation}
g_{1}(w)=\left( \sum_{x\in \mathbf{U}_{x_{0},h}^{\prime }}{w}(x)\rho
_{f,x_{0}}(x)\right) ^{2}+\sigma ^{2}\sum_{x\in \mathbf{U}_{x_{0},h}^{\prime
}}{w}^{2}(x).  \label{upper bound g1}
\end{equation}%
As $\widehat{w}^{\prime \prime }$ minimizes the function in (\ref{s3ww}),
from (\ref{espect mse}) we obtain
\begin{equation}
\begin{split}
& \mathbb{E}\{|\widehat{f}_{h,\eta }^{\prime }(x_{0})-f(x_{0})|^{2}\,\,\big|%
\,\,Y(x),x\in \mathbf{I}_{x_{0}}^{\prime \prime }\} \\
\leq \,\,& \left( \sum_{x\in \mathbf{U}_{x_{0},h}^{\prime }}{w}_{1}^{\ast
}(x)\widehat{\rho }_{x_{0}}^{\prime \prime }(x)\right) ^{2}+\sigma
^{2}\sum_{x\in \mathbf{U}_{x_{0},h}^{\prime }}{w}_{1}^{\ast
2}(x)+L^{2}h^{2\beta }.
\end{split}
\label{s5ef4}
\end{equation}%
By Lemma \ref{lm5_4}, with probability $1-O(n^{-2})$ we have
\begin{equation*}
\sum_{x\in \mathbf{U}_{x_{0},h}^{\prime }}{w}_{1}^{\ast }(x)\widehat{\rho }%
_{x_{0}}^{\prime \prime }(x)\leq c_{1}n^{-\frac{\beta }{2\beta +2}}\sqrt{\ln
n}.
\end{equation*}%
Therefore by (\ref{s5ef4}), with probability $1-O(n^{-2})$,
\begin{equation*}
\begin{split}
& \mathbb{E}\{|\widehat{f}_{h,\eta }^{\prime }(x_{0})-f(x_{0})|^{2}\,\,\big|%
\,\,Y(x),x\in \mathbf{I}_{x_{0}}^{\prime \prime }\} \\
\leq \,\,& \sigma ^{2}\sum_{x\in \mathbf{U}_{x_{0},h}^{\prime }}{w}%
_{1}^{\ast 2}(x)+c_{1}^{2}n^{-\frac{2\beta }{2\beta +2}}\ln n+L^{2}h^{2\beta
} \\
\leq \,\,& g_{1}(w_{1}^{\ast })+c_{1}^{2}n^{-\frac{2\beta }{2\beta +2}}\ln
n+L^{2}h^{2\beta }.
\end{split}%
\end{equation*}%
This gives the assertion of Lemma \ref{lm5_6}, as $h^{2\beta }=O\left( n^{-%
\frac{2\beta }{2\beta +2}}\right) $ and $g_{1}(w_{1}^{\ast })=O\left( n^{-%
\frac{2\beta }{2\beta +2}}\right) ,$ by Lemma \ref{lm5_2} with $\mathbf{U}%
_{x_{0},h}^{\prime }$ instead of $\mathbf{U}_{x_{0},h}.$
\end{proof}

Now we are ready to prove Theorem \ref{Th adapt 001}.\newline
\textit{Proof of Theorem \ref{Th adapt 001}.} Since the function $f$
satisfies H\"{o}lder's condition, by the definition of $g_{1}(w)$ (cf. (\ref%
{upper bound g1})) we have%
\begin{equation*}
\begin{split}
g_{1}(w)& \leq \left( \sum_{x\in \mathbf{U}_{x_{0},h}^{\prime
}}w(x)Lh^{\beta }\right) ^{2}+\sigma ^{2}\sum_{x\in \mathbf{U}%
_{x_{0},h}^{\prime }}w^{2}(x) \\
& \leq L^{2}h^{2\beta }+\sigma ^{2}\leq L^{2}+\sigma ^{2},
\end{split}%
\end{equation*}%
so that%
\begin{equation*}
\mathbb{E}\left( |\widehat{f}_{h,\eta }^{\prime }(x_{0})-f(x_{0})|^{2}\,\,%
\big|\,\,Y(x),x\in \mathbf{I}_{x_{0}}^{\prime \prime }\right) \leq g_{1}(%
\widehat{w}^{\prime \prime })\leq L^{2}+\sigma ^{2}.
\end{equation*}%
Denote by $X$ the conditional expectation in the above display and write $%
\mathbbm{1}\{\cdot \}$ for the indicator function of the set $\{\cdot \}$.
Then%
\begin{equation*}
\begin{split}
\mathbb{E}X& =\mathbb{E}X\cdot \mathbbm{1}{\{X\geq c n^{-\frac{2\beta
}{2\beta +2}}\ln n\}}+\mathbb{E}X\cdot \mathbbm{1}{\{X< c n^{-\frac{2\beta
}{2\beta +2}}\ln n\}} \\
& \leq \left( L^{2}+\sigma ^{2}\right) \mathbb{P}\{X\geq cn^{-\frac{2\beta }{%
2\beta +2}}\ln n\}+cn^{-\frac{2\beta }{2\beta +2}}\ln n.
\end{split}%
\end{equation*}%
So applying Lemma \ref{s5ef}, we see that
\begin{equation*}
\begin{split}
\mathbb{E}\left( |\widehat{f}_{h,\eta }^{\prime
}(x_{0})-f(x_{0})|^{2}\right) & =\mathbb{E}X \\
& \leq O(n^{-2})+cn^{-\frac{2\beta }{2\beta +2}}\ln n \\
& =O\left( n^{-\frac{2\beta }{2\beta +2}}\ln n\right) .
\end{split}%
\end{equation*}%
This proves Theorem \ref{Th adapt 001}.

\subsection{\label{Sec: proof of Th adapt 002}Proof of Theorem \protect\ref%
{Th adapt 002}}

We keep the notations of the prevoius subsection. The following result gives
a two sided bound for $\widehat{\rho }_{x_{0}}^{\prime \prime }(x).$

\begin{lemma}
\label{lemma-ad2-002} Suppose that the function $f$ satisfies the local H%
\"{o}lder condition (\ref{Local Holder cond}). Assume that $%
h=c_{1}n^{-\alpha }$ with $c_{1}>0$ and $\alpha <\frac{1}{2\beta +2}$ and
that $\eta =c_{2}n^{-\frac{1}{2\beta +2}}.$ Then there exists positive
constants $c_{3},$ $c_{4},$ $c_{5}$ and $c_{6}$ depending only on $\beta ,$ $%
L$ and $\sigma $, such that
\begin{equation}
\mathbb{P}\left\{ \max_{x\in \mathbf{U}_{x_{0},h}}\left( \widehat{\rho }%
_{x_{0}}^{\prime \prime }(x)-c_{3}\rho _{f,x_{0}}^{2}\left( x\right) \right)
\leq c_{4}n^{-\frac{\beta }{2\beta +2}}\sqrt{\ln n}\right\} =1-O\left(
n^{-2}\right)   \label{eq ad2-001}
\end{equation}%
and%
\begin{equation}
\mathbb{P}\left\{ \max_{x\in \mathbf{U}_{x_{0},h}}\left( \rho
_{f,x_{0}}^{2}\left( x\right) -c_{5}\widehat{\rho }_{x_{0}}^{\prime \prime
}(x)\right) \leq c_{6}n^{-\frac{\beta }{2\beta +2}}\sqrt{\ln n}\right\}
=1-O\left( n^{-2}\right) .  \label{eq ad2-002}
\end{equation}
\end{lemma}

\begin{proof}
As in the proof of Lemma \ref{lm5_4}, we have
\begin{equation*}
\mathbb{P}\left( \max_{x\in \mathbf{U}_{x_{0},h}^{\prime }}\left\vert
S(x)\right\vert \geq c_{7}n^{-\frac{\beta }{2\beta +2}}\sqrt{\ln n}\right)
\leq \frac{c_{8}}{n^{2}}.
\end{equation*}%
Using Lemma \ref{lemma-ad2-001}, for any $x\in \mathbf{U}_{x_{0},h}^{\prime
},$%
\begin{equation}
\frac{1}{3}\rho _{f,x_{0}}^{2}\left( x\right) -2L^{2}\eta ^{2\beta }\leq
\Delta ^{2}\left( x\right) \leq 3\rho _{f,x_{0}}^{2}\left( x\right)
+6L^{2}\eta ^{2\beta }.  \label{eqadxxx002}
\end{equation}%
From (\ref{rho hat bound}) we have%
\begin{equation*}
d\left( \mathbf{Y}_{x,\eta }^{\prime \prime },\mathbf{Y}_{x_{0},\eta
}^{\prime \prime }\right) -\sigma \sqrt{2}=\frac{\Delta ^{2}(x)+S(x)}{\sqrt{%
\Delta ^{2}(x)+S\left( x\right) +2\sigma ^{2}}+\sqrt{2\sigma ^{2}}}.
\end{equation*}%
For the upper bound we have, for any $x\in \mathbf{U}_{x_{0},h}^{\prime },$%
\begin{equation*}
\widehat{\rho }_{x_{0}}^{\prime \prime }(x)=\left( d\left( \mathbf{Y}%
_{x,\eta }^{\prime \prime },\mathbf{Y}_{x_{0},\eta }^{\prime \prime }\right)
-\sigma \sqrt{2}\right) ^{+}\leq \frac{3\rho _{f,x_{0}}^{2}\left( x\right)
+6L^{2}\eta ^{2\beta }+\left\vert S(x)\right\vert }{\sqrt{2\sigma ^{2}}}
\end{equation*}%
Therefore, with probability $1-O\left( n^{-2}\right) ,$%
\begin{eqnarray*}
\max_{x\in \mathbf{U}_{x_{0},h}^{\prime }}\left( \widehat{\rho }%
_{x_{0}}^{\prime \prime }(x)-\frac{3\rho _{f,x_{0}}^{2}\left( x\right) }{%
\sqrt{2\sigma ^{2}}}\right)  &\leq &\frac{6L^{2}\eta ^{2\beta }+c_{7}n^{-%
\frac{\beta }{2\beta +2}}\sqrt{\ln n}}{\sqrt{2\sigma ^{2}}} \\
&\leq &c_{8}n^{-\frac{\beta }{2\beta +2}}\sqrt{\ln n}.
\end{eqnarray*}

For the lower bound, we have, for any $x\in \mathbf{U}_{x_{0},h}^{\prime },$
we have%
\begin{eqnarray*}
\widehat{\rho }_{x_{0}}^{\prime \prime }(x) &=&\left( d\left( \mathbf{Y}%
_{x,\eta }^{\prime \prime },\mathbf{Y}_{x_{0},\eta }^{\prime \prime }\right)
-\sigma \sqrt{2}\right) ^{+} \\
&=&\frac{\left( \Delta ^{2}(x)+S(x)\right) ^{+}}{\sqrt{\Delta
^{2}(x)+S\left( x\right) +2\sigma ^{2}}+\sqrt{2\sigma ^{2}}} \\
&\geq &\frac{\left( \Delta ^{2}(x)+S(x)\right) ^{+}}{\sqrt{\Delta
^{2}(x)+c_{7}n^{-\frac{\beta }{2\beta +2}}\sqrt{\ln n}+2\sigma ^{2}}+\sqrt{%
2\sigma ^{2}}} \\
&\geq &c_{9}\left( \Delta ^{2}(x)+S(x)\right) ^{+} \\
&\geq &c_{9}\left( \Delta ^{2}(x)-\left\vert S(x)\right\vert \right) .
\end{eqnarray*}%
Taking into account (\ref{eqadxxx002}), on the set $\left\{ \max_{x\in
\mathbf{U}_{x_{0},h}^{\prime }}\left\vert S(x)\right\vert <c_{7}n^{-\frac{%
\beta }{2\beta +2}}\sqrt{\ln n}\right\} ,$%
\begin{eqnarray*}
\widehat{\rho }_{x_{0}}^{\prime \prime }(x) &\geq &c_{9}\left( \frac{1}{3}%
\rho _{f,x_{0}}^{2}\left( x\right) -2L^{2}\eta ^{2\beta }-\left\vert
S(x)\right\vert \right)  \\
&\geq &c_{10}\left( \rho _{f,x_{0}}^{2}\left( x\right) -\eta ^{2\beta }-n^{-%
\frac{\beta }{2\beta +2}}\sqrt{\ln n}\right)
\end{eqnarray*}

Therefore, with probability $1-O\left( n^{-2}\right) ,$%
\begin{eqnarray*}
\max_{x\in \mathbf{U}_{x_{0},h}^{\prime }}\left( c_{10}\rho _{f,x_{0}}^{2}\left(
x\right) -\widehat{\rho }_{x_{0}}^{\prime \prime }(x)\right)  &\leq
&c_{10}\left( \eta ^{2\beta }+n^{-\frac{\beta }{2\beta +2}}\sqrt{\ln n}%
\right)  \\
&\leq &c_{11}n^{-\frac{\beta }{2\beta +2}}\sqrt{\ln n}.
\end{eqnarray*}%
So the lemma is proved.
\end{proof}


We then prove that given $\{Y(x),x\in \mathbf{I}_{x_{0}}^{\prime \prime }\}$%
, the conditional expectation of $|\widehat{f}_{h,\eta }^{\prime
}(x_{0})-f(x_{0})|$ is of order $O\left( n^{-\frac{\beta }{2\beta +2}}\sqrt{%
\ln n}\right) $ with probability $1-O\left( n^{-2}\right) $.

\begin{lemma}
\label{lemma-ad-003}Suppose that the conditions of Theorem \ref{Th adapt 002}
are satisfied. Then
\begin{equation*}
\mathbb{P}\left( \mathbb{E}\{|\widehat{f}_{h,\eta }^{\prime
}(x_{0})-f(x_{0})|^{2}\,\,\big|\,\,Y(x),x\in \mathbf{I}_{x_{0}}^{\prime
\prime }\}\geq cn^{-\frac{\beta }{2\beta +2}}\ln n\right) =O(n^{-2}),
\end{equation*}%
where $c>0$ is a constant depending only on $\beta $, $L$ and $\sigma $. %
\label{lm5_6}
\end{lemma}


\begin{proof}
By (\ref{OWFilter prim}) and the independence of $\varepsilon (x)$, we have
\begin{equation*}
\mathbb{E}\{|\widehat{f}_{h,\eta }^{\prime }(x_{0})-f(x_{0})|^{2}\,\,\big|%
\,\,Y(x),x\in \mathbf{I}_{x_{0}}^{\prime \prime }\} \\
\leq \,\, \left( \sum_{x\in \mathbf{U}_{x_{0},h}^{\prime }}\widehat{w}%
^{\prime \prime }(x)\rho _{f,x_{0}}(x)\right) ^{2}+\sigma ^{2}\sum_{x\in
\mathbf{U}_{x_{0},h}^{\prime }}\widehat{w}^{\prime \prime 2}(x).
\end{equation*}%
Since, by Lemma \ref{lemma-ad2-002}, with probability $1-O\left(
n^{-2}\right) ,$

\begin{equation*}
\max_{x\in \mathbf{U}_{x_{0},h}}\left( \rho _{f,x_{0}}^{2}\left( x\right)
-c_{1}\widehat{\rho }_{x_{0}}^{\prime \prime }(x)\right) \leq c_{2}n^{-\frac{%
\beta }{2\beta +2}}\sqrt{\ln n},
\end{equation*}
we get 
(with probability $1-O\left( n^{-2}\right) $),%
\begin{align}
& \mathbb{E}\{|\widehat{f}_{h,\eta }^{\prime }(x_{0})-f(x_{0})|^{2}\,\,\big|%
\,\,Y(x),x\in \mathbf{I}_{x_{0}}^{\prime \prime }\}  \notag \\
& \leq c_{3}\,\left( \sum_{x\in \mathbf{U}_{x_{0},h}^{\prime }}\widehat{w}%
^{\prime \prime }(x)\sqrt{\widehat{\rho }_{x_{0}}^{\prime \prime }(x)}%
\right) ^{2}+c_{2}n^{-\frac{\beta }{2\beta +2}}\sqrt{\ln n}+\sigma
^{2}\sum_{x\in \mathbf{U}_{x_{0},h}^{\prime }}\widehat{w}^{\prime \prime
2}(x).  \label{expect2}
\end{align}%
A simple truncation argument, using the decomposition
\begin{eqnarray*}
\widehat{\rho }_{x_{0}}^{\prime \prime }(x) &=&\widehat{\rho }%
_{x_{0}}^{\prime \prime }(x)\mathbbm{1}\left\{ \widehat{\rho }%
_{x_{0}}^{\prime \prime }(x)\leq n^{-\frac{\beta }{2\beta +2}}\right\} \\
&&+\widehat{\rho }_{x_{0}}^{\prime \prime }(x)\mathbbm{1}\left\{ \widehat{%
\rho }_{x_{0}}^{\prime \prime }(x)>n^{-\frac{\beta }{2\beta +2}}\right\} ,
\end{eqnarray*}%
gives
\begin{eqnarray}
\sum_{x\in \mathbf{U}_{x_{0},h}^{\prime }}\widehat{w}^{\prime \prime }(x)%
\sqrt{\widehat{\rho }_{x_{0}}^{\prime \prime }(x)} &\leq &n^{-\frac{1}{2}%
\frac{\beta }{2\beta +2}}\sum_{x\in \mathbf{U}_{x_{0},h}^{\prime }}\widehat{w%
}^{\prime \prime }(x)+n^{\frac{1}{2}\frac{\beta }{2\beta +2}}\sum_{x\in
\mathbf{U}_{x_{0},h}^{\prime }}\widehat{w}^{\prime \prime }(x)\widehat{\rho }%
_{x_{0}}^{\prime \prime }(x)  \notag \\
&\leq &n^{-\frac{1}{2}\frac{\beta }{2\beta +2}}+n^{\frac{1}{2}\frac{\beta }{%
2\beta +2}}\sum_{x\in \mathbf{U}_{x_{0},h}^{\prime }}\widehat{w}^{\prime
\prime }(x)\widehat{\rho }_{x_{0}}^{\prime \prime }(x).  \label{expect3}
\end{eqnarray}

From (\ref{expect2}) and (\ref{expect3}) one gets%
\begin{eqnarray*}
&&\mathbb{E}\{|\widehat{f}_{h,\eta }^{\prime }(x_{0})-f(x_{0})|^{2}\,\,\big|%
\,\,Y(x),x\in \mathbf{I}_{x_{0}}^{\prime \prime }\} \\
&\leq &c_{4}n^{\frac{\beta }{2\beta +2}}\,\left( \left( \sum_{x\in \mathbf{U}%
_{x_{0},h}^{\prime }}\widehat{w}^{\prime \prime }(x)\widehat{\rho }%
_{x_{0}}^{\prime \prime }(x)\right) ^{2}+\sigma ^{2}\sum_{x\in \mathbf{U}%
_{x_{0},h}^{\prime }}\widehat{w}^{\prime \prime 2}(x)\right) +c_{5}n^{-\frac{%
\beta }{2\beta +2}}\sqrt{\ln n}.
\end{eqnarray*}

Let $w_{1}^{\ast }=\arg \min\limits_{w}g_{1}(w)$, where $g_{1}$ was defined
in (\ref{s5ef4}). As $\widehat{w}^{\prime \prime }$ minimize the function in
(\ref{s3ww}), from (\ref{espect mse}) we obtain
\begin{equation}
\begin{split}
& \mathbb{E}\{|\widehat{f}_{h,\eta }^{\prime }(x_{0})-f(x_{0})|^{2}\,\,\big|%
\,\,Y(x),x\in \mathbf{I}_{x_{0}}^{\prime \prime }\} \\
& \leq \,\,c_{4}n^{\frac{\beta }{2\beta +2}}\,\left( \left( \sum_{x\in \mathbf{U}%
_{x_{0},h}^{\prime }}{w}_{1}^{\ast }(x)\widehat{\rho }_{x_{0}}^{\prime
\prime }(x)\right) ^{2}+\sigma ^{2}\sum_{x\in \mathbf{U}_{x_{0},h}^{\prime }}%
{w}_{1}^{\ast 2}(x)\right) +c_{5}n^{-\frac{\beta }{2\beta +2}}\sqrt{\ln n}.
\end{split}
\label{s5ef4bis}
\end{equation}%
By Lemma \ref{lemma-ad2-002}, with probability $1-O\left( n^{-2}\right) ,$
\begin{equation*}
\max_{x\in \mathbf{U}_{x_{0},h}}\left( \widehat{\rho }_{x_{0}}^{\prime
\prime }(x)-c_{6}\rho _{f,x_{0}}^{2}\left( x\right) \right) \leq c_{7}n^{-%
\frac{\beta }{2\beta +2}}\sqrt{\ln n}.
\end{equation*}%
Therefore, with probability $1-O(n^{-2}),$
\begin{equation*}
\begin{split}
& \mathbb{E}\{|\widehat{f}_{h,\eta }^{\prime }(x_{0})-f(x_{0})|^{2}\,\,\big|%
\,\,Y(x),x\in \mathbf{I}_{x_{0}}^{\prime \prime }\} \\
& \leq c_{8}n^{\frac{\beta }{2\beta +2}}\left( \left( \sum_{x\in \mathbf{U}%
_{x_{0},h}^{\prime }}w_{1}^{\ast }(x)\rho _{f,x_{0}}(x)\right) ^{2}+\sigma
^{2}\sum_{x\in \mathbf{U}_{x_{0},h}^{\prime }}w_{1}^{\ast 2}(x)\right) +c_{9}n^{-%
\frac{\beta }{2\beta +2}}\sqrt{\ln n} \\
& =c_{8}n^{\frac{\beta }{2\beta +2}}g_{1}(w_{1}^{\ast })+c_{9}n^{-\frac{\beta }{%
2\beta +2}}\sqrt{\ln n}.
\end{split}%
\end{equation*}%
This gives the assertion of Lemma \ref{lm5_6}, as $g_{1}(w_{1}^{\ast
})=O\left( n^{-\frac{2\beta }{2\beta +2}}\right) $ by Lemma \ref{lm5_2} with
$\mathbf{U}_{x_{0},h}^{\prime }$ instead of $\mathbf{U}_{x_{0},h}$.
\end{proof}

\textit{Proof of Theorem \ref{Th adapt 002}}. Since the function $f$
satisfies H\"{o}lder's condition, by the definition of $g_{1}(w)$ (cf. (\ref%
{upper bound g1})) we have (see the proof of Theorem \ref{Th adapt 001})%
\begin{equation*}
g_{1}(w)\leq L^{2}+\sigma ^{2}
\end{equation*}%
so that
\begin{equation*}
\mathbb{E}\left( |\widehat{f}_{h,\eta }^{\prime }(x_{0})-f(x_{0})|^{2}\,\,%
\big|\,\,Y(x),x\in \mathbf{I}_{x_{0}}^{\prime \prime }\right) \leq g_{1}(%
\widehat{w}^{\prime \prime })\leq L^{2}+\sigma ^{2}.
\end{equation*}%
Denote by $X$ the conditional expectation in the above display. Then
\begin{equation*}
\begin{split}
\mathbb{E}X& =\mathbb{E}X\cdot \mathbbm{1}{\{X\geq cn^{-\frac{\beta
}{2\beta +2}}\ln n\}}+\mathbb{E}X\cdot \mathbbm{1}{\{X<
cn^{-\frac{\beta }{2\beta +2}}\ln n\}} \\
& \leq \left( L^{2}+\sigma ^{2}\right) \mathbb{P}\{X\geq cn^{-\frac{\beta }{%
2\beta +2}}\ln n\}+cn^{-\frac{\beta }{2\beta +2}}\ln n.
\end{split}%
\end{equation*}%
So applying Lemma \ref{lemma-ad-003}, we see that
\begin{equation*}
\begin{split}
\mathbb{E}\left( |\widehat{f}_{h,\eta }^{\prime
}(x_{0})-f(x_{0})|^{2}\right) & =\mathbb{E}X \\
& \leq O(n^{-2})+cn^{-\frac{\beta }{2\beta +2}}\ln n \\
& =O\left( n^{-\frac{\beta }{2\beta +2}}\ln n\right) .
\end{split}%
\end{equation*}%
This proves Theorem \ref{Th adapt 002}

\section{Conclusion}

A new image denoising filter to deal with the additive Gaussian white noise
model based on a weights optimization problem is proposed. The proposed
algorithm is computationally fast and its implementation is straightforward.
Our work leads to the following conclusions.

\begin{enumerate}
\item In the non-local means filter the choice of the Gaussian kernel is not
justified. Our approach shows that it is preferable to choose the triangular
kernel.

\item The obtained estimator is shown to converge at the usual optimal rate,
under some regularity conditions on the target function. To the best of our
knowledge such convergence results have not been established so far.

\item Our filter is parameter free in the sense that it chooses
automatically the bandwidth parameter.

\item Our numerical results confirm that optimal choice of the kernel
improves the performance of the non-local means filter, under the same
conditions.
\end{enumerate}


\begin{thebibliography}{99}
\bibitem{Bu} A.~Buades, B.~Coll, and J.M. Morel.
\newblock {A review of
image denoising algorithms, with a new one}.
\newblock {\em Multiscale
Model. Simul.}, 4(2):490--530, 2006.

\bibitem{buades2009note} T.~Buades, Y.~Lou, JM~Morel, and Z.~Tang. %
\newblock{A note on multi-image denoising}. \newblock In \emph{Int. workshop
on Local and Non-Local Approximation in Image Processing}, pages 1--15,
August 2009.

\bibitem{cai2008two} J.F. Cai, R.H. Chan, and M.~Nikolova.
\newblock {Two-phase approach for deblurring images corrupted by impulse plus
  Gaussian noise}. \newblock {\em Inverse Problems and Imaging},
2(2):187--204, 2008.

\bibitem{Dabov2007color} K.~Dabov, A.~Foi, V.~Katkovnik, and K.~Egiazarian.
\newblock {Color image denoising via sparse 3D collaborative filtering with
  grouping constraint in luminance-chrominance space}. \newblock In \emph{%
IEEE Int. Conf. Image Process.}, volume~1, pages 313--316, September 2007.

\bibitem{Dabov2007image} K.~Dabov, A.~Foi, V.~Katkovnik, and K.~Egiazarian.
\newblock {Image denoising by sparse 3-D transform-domain collaborative
  filtering}. \newblock {\em IEEE Trans. Image Process.}, 16(8):2080--2095,
2007.

\bibitem{donoho1994ideal} D.L. Donoho and J.M. Johnstone.
\newblock {Ideal
spatial adaptation by wavelet shrinkage}. \newblock {\em Biometrika},
81(3):425, 1994.

\bibitem{FanGijbels1996} J.Q. Fan and I.~Gijbels.
\newblock {Local
polynomial modelling and its applications}. \newblock In \emph{Chapman} \&
\emph{Hall, London}, 1996.

\bibitem{Garnett2005universal} R.~Garnett, T.~Huegerich, C.~Chui, and W.~He. %
\newblock {A universal noise removal algorithm with an impulse detector}. %
\newblock {\em IEEE Trans. Image Process.}, 14(11):1747--1754, 2005.

\bibitem{Katkovnik2010local} V.~Katkovnik, A.~Foi, K.~Egiazarian, and
J.~Astola.
\newblock {From local kernel to nonlocal multiple-model image
denoising}. \newblock {\em Int. J. Comput. Vis.}, 86(1):1--32, 2010.

\bibitem{kervrann2006optimal} C.~Kervrann and J.~Boulanger. %
\newblock{Optimal spatial adaptation for patch-based image denoising}.
\newblock {\em
IEEE Trans. Image Process.}, 15(10):2866--2878, 2006.

\bibitem{kervrann2008local} C.~Kervrann and J.~Boulanger.
\newblock {Local adaptivity to variable smoothness for exemplar-based image
  regularization and representation}. \newblock {\em Int. J. Comput. Vis.},
79(1):45--69, 2008.

\bibitem{kervrann2010bayesian} C.~Kervrann, J.~Boulanger, and P.~Coup{\'{e}}%
.
\newblock {Bayesian non-local means filter, image redundancy and adaptive
  dictionaries for noise removal}. \newblock In \emph{Proc. Conf.
Scale-Space and Variational Meth. (SSVM' 07)}, pages 520--532. Springer,
June 2007.

\bibitem{li2010new} B.~Li, Q.S. Liu, J.W. Xu, and X.J. Luo.
\newblock {A new
method for removing mixed noises}.
\newblock {\em Sci. China Ser. F
(Information sciences)}, 54:1--9, 2010.

\bibitem{lou2010image} Y.~Lou, X.~Zhang, S.~Osher, and A.~Bertozzi. %
\newblock {Image recovery via nonlocal operators}.
\newblock {\em J. Sci.
Comput.}, 42(2):185--197, 2010.

\bibitem{Nazin2008direct} A.V. Nazin, J.~Roll, L.~Ljung, and I.~Grama.
\newblock {Direct weight optimization in statistical estimation and system
  identification}.
\newblock {\em System Identification and Control Problems (SICPRO��08),
  Moscow}, January 28-31 2008.

\bibitem{polzehl2003image} J.~Polzehl and V.~Spokoiny.
\newblock {Image
denoising: pointwise adaptive approach}. \newblock {\em Ann. Stat.},
31(1):30--57, 2003.

\bibitem{polzehl2006propagation} J.~Polzehl and V.~Spokoiny. %
\newblock{Propagation-separation approach for local likelihood estimation}. %
\newblock\emph{Probab. Theory Rel.}, 135(3):335--362, 2006.

\bibitem{polzehl2000adaptive} J.~Polzehl and V.G. Spokoiny. %
\newblock{Adaptive weights smoothing with applications to image restoration}%
. \newblock {\em J. Roy. Stat. Soc. B}, 62(2):335--354, 2000.

\bibitem{Roll2003local} J.~Roll.
\newblock {Local and piecewise affine
approaches to system identification}.
\newblock {\em Ph.D. dissertation, Dept. Elect. Eng., Link\"{o}ing University,
  Link\"{o}ing, Sweden}, 2003.

\bibitem{Roll2004} J.~Roll and L.~Ljung.
\newblock {Extending the direct
weight optimization approach}. \newblock In \emph{Technical Report
LiTH-ISY-R-2601. Dept. of EE, Link\"{o}ping Univ., Sweden}, 2004.

\bibitem{roll2005nonlinear} J.~Roll, A.~Nazin, and L.~Ljung. %
\newblock{Nonlinear system identification via direct weight optimization}. %
\newblock\emph{Automatica}, 41(3):475--490, 2005.

\bibitem{Sacks1978linear} J.~Sacks and D.~Ylvisaker.
\newblock {Linear
estimation for approximately linear models}. \newblock {\em Ann. Stat.},
6(5):1122--1137, 1978.

\bibitem{tomasi1998bilateral} C.~Tomasi and R.~Manduchi.
\newblock{Bilateral
filtering for gray and color images}. \newblock In \emph{Proc. Int. Conf.
Computer Vision}, pages 839--846, January 04-07 1998.

\bibitem{Wh} P.~Whittle.
\newblock {Optimization under constraints: theory and applications of nonlinear
  programming}. \newblock In \emph{Wiley-Interscience, New York}, 1971.

\bibitem{yaroslavsky1985digital} L.~P. Yaroslavsky.
\newblock {Digital
picture processing. An introduction}. \newblock In \emph{Springer-Verlag,
Berlin}, 1985.
\end{thebibliography}
\end{document}